\documentclass[11pt]{article}
\usepackage[margin=1in]{geometry} 
\usepackage[dvipsnames]{xcolor}
\usepackage[colorlinks=true,       
linkcolor=black,        
citecolor=blue,        
filecolor=black,     
urlcolor=black         
]{hyperref}

\usepackage[authoryear,round]{natbib}
\setcitestyle{citesep={;}}

\usepackage{url} 

\usepackage{amsthm,amsmath,amssymb,amsfonts} 

\usepackage[ruled, linesnumbered, vlined]{algorithm2e}
\SetKwInput{KwInput}{Input} 

\usepackage[font=small]{caption}

\usepackage[shortlabels]{enumitem}

\usepackage[title,toc]{appendix}

\usepackage[labelformat=simple]{subcaption}

\usepackage{graphicx}

\usepackage[export]{adjustbox}

\usepackage{booktabs} 
\usepackage{multirow} 
\usepackage{bbm}
\usepackage{authblk}

\newtheorem{theorem}{Theorem}
\newtheorem{proposition}{Proposition}
\newtheorem{lemma}{Lemma}

\usepackage{xpatch}
\xpatchcmd{\proof}{\itshape}{\normalfont\proofnamefont}{}{}
\newcommand{\proofnamefont}{\bfseries}

\newcommand{\RS}{\mathbb{R}}
\newcommand{\PS}{\mathcal{P}}
\newcommand{\QS}{\mathcal{Q}}

\newcommand{\Y}{Y}
\newcommand{\dbeta}{d_{S_0,r}}

\newcommand{\T}{ \top }
\def\qt#1{\qquad\text{#1}}

\def\argmax{\mathop{\rm argmax}}
\newcommand{\diff}{d} 

\newcommand{\ball}{\mathrm{B}}

\newcommand{\lipsz}{\mathfrak{L}}
\newcommand{\Log}{\mathrm{Log}~}

\newcommand{\fv}{\mathrm{f}}

\newcommand{\gv}{\mathrm{g}}

\newcommand{\dash}{-}
\newcommand{\pb}{\mathfrak{p}} 
\newcommand{\Expt}{\mathbb{E}} 
\newcommand{\E}{\mathbb{E}}
\newcommand{\PP}{\mathbb{P}}

\newcommand{\EB}{\mathrm{EB}}
\newcommand{\OB}{\mathrm{OB}}

\newcommand{\dlp}{\mathfrak{d}_{\mathrm{LP}}}

\usepackage{xcolor}

\newenvironment{revise}[0]{\color{black}}{} 

\begin{document}


\title{A Nonparametric Maximum Likelihood Approach to Mixture of Regression}
\author[$\ast$]{Hansheng Jiang}
\author[$\dagger$]{Adityanand Guntuboyina}
\date{}
\affil[$\ast$]{Rotman School of Management, University of Toronto}
\affil[$\dagger$]{Department of Statistics, University of California, Berkeley}
\maketitle
\renewcommand{\thefootnote}{\fnsymbol{footnote}}
\footnotetext[1]{\texttt{hansheng.jiang@rotman.utoronto.ca}}
\footnotetext[2]{\texttt{aditya@stat.berkeley.edu}}

\begin{abstract}
We study mixture of linear regression (random coefficient) models,
which capture population heterogeneity by allowing the regression
coefficients to follow an unknown distribution \(G^*\). In contrast to
common parametric methods that fix the mixing distribution form and
rely on the EM algorithm, we develop a fully nonparametric maximum
likelihood estimator (NPMLE). We show that this estimator exists under
broad conditions and can be computed via a discrete approximation
procedure inspired by the exemplar method. We further establish
theoretical guarantees demonstrating that the NPMLE achieves
near-parametric rates in estimating the conditional density of \(Y
\mid X\), both for fixed and random designs, when \(\sigma\) is known
and \(G^*\) has compact support. In the random design setting, we
also prove consistency of the estimated mixing distribution in the
L\'evy--Prokhorov distance. Numerical experiments indicate that our
approach performs well and 
additionally enables posterior-based individualized coefficient
inference through an empirical Bayes framework. 
\end{abstract}

\noindent%

\section{Introduction}
\label{section:introduction}

Given a univariate response $Y$ and a $p$-dimensional
regressor $X$, the linear
regression model with homoscedastic Gaussian errors assumes that $Y
\mid X = x$ is normal with mean $x^\T \beta$ and variance $\sigma^2$
for some $\beta \in \mathbb{R}^p, \sigma > 0$. In
contrast, the mixture of linear regression model assumes $Y \mid
X = x$ has the \textit{mixture} density  (below $\phi$ denotes
the standard normal density)
\begin{equation}\label{cond.dens}
  y \mapsto f_x^{G^*}(y) := \int \frac{1}{\sigma} \phi \left(\frac{y -
      x^\T \beta}{\sigma} 
\right) \diff G^*(\beta)
\end{equation}
for some probability measure $G^*$ on $\mathbb{R}^p$ and $\sigma >0
$. Equivalently, given $X = x$, 
\begin{equation}\label{model:mixlin}
  Y = x^\T \beta + \sigma \epsilon \qt{where $\beta \sim G^*$ and
    $\epsilon \sim N(0, 1)$ are independent}.  
\end{equation}
Mixture of linear regression models, also known as random coefficient
regression models \citep{hildreth1968some, Longford1994RandomCM,
  beran1994minimum, beran1992estimating, beran1996nonparametric}, 
offer a simple way to capture population heterogeneity
\citep{quandt1958estimation, de1989mixtures, jordan1994hierarchical,
  faria2010fitting}. They have been widely used in numerous fields, including biology
\citep{martin2008chipmix}, economics 
\citep{battisti2008spatially}, engineering \citep{liem2015surrogate},
epidemiology \citep{turner2000estimating}, marketing
\citep{wedel2012market}, and transportation
\citep{kim2023finite}.

Suppose we observe $n$ independent observations $(x_1, y_1), \dots, (x_n,
y_n)$ from \eqref{model:mixlin}, i.e., 
\begin{equation}\label{model:mixlin_i_1}
  y_i = x_i^\T \beta^i + \sigma \epsilon_i \qt{for $i = 1, \dots, n$,}
\end{equation}
where $\beta^1, \dots, \beta^n, \epsilon_1, \dots, \epsilon_n$ are independent with
\begin{equation}\label{model:mixlin_i_2}
  \beta^1, \dots, \beta^n \overset{\text{i.i.d.}}{\sim} G^* \text{ and
  } \epsilon_1, \dots, \epsilon_n \overset{\text{i.i.d.}}{\sim} N(0, 1). 
\end{equation}
If $\sigma$ and $G^*$ are known, this is a
Bayesian model with parameters $\beta^i, i = 1, \dots, n$, and one can
perform individual inference on each
$\beta^i$ via its posterior distribution: 
\begin{equation}\label{fullpost}
  \mathbb{P} \left\{\beta^i \in A \mid x_i, y_i \right\} = \frac{\int_A
    \frac{1}{\sigma} \phi \left(\frac{y_i - x_i^{\T} \beta}{\sigma}
    \right) dG^*(\beta)}{\int
    \frac{1}{\sigma} \phi \left(\frac{y_i - x_i^{\T} \beta}{\sigma}
    \right) dG^*(\beta)}
\end{equation}
for subsets $A \subseteq \mathbb{R}^p$. Point
estimates for $\beta^i$ can be obtained via the
posterior mean: 
\begin{equation}\label{postmean}
\hat{\beta}^i_{\OB} :=  \mathbb{E} \left( \beta^i \mid x_i, y_i \right) = \frac{\int
    \frac{1}{\sigma} \phi \left(\frac{y_i - x_i^{\T} \beta}{\sigma}
    \right) \beta dG^*(\beta)}{\int
    \frac{1}{\sigma} \phi \left(\frac{y_i - x_i^{\T} \beta}{\sigma}
    \right) dG^*(\beta)}. 
\end{equation}
The subscript OB in $\hat{\beta}^i_{\OB}$ stands for ``Oracle
Bayes''; Oracle here is used to refer to the fact that $G^*$ is
typically unknown and thus known only to an Oracle. 

This ability to do individual inference on the regression coefficient
$\beta^i$ corresponding to each separate data point $(x_i, y_i)$ is
the main attractive feature of the mixture of linear regression
model. This would, of course, require knowledge of $G^*$ (as well as
$\sigma$). The goal of this paper is to study the problem of
estimating $G^*$ from the data $(x_1, y_1), \dots, (x_n, y_n)$. We
shall assume for most of the paper that $\sigma$ is known. In
practice, it is easy to estimate $\sigma$ by $\hat{\sigma}$ using a
simple cross-validation procedure as described in Subsection
\ref{subsection:select_sigma}. If 
$G^*$ is estimated by, say, a discrete probability measure $\hat{G} :=
\sum_{j=1}^{\hat{k}} \hat{\pi}_j \delta_{\{\hat{\beta}_j\}} $, then
the posterior distribution \eqref{fullpost} and the posterior mean
\eqref{postmean} will be estimated by
\begin{equation}
  \label{postmean.est}
  \sum_{j=1}^{\hat{k}} \left[\frac{\hat{\pi}_j \phi
    \left(\frac{y_i - x_i^{\T} \hat{\beta}_j}{\hat{\sigma}} \right)}{\sum_{l=1}^{\hat{k}}
    \hat{\pi}_l \phi \left(\frac{y_i - x_i^{\T}
        \hat{\beta}_l}{\hat{\sigma}} 
    \right)} \right] \delta_{\{\hat{\beta}_j\}} ~~\text{ and }~~ \hat{\beta}^i_{\EB} := \frac{\sum_{l=1}^{\hat{k}} \hat{\beta}_l \hat{\pi}_l \phi 
    \left(\frac{y_i - x_i^{\T} \hat{\beta}_l}{\hat{\sigma}} \right)}{\sum_{l=1}^{\hat{k}}
    \hat{\pi}_l \phi \left(\frac{y_i - x_i^{\T} \hat{\beta}_l}{\hat{\sigma}} \right)},
\end{equation}
respectively. These can be used for approximate individual inference
for $\beta^i, i = 1, \dots, n$. $\EB$ in
$\hat{\beta}^i_{\EB}$ denotes ``Empirical Bayes'' (empirical as
$G^*, \sigma$ are estimated from data).   

Most existing methods for estimating \(G^*\) assume a parametric form,
such as a discrete distribution with a known number of atoms, and then
use maximum likelihood via the EM algorithm
\citep{leisch2004flexmix,faria2010fitting}. In contrast, we take a
nonparametric approach with no parametric assumptions on \(G^*\)
but still relying on maximum likelihood. We thus use
\emph{nonparametric maximum likelihood estimation} (NPMLE). 
 

NPMLE for mixture models has a long history, beginning with
\citet{robbins1950generalization} and \citet{kiefer1956consistency}.
Comprehensive treatments include
\citet{lindsay1995mixture,groeneboom1992information,bohning2000computer,schlattmann2009medical},   
with renewed interest more recently for normal mixtures
\citep{zhang2009generalized,koenker2014convex,dicker2016high,saha2020nonparametric,
deb2021two,polyanskiy2020self}.
Beyond normal mixtures, \citet{gu2020nonparametric} study mixtures of
binary regression, and \citet{jagabathula2020conditional} deal with
mixtures of logit models. 


The likelihood function here is the conditional density of $y_1,
\dots, y_n$ given $x_1, \dots, x_n$ and its logarithm (the
log-likelihood function) is given by  
\begin{equation}\label{eq:defn_f_G}
G \mapsto  \sum_{i=1}^n \log f_{x_i}^G(y_i) \qt{where  
$f^G_{x_i}(y_i) = \frac{1}{\sigma}\int \phi\left( \frac{y_i -x_i^\T
    \beta}{\sigma} \right) \diff G(\beta), i=1,\dots,n.$}  
\end{equation}
We impose bounds on the support of $G$ in the maximization of the
likelihood. Specifically, we consider, for a given set $K \subseteq
\mathbb{R}^p$, the NPMLE: 
\begin{equation}\label{eq:maximize_G}
\hat{G}   \in \argmax \left\{ \sum_{i=1}^n \log f_{x_i}^G(y_i) : G
  \text{ is a probability supported on } K \right\}  
\end{equation}
If no information about the support of $G^*$ is available, then one
can either take $K$ to be $\mathbb{R}^p$ or a large compact set
such as a closed ball centered at the origin having a large radius. 

The optimization in \eqref{eq:maximize_G} is infinite-dimensional (as $K$
is usually uncountable) and convex as
the constraint set (the set of all probability measures on $K$) is
convex and the objective function is concave in $G$. In 
Section~\ref{section:existenceandcomputation}, we prove $\hat{G}$ exists
when $K$ is compact or when $K$ satisfies a technical condition which
holds when $K = \mathbb{R}^p$. We provide an  
iterative algorithm for computing  an approximate
 solution $\hat{G}$ that is discrete. This algorithm is inspired by
 the exemplar method that was previously used for computing
 approximate NPMLEs in Gaussian location mixture density estimation
 (see e.g., \citet{bohning1992computer}, \citet{lashkari2008convex},
 \citet{soloff2021multivariate}) but our setting  introduces
additional complications (especially in computing the exemplars)
detailed in Section \ref{section:existenceandcomputation}.  


Our estimator $\hat{G}$ performs well without excessive overfitting
despite being obtained by maximization over a very large class of
probability measures. We prove that the estimated conditional density
function $(x, y) \mapsto f_{x}^{\hat{G}}(y)$ approximates the true
conditional density $(x, y) \mapsto f_{x}^{G^*}(y)$ with high accuracy
when $G^*$ is compactly supported and $\sigma$ is known. Using loss
functions based on squared Hellinger distance, we establish 
theoretical guarantees in both fixed and random design settings
(Theorems \ref{mainthem} and
\ref{them:randomdesign_predictionerror}). These results demonstrate
that our fully nonparametric approach effectively estimates $G^*$
while achieving near-parametric rates for conditional density
estimation. For random designs, we also prove $\hat{G}$ is
consistent for $G^*$ with their Lévy-Prokhorov distance converging to
zero in probability as $n 
\rightarrow \infty$. 



The remainder of this paper 
is organized as
follows. Section~\ref{section:existenceandcomputation} discusses  
existence and computation of the NPMLE. Section~\ref{section:hellinger_distance_bound} contains our
theoretical results on the accuracy of the
NPMLE. Section~\ref{section:experimental} contains experimental
results  
including simulation studies  and real data
analysis. Section \ref{section:conclusions} discusses issues naturally
connected to our main results. Proofs of all our results are in Section
\ref{allproofs} of the supplement. The supplement also contains two
tables (see Section \ref{subsec:detailed_coef}) showing the results
of simulations in Sections 
\ref{sec:sinusoid_simulation} and \ref{sec:changepoint_simulation}.    

\section{Existence and Computation}
\label{section:existenceandcomputation}

Our first result is on existence, and it uses the following notation. For a probability measure $G$, let $\fv^G =
(f^G_{x_1}(y_1),\dots,f^G_{x_n}(y_n))^\T$ with $f_{x_i}^G(y_i)$ in
\eqref{eq:defn_f_G}.  When $G$ is the Dirac measure concentrated on
some $\beta \in \mathbb{R}^p$, we write $\fv^{\beta}$ for $\fv^G$. Let $\PS_K :=
\left\{\fv^{\beta} : \beta \in K \right\}$. 

\begin{theorem}\label{thm:existence}
  Assume $K$ is closed and satisfies
 one of the following two conditions:
  \begin{enumerate}
  \item $K$ is bounded (and hence compact).
  \item $P_V(x) \in K$ for every $x \in K$ and linear subspace $V$ of
     $\mathbb{R}^p$ (here $P_V(x)$ is the projection of $x$ onto the
    linear subspace $V$). 
  \end{enumerate}
  Then, for every dataset $(x_1, y_1), \dots, (x_n, y_n)$, the
  optimization problem in \eqref{eq:maximize_G} admits a solution
  $\hat{G}$ that is a probability measure supported 
  on at most $n$ points in $K$. Moreover the vector $\fv^{\hat{G}}$ is
  unique for every maximizer $\hat{G}$ and is the unique solution to: 
  \begin{equation}
\text{maximize} ~~~ L(\fv) := \frac{1}{n} \sum_{i=1}^n \log \fv(i) ~~~
\text{subject to} ~~ \fv = (\fv(1), \dots, \fv(n)) \in
\mathrm{conv}(\PS_K),  
\label{conProbVectorOpt}
\end{equation}
where $\mathrm{conv}(\PS_K)$ denotes the convex hull of the set
$\PS_K$. 
 \end{theorem}
When $K$ is compact, $\PS_K$ is also compact,
and the existence of $\hat{G}$ follows directly from \citet[Theorem
18]{lindsay1995mixture}. When $K$ is not compact (e.g., $K =
\mathbb{R}^p$), $\PS_K$ fails to be compact as $0 \notin \PS_K$ is a limit
point of $\PS_K$. Although \citet[Subsection 5.2.2]{lindsay1995mixture} discusses
non-compact $\PS_K$, their approaches do not
directly apply here --- for example, it is unclear if $\PS_K \cup \{0\}$
is compact in our case. Consequently, our argument is
more involved, and we require the technical condition 2 on $K$ in Theorem \ref{thm:existence}. 


Next we show $\hat{G}$ is not unique if the design
matrix $\mathbf{X} = [x_1:x_2:\dots:x_n]^\T$ does not have full column
rank. We do not know if $\hat{G}$ will be unique if $\mathbf{X}$ is of
full column rank.  

\begin{proposition}
If $\mathbf{X}$ does not
have full column rank, then $\hat{G}$ is not unique. 
\label{proposition:nonunique}
\end{proposition}

Next result is a characterization of $\hat{G}$ via
the first order optimality condition for \eqref{eq:maximize_G}. 
\begin{proposition}\label{prop:opt_condition}
  $\hat{G}$ solves \eqref{eq:maximize_G} if and only if
  \begin{equation}\label{eq:D_G_beta}
  D(\hat{G}, \beta) :=  \frac{1}{n} \sum_{i=1}^n
    \frac{f^{\beta}_{x_i}(y_i)}{f^{\hat{G}}_{x_i}(y_i)} - 1 \leq 0 \qt{for
      all $\beta \in K$}. 
  \end{equation}
  Further, for every $\hat{G}$ maximizing \eqref{eq:maximize_G}, we
  have
  \begin{equation}\label{eq:D_G_as}
    D(\hat{G}, \beta) = 0 \qt{for $\beta$ a.s $\hat{G}$}. 
  \end{equation}
  If $\hat{G}$  is discrete and $\tilde{\beta} \in \text{int}(K)$
  ($\text{int}(K)$ is the interior of $K$) is a support point of $\hat{G}$,
  then the gradient of $D(\hat{G}, \beta)$ with respect to $\beta$
  equals 0 at $\beta = \tilde{\beta}$. 
\end{proposition}
Proposition
\ref{prop:opt_condition}, along with the
Carath{\'e}odory theorem, leads to the  following. For a probability
vector $w = (w_1, \dots, w_n)$  with $w_i \geq 0$ and $\sum_{i=1}^n
w_i = 1$, let  
  \begin{equation}\label{sw}
    S(w) := \left\{\beta \in \mathbb{R}^p : \left(\sum_{i=1}^n w_i x_i
      x_i^T \right) \beta = \sum_{i=1}^n w_i x_i y_i \right\}. 
  \end{equation}
If $\sum_{i=1}^n w_i x_i x_i^T$ is nonsingular, then $S(w)$ is the
singleton $\left(\sum_{i=1}^n w_i x_i
  x_i^T \right)^{-1} \left(\sum_{i=1}^n w_i x_i y_i \right)$. 

\begin{proposition}\label{exem_lemma}
  If $\hat{G}$ is a discrete solution to \eqref{eq:maximize_G} and
  $\tilde{\beta} \in \text{int}(K)$ is a support point of $\hat{G}$, then
  $\tilde{\beta} \in 
  S(w)$ for some probability vector $w$ with at most $p+1$ non-zero
  entries. 
\end{proposition}

Proposition \ref{exem_lemma} shows every support point in
$\text{int}(K)$ of every discrete NPMLE $\hat{G}$ is in
$\mathcal{M} := \cup_w S(w)$, the union being over all
$(p+1)$-sparse probability vectors $w$. This
suggests the following algorithm to compute an 
approximate solution to \eqref{eq:maximize_G}. The basic idea (see
e.g., \cite{koenker2014convex}) is to restrict the support of $G$ to a
finite set $\mathcal{A} \subseteq \mathcal{M}$ that is constructed as
follows. Fix a large $M$ and generate $\beta^{(j)}, j = 1,
\dots, M$ in $K$ as follows: 
\begin{enumerate}
\item Generate a random subset $S \subseteq \{1, \dots, n\}$ of
  cardinality $p+1$.
\item Generate a probability vector $w_i, i \in S$. We use the simple choice
   $w_i = 1/(p+1), i \in S$.
\item Take $\beta^{(j)} = \left(\sum_{i \in S} w_i x_i x_i^T \right)^{-1} \left(\sum_{i \in S} w_i
  x_i y_i\right)$. If $\sum_{i \in S} w_i x_i x_i^T$ is singular
or if the generated $\beta^{(j)}$  is not in $K$, then discard it and
repeat steps 1, 2, 3.  
\end{enumerate}
With $\mathcal{A} = \{\beta^{(1)}, \dots, \beta^{(M)}\}$, we solve the
following discrete approximation to \eqref{eq:maximize_G}: 
\begin{equation}\label{disc.prob}
\text{maximize} ~~~ \frac{1}{n} \sum_{i=1}^n \log
\left(\sum_{j=1}^M w_j \fv_{x_i}^{\beta^{(j)}}(y_i) \right) ~~~
\text{subject to} ~~ w_1, \dots, w_M \geq 0 \text{ with } \sum_{j=1}^M w_j
= 1. 
\end{equation}
\eqref{disc.prob} can be solved via standard
algorithms such as the conditional gradient method (e.g.,
\citet{jaggi2013revisiting}) or  via standard
software for convex optimization such as \texttt{mosek}   
\citep{rmosek}.

We summarize our overall algorithm in Algorithm~\ref{exemplaralgo}. 

\begin{algorithm}[H] \label{exemplaralgo}
\caption{Exemplar Algorithm for obtaining an approximate NPMLE $\hat{G}$}
\KwIn{Data $(x_1, y_1), \dots, (x_n, y_n)$, $\sigma > 0$,
  constraint $K$, integer $M$ (e.g., $M = 4n$)}

Generate candidate vectors $\beta^{(j)}$, $1 \leq j \leq M$, in $K$ using
\[
\beta^{(j)} = \left(\sum_{i \in S} x_i x_i^T\right)^{-1} \left(\sum_{i \in S} x_i y_i\right)
\]
where $S$ is uniformly randomly generated with $|S| = p+1$ (discard
$\beta^{(j)} \notin K$)\

Solve
\eqref{disc.prob} to obtain $\hat{w}_1, \dots, \hat{w}_M$. 

\KwOut{$\hat{G} = \sum_{j=1}^M \hat{w}_j
  \delta_{\{\beta^{(j)}\}}$ is our approximate solution to
  \eqref{eq:maximize_G}.} 
\end{algorithm}

We refer to Algorithm \ref{exemplaralgo} as the ``Exemplar
Method'' because $\beta^{(1)}, \dots, \beta^{(M)}$ are 
examples for the possible support points of $\hat{G}$. For Gaussian
location mixtures, exemplar methods have been employed by
\citet{bohning1992computer}, \citet{lashkari2008convex} and
 \citet{soloff2021multivariate}. Exemplar methods  avoid placing grids
 and are thus useful in  multidimensions.

\section{Theoretical Accuracy
  Results} \label{section:hellinger_distance_bound}

We provide theoretical guarantees for our estimator
$\hat{G}$. For these results, we assume $K$ in definition
\eqref{eq:maximize_G} takes the form $K = \{\beta \in \mathbb{R}^p :
  \|\beta\| \le R\}$ for some $R > 0$, where $\|\cdot\|$ denotes the
Euclidean norm. While any compact set $K$ can be embedded in such a
ball, our results specifically require this ball formulation and do
not extend to non-compact sets. 


We focus on convergence rates for estimating the conditional density
of $Y$ given $X$. Under model \eqref{model:mixlin}, the true
conditional density (of $Y$ given $X = x$) is $f_x^{G^*}(\cdot)$,
while our estimate is $f_x^{\hat G}(\cdot)$. We use their discrepancy
via the squared
Hellinger distance: 
\begin{equation}
  \mathfrak{H}^2\left(f^{\hat{G}}_x, f^{G^*}_x \right)  = \int \left\{
    \sqrt{f^{\hat{G}}_x(y)} - \sqrt{f^{G^*}_x(y)} \right\}^2 \diff y.  
  \label{eq:hellinger}
\end{equation}
To evaluate overall estimation accuracy in fixed design, we average across all
design points: 
\begin{equation}
\mathfrak{H}_{\mathrm{fixed}}^2\left(f^{\hat{G}}, f^{G^*}\right) = \frac{1}{n}
\sum_{i=1}^n \mathfrak{H}^2\left(f^{\hat{G}}_{x_i},
  f^{G^*}_{x_i}\right). 
\label{eq:fixed_loss}
\end{equation}
The following theorem gives a bound on
$\mathfrak{H}_{\mathrm{fixed}}^2(f^{\hat{G}}, f^{G^*})$ that holds for every $x_1,
 \dots, x_n$. 

\begin{theorem}[Fixed design conditional density estimation
  accuracy]\label{fdrate} \label{mainthem}
  Consider data $(x_1, y_1),$ $\dots,$ $(x_n, y_n)$ with $n \geq 3$ where
$x_1, \dots, x_n$ are fixed and $y_i \overset{\text{ind}}{\sim}
f_{x_i}^{G^*}(\cdot)$. Assume that
\begin{equation*}
  G^* \left\{\beta \in \mathbb{R}^p : \|\beta\| \le R \right\} = 1 ~~ \text{ and } ~~ \max_{1 \leq i
    \leq n} \|x_i\| \leq B 
\end{equation*}
for some $B > 0$ and $R > 0$. Let $\hat{G}$ be the estimator for $G^*$
defined as in \eqref{eq:maximize_G} with $K = \left\{\beta \in
  \mathbb{R}^p : 
  \|\beta\| \leq R \right\}$. Let $\epsilon_n = 
\epsilon_n(B, R, \sigma)$ be defined via  
\begin{equation}\label{eq:epsilon_defn}
  \epsilon^2_n := n^{-1} \max \left(\left(\Log  \frac{n}{\sqrt{\sigma}} \right)^{p+1} , \left(\frac{R B}{\sigma} \right)^p \left(\Log \left\{\frac{n}{\sqrt{\sigma}} \left(\frac{\sigma}{RB} \right)^p \right\} \right)^{\frac{p}{2} + 1}  \right),
\end{equation}
where we use $\Log x := \max(\log x, 1)$.  Then there exists a
constant $C_p$ such that  
\begin{equation}
\label{tailbound}
\PP \left\{\mathfrak{H}_{\mathrm{fixed}}(f^{\hat{G}}, f^{G^*}) \geq t \epsilon_n \sqrt{C_p}  \right\} \leq  \exp(-nt^2 \epsilon_n^2) \qt{ for every $t \geq 1$},  
\end{equation}
and 
\begin{equation}
\E\mathfrak{H}_{\mathrm{fixed}}^2(f^{\hat{G}}, f^{G^*}) \leq C_p
\epsilon^2_n.  
\label{expectationbound}
\end{equation}
\end{theorem}

We next consider the random design setting with common design density
$\mu$ and loss: 
\begin{equation*}
  \mathfrak{H}^2_{\mathrm{random}} \left(f^{\hat{G}}, f^{G^*} \right)
  := \int \mathfrak{H}^2\left(f^{\hat{G}}_x, f_{x}^{G^*} \right) d\mu(x).
\end{equation*}

\begin{theorem}[Random design conditional density estimation
  accuracy]\label{them:randomdesign_predictionerror}
Consider $n \geq 3$ and i.i.d. data $(x_1, y_1), \dots, (x_n, y_n)$
with $n \geq 3$ with $x_i \sim \mu$ and $y_i|x_i \sim
f_{x_i}^{G^*}(\cdot)$. Assume
\begin{equation*}
  G^* \left\{\beta \in \mathbb{R}^p : \|\beta\| \le R \right\} = 1 ~~
  \text{ and } ~~ \mu \{x \in \mathbb{R}^p : \|x\| \leq B\} = 1
\end{equation*}
for some $B > 0$ and $R > 0$. Let $\hat{G}$ be the estimator for $G^*$
defined as in 
\eqref{eq:maximize_G} with $K = \left\{\beta \in \mathbb{R}^p :
  \|\beta\| \leq R \right\}$. Let $\epsilon_n$ be as in
\eqref{eq:epsilon_defn} and let $\beta_n$ be such that its square
$\beta_n^2$ equals: 
\begin{equation}\label{eq:beta_defn}
 n^{-1} \max \left(\left(\Log  \frac{n(BR +
        \sigma)^2}{\sigma^2} \right)^{p+1} , \left(\frac{R B}{\sigma}
    \right)^p \left(\Log \left\{\frac{n (BR + \sigma)^2}{\sigma^2}
        \left(\frac{\sigma}{RB} \right)^p \right\}
    \right)^{\frac{p}{2} + 1}  \right), 
\end{equation}
where, again, $\Log x := \max(\log x, 1)$. Then there exists $C_p$
such that
\begin{equation}
\label{tailbound.random}
\PP \left\{\mathfrak{H}_{\mathrm{random}}(f^{\hat{G}}, f^{G^*}) \geq t
  \left(\epsilon_n + \beta_n \right) \sqrt{C_p}  \right\} \leq
\exp(-nt^2 \epsilon_n^2) + \exp \left(-\frac{nt^2 \beta_n^2}{C_p}
\right) 
\end{equation}
for every $t \geq 1$, and 
\begin{equation}
\E\mathfrak{H}_{\mathrm{random}}^2(f^{\hat{G}}, f^{G^*}) \leq C_p
\left(\epsilon^2_n + \beta_n^2 \right). 
\label{expectationbound.random}
\end{equation}
\end{theorem}

Both $\epsilon_n$ (defined in \eqref{eq:epsilon_defn}) and $\beta_n$
(defined as the square root of \eqref{eq:beta_defn}) satisfy
$\epsilon_n^2 = O(n^{-1} (\log n)^{p+1})$ and $\beta_n^2 = O(n^{-1}
(\log n)^{p+1})$ as $n \rightarrow \infty$ (with $R, B, \sigma$
fixed). Theorem \ref{mainthem} and Theorem
\ref{them:randomdesign_predictionerror} 
   give the same rate $O(n^{-1} 
   (\log n)^{p+1})$ (assuming $R, B, \sigma$ are fixed) for
   $\mathfrak{H}^2_{\mathrm{fixed}}(f^{\hat{G}}, f^{G^*})$ and
   $\mathfrak{H}^2_{\mathrm{random}}(f^{\hat{G}}, f^{G^*})$
   respectively.  Thus, in both fixed and random designs,
   the NPMLE is a very good estimator for the true conditional
   density function if $p$ is small.

In the next result (proved in 
Appendix~\ref{appendix:randomdesign_identifiability}), we establish
identifiability of $G^*$ in the random design setting. We use the same
assumptions as in  Theorem~\ref{them:randomdesign_predictionerror}
with the additional assumption that the support of $\mu$ contains an
open set.

\begin{theorem}[Identifiability under random design]
Suppose $G_1$ and $G_2$ are two probability measures contained in $\left\{\beta \in \mathbb{R}^p : \|\beta\| \le R \right\}$. Assume that the support of $\mu$ is
contained in $ \{x \in \mathbb{R}^p : \|x\| \leq B\}$ for some $B > 0$
and also that the support of $\mu$ contains an open set. If
\[
\int  \frac{1}{\sigma} \phi\left(\frac{y - x^\T\beta}{\sigma} \right) d G_1(\beta) = 
\int \frac{1}{\sigma} \phi\left(\frac{y - x^\T\beta}{\sigma} \right) d G_2(\beta)
\]
holds for all $y \in \RS$ and all $x$ in the support of $\mu$, then $G_1 = G_2$.
\label{thm:randomdesign_identifiability}	
\end{theorem}

Theorem \ref{thm:weakconsistency} shows that the NPMLE is weakly
consistent (in random design) in 
that its L\'{e}vy-Prokhorov distance to  $G^*$ approaches $0$ in
probability as $n \rightarrow \infty$. The
  L\'{e}vy–Prokhorov metric $\dlp$ metrizes weak
convergence of probability measures
\citep[Chapter~11.3]{dudley1989real}. 

\begin{theorem}
Consider $n \geq 3$ and i.i.d. data $(x_1, y_1), \dots, (x_n, y_n)$
with $n \geq 3$ with $x_i \sim \mu$ and $y_i|x_i$ having the density
\eqref{cond.dens}. Assume that $G^* \left\{\beta \in \mathbb{R}^p :
  \|\beta\| \le R \right\} = 1$ for some $R > 0$ and that the  support
of $\mu$ is contained in $\{x \in \mathbb{R}^p : \|x\| \leq B\}$
for some $B > 0$ and also that the support of $\mu$ contains an 
open set.

Let $\hat{G}_n$, where we add subscript $n$ to denote
      the number of data points,  be the estimator for $G^*$ defined
      as in \eqref{eq:maximize_G} with $K = \left\{\beta \in
        \mathbb{R}^p : \|\beta\| \leq R \right\}$. Then
      $\dlp(\hat{G}_n, G^*)\rightarrow 0$ in probability as $n
      \rightarrow \infty$ where $\dlp$  is the L\'{e}vy–Prokhorov
      metric. 
\label{thm:weakconsistency}
\end{theorem}

\noindent
\textbf{Relation to Existing Results:} \label{page:proof_distinct}
The proof of Theorem~\ref{mainthem} relies on empirical process and
metric entropy arguments, and follows a strategy similar to those used
in Gaussian location mixture density estimation
\citep{ghosal2001entropies, zhang2009generalized,
  saha2020nonparametric}. Key here is to bound the metric
entropy (Theorem~\ref{theorem:metric_entropy}) of  
\[
  \mathcal{M}_K \;=\; \bigl\{\, f^G_x(y) \;:\; G \text{ is a
    probability measure supported on } K \bigr\}, 
\]
under an $L_{\infty}$ metric. The above function class depends on both
$x$ and $y$, which makes it distinct from the standard
density-estimation classes in the aforementioned literature.

No analogous results to
Theorem~\ref{them:randomdesign_predictionerror} exist in the
Gaussian mixture density estimation literature. For its proof, we use 
Theorem~\ref{mainthem} together with an existing empirical process
lemma (Lemma~\ref{thm:vdg}), which connects the random design loss to
the fixed design loss. The additional rate $\beta_n$ captures the cost
of moving from the fixed-design to the random-design setting, arising
from the growth rate of the bracketing number $N_{[]}(\epsilon,
\mathcal{G}, L_2(\mu))$ for 
\[
  \mathcal{G} \;=\; \Bigl\{\,
    x \;\mapsto\; \tfrac12\,\mathfrak{H}^2\bigl(f_x^G, f_x^{G^*}\bigr)
    : G \in \{\beta \in \mathbb{R}^p : \|\beta\|\le R\}
  \Bigr\}
\]
defined on $S_0 := \{x \in \mathbb{R}^p : \|x\| \le B\}$. 

Theorem \ref{thm:randomdesign_identifiability} is established using an
argument similar to the one used in the ``strong identifiability''
result in \citet[Proposition 2.2]{beran1994minimum}. Note that
identifiability is a well-studied issue for mixture models (see e.g.,
\cite{nguyen2013convergence, ho2016convergence}). Theorem
\ref{thm:weakconsistency}  is proved using the Hellinger error bound
from Theorem \ref{them:randomdesign_predictionerror} as well as a
variant (Lemma~\ref{strongid}) of \citet[Proposition
2.2]{beran1994minimum}.

\section{Experimental Results}
\label{section:experimental}

\subsection{Cross-Validation for \texorpdfstring{$\sigma$}{sigma}} 
\label{subsection:select_sigma}

For estimating $\sigma$, we employ $C$-fold cross-validation by
dividing dataset $\mathcal{D} = \{(x_i,y_i)\}_{i=1}^n$ into $C$ equal
parts. For each fold $c$ and fixed $\sigma$, we compute NPMLE estimate
$\hat{G}^{-c,\sigma}$ using all data except $\mathcal{D}_c$. Our
cross-validation score is
\begin{equation}
	\mathrm{CV}(\sigma) = - \sum_{c = 1}^C \sum_{(x_i,y_i) \in
		\mathcal{D}_c } \log \hat{f}^{-c, \sigma}_{x_i}(y_i)
               \text{ with } \hat{f}^{-c, \sigma}_{x_i}(y_i) = \frac{1}{\sigma} \int \phi
\left(\frac{y_i- x_i^\top\beta}{\sigma} \right) d\hat{G}^{-c, \sigma}(\beta). 
	\label{eq:cv_criterion}
\end{equation}
We select $\hat{\sigma}$ to minimize $\mathrm{CV}(\sigma)$, with 
candidate $\sigma$'s ranging from $\sigma_{\min}$
(typically 0.1) to $\sigma_{\max}$ ($=\sqrt{\mathrm{Var}(y)}$), with
intervals of 0.1 in log space. We use $C=5$ or $C = 10$. 

\subsection{BIC Trimming}
\label{subsec:BIC_selection}

When $G^*$ is finitely supported, the NPMLE
often estimates more components than actually present in $G^*$. For
clearer summarization and visualization, we reduce
components using the Bayesian Information Criterion (BIC). For an
NPMLE estimate with $K$ components with log-likelihood $L_K$, we take
$\mathrm{BIC}_K = -2 L_K + pK \log n$. To compute $\mathrm{BIC}_k$ for
$k = 1,\ldots,K-1$, we use the following iterative procedure starting
at $i = 1$ (until reaching one component). With $K+1-i$ components, we
remove the component with 
smallest mixing proportion, and 
  reoptimize the remaining mixing proportions to obtain log-likelihood
  $L_{K-i}$, calculate $\mathrm{BIC}_{K-i} = -2 L_{K-i} + p(K-i)\log
  n$, and  increment $i$. 

The final selected model corresponds to the number of components $k^*$
that minimizes $\mathrm{BIC}_k$. This BIC-based
procedure effectively balances model complexity and goodness-of-fit as
shown in the subsequent experimental results.  

\subsection{Simulation studies}

\subsubsection{Simulation: Discrete Distribution}\label{simpeg}

Figure \ref{fig:3comp_noisy} displays data generated from our model
with $n = 200$, $\sigma = 0.5$, and a discrete $G^*$ assigning
probabilities $0.3, 0.3, 0.4$ to $\beta$-values $(3, -1), (1, 1.5),
(-1, 0.5)$. Each observation has covariates $x_i = (1, w_i)^T$ with
$w_i \sim \text{Uniform}[-1, 3]$. In Figure~\ref{fig:3comp_true},
points are color-coded by their generating line, though this
information would be unavailable in practice.

\begin{figure}[!htbp]
	\centering
	\begin{subfigure}[t]{0.4\textwidth}
		\caption{}
		\includegraphics[width = \textwidth,valign = t]{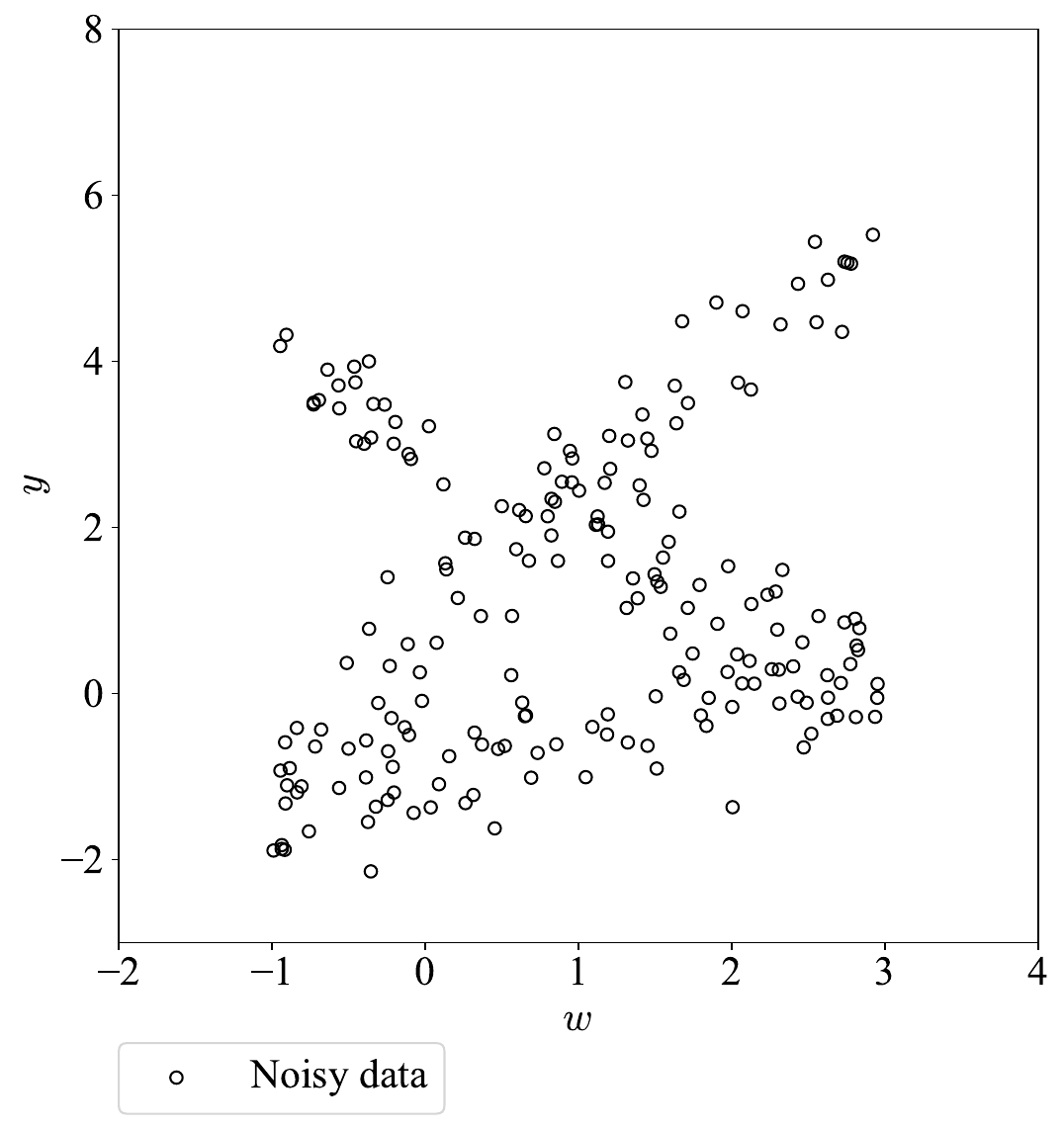}
		\label{fig:3comp_noisy}
	\end{subfigure}
	\begin{subfigure}[t]{0.4\textwidth}
		\caption{}
		\includegraphics[width = \textwidth,valign = t]{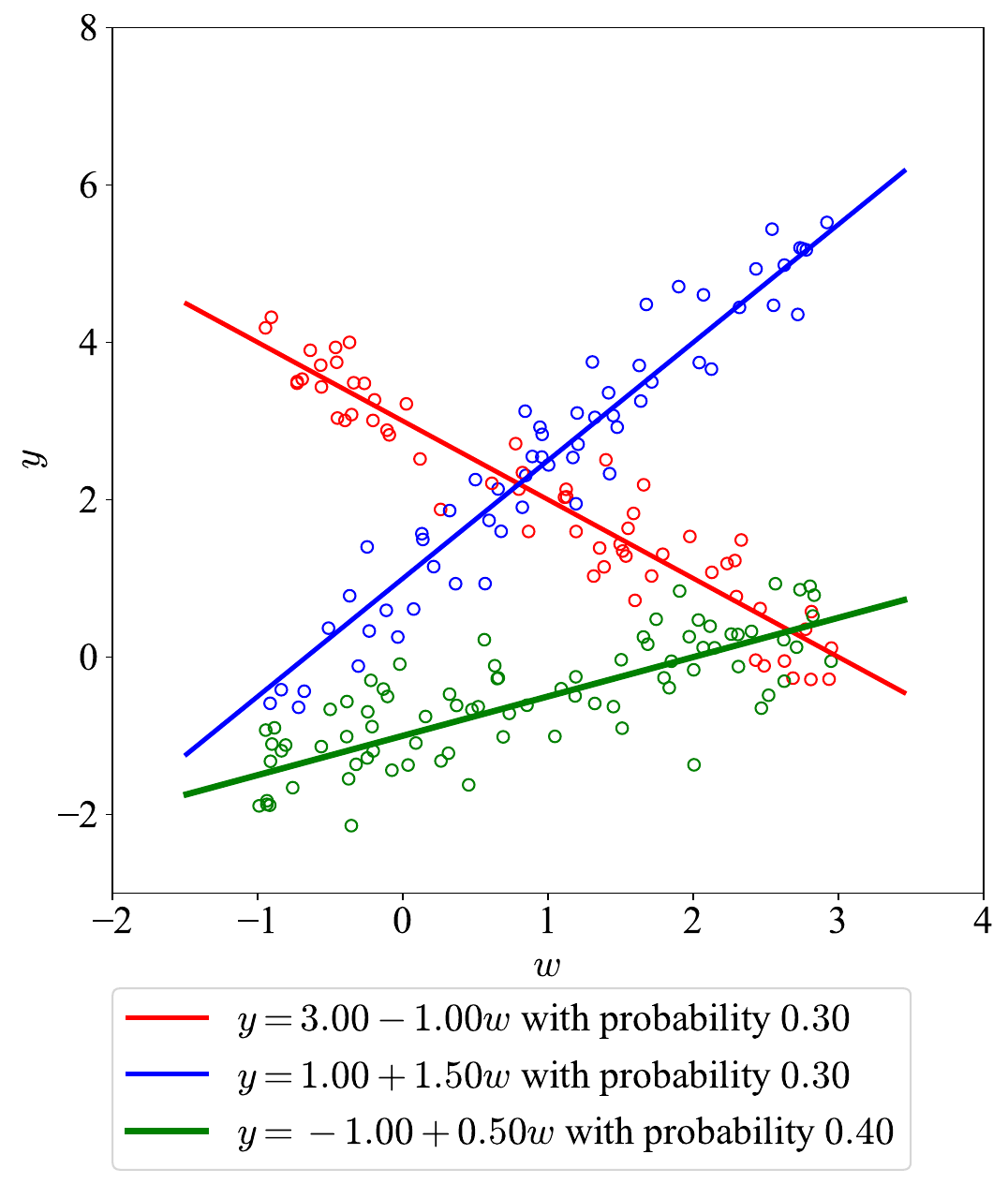}
		\label{fig:3comp_true}
	\end{subfigure}
	\caption{\subref{fig:3comp_noisy} Data points; \subref{fig:3comp_true} True regression components.}
	\label{fig:3comp_noisytrue-intro}
\end{figure}

Our cross-validation procedure accurately estimated $\hat{\sigma} =
0.4953$ (true value: 0.5). The resulting NPMLE $\hat{G}_{\text{CV}}$
was a discrete measure supported on $\hat{k} = 11$ points, with
corresponding regression lines shown in
Figure~\ref{fig:3comp_fitted_withoutBIC} (line thickness indicates
estimated mixing probability). We colored each data point according to
its most probable posterior line. 

\begin{figure}[!htbp]
	\centering
	\begin{subfigure}[t]{0.4\textwidth}
		\subcaption{}
		\includegraphics[width = \textwidth,valign = t]{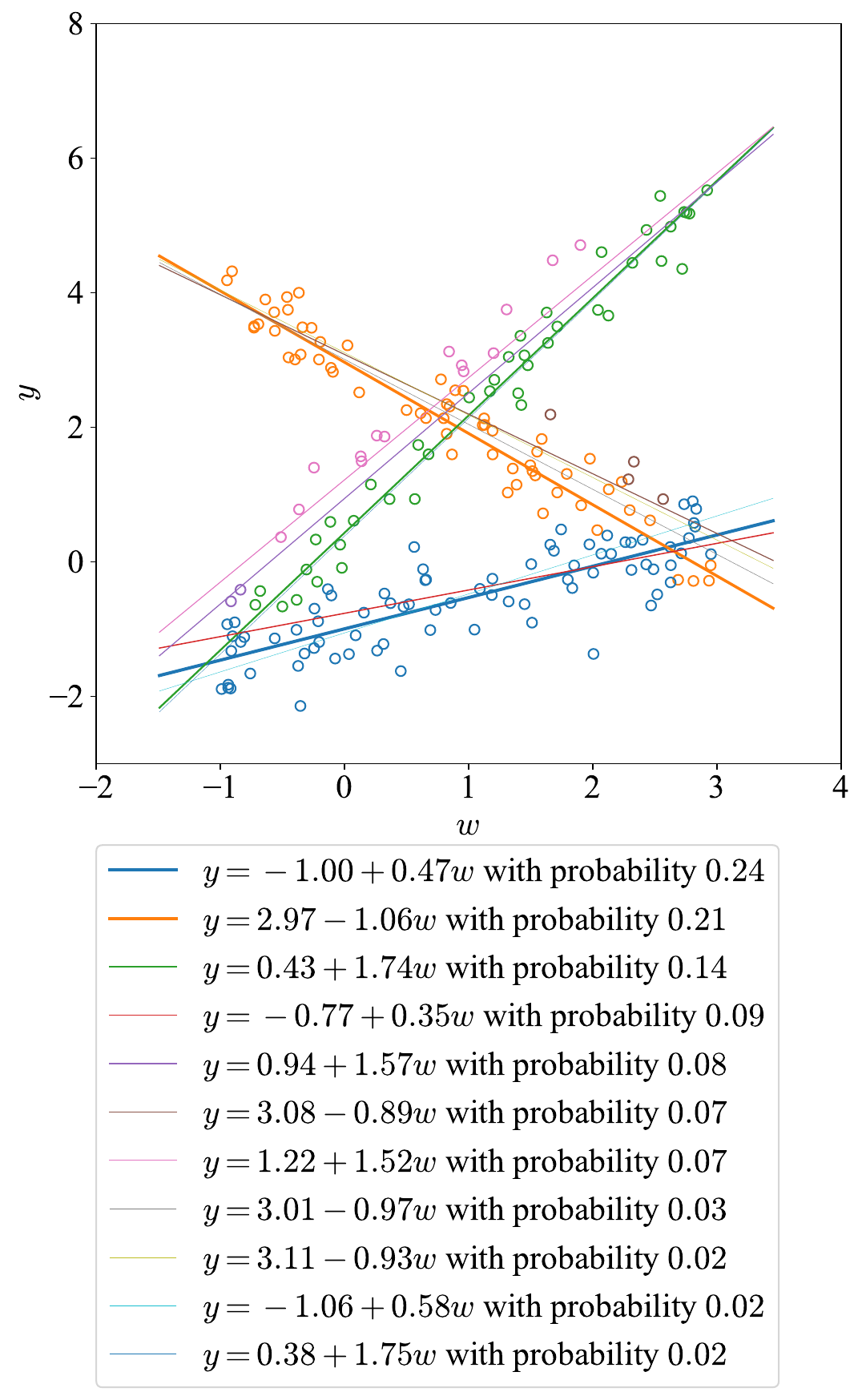}
		\label{fig:3comp_fitted_withoutBIC}
	\end{subfigure}
	\begin{subfigure}[t]{0.4\textwidth}
		\subcaption{}
		\includegraphics[width = \textwidth,valign = t]{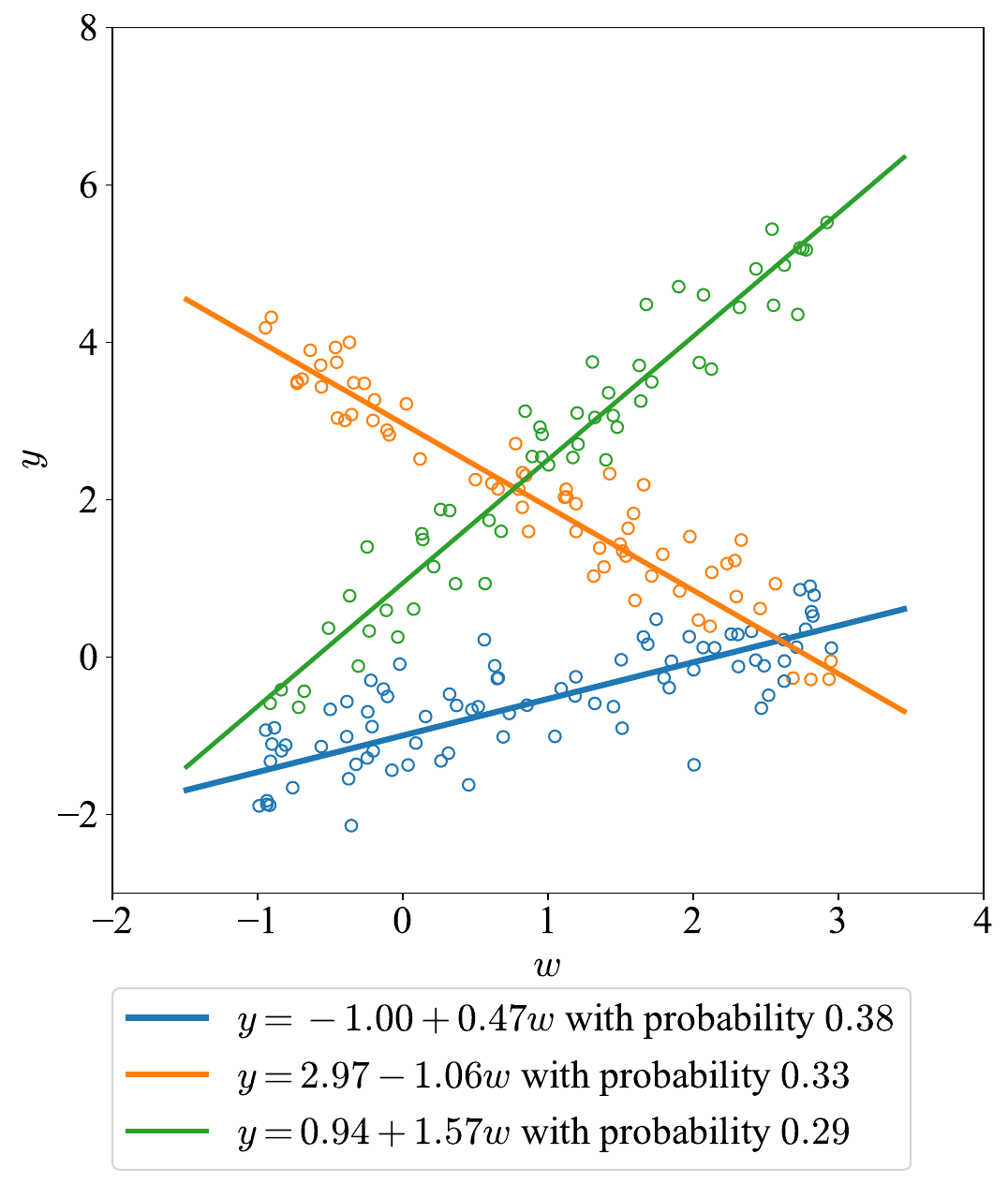}
		\label{fig:3comp_fitted}
	\end{subfigure}
	\caption{\subref{fig:3comp_fitted_withoutBIC} Fitted mixture before BIC selection; \subref{fig:3comp_fitted} Fitted mixture after BIC selection.}
	\label{fig:3comp_fitted-intro}
\end{figure}

\begin{figure}[!htbp]
\centering
\includegraphics[width = 0.8\textwidth]{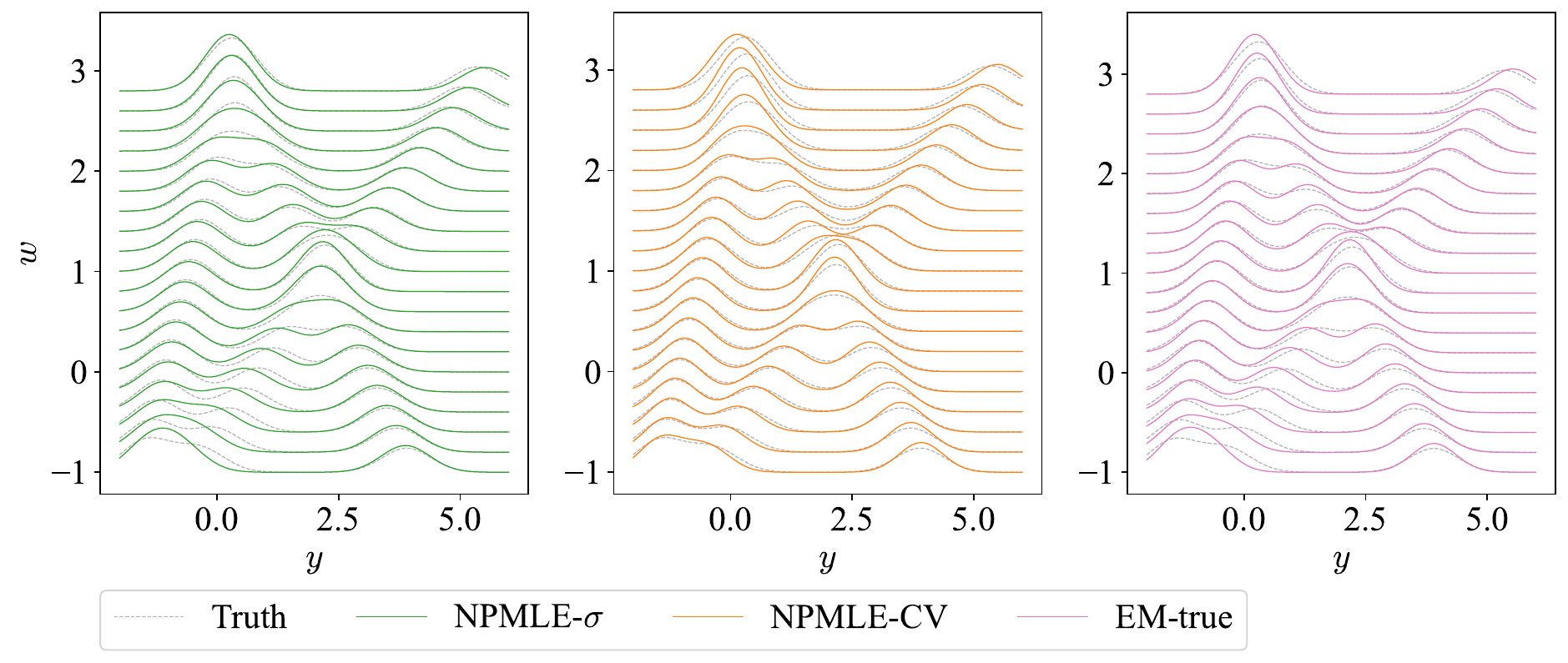}
\caption{Ridgeline plots of density functions $f_x^{G^*}(y)$ in comparison with its estimates via (i) NPMLE{\dash}$\sigma$ (NPMLE with known $\sigma$), (ii) NPMLE-CV (NPMLE with $\hat{\sigma}$ selected by cross-validation), and (iii) EM-true (EM initialized
  with true parameters of $G^*$ and $\sigma$) respectively.} 
\label{fig:3comp_density-intro}
\end{figure}

While NPMLE overestimates the number of components, BIC
selection correctly identified
exactly 3 components (Figure~\ref{fig:3comp_fitted}). The estimated
conditional density function 
$$f_x^{\hat{G}_{\text{CV}}, \hat{\sigma}}(y) = \frac{1}{\hat{\sigma}} \int \phi\left(\frac{y_i - x_i^T\beta}{\hat{\sigma}}\right) d\hat{G}_{\text{CV}}(\beta)$$
closely approximates the truth, with accuracy
comparable to both the 3-component EM algorithm initialized with
true parameters and the NPMLE with true
$\sigma$. Figure~\ref{fig:3comp_density-intro} illustrates this with
ridgeline plots showing conditional densities of $y$ for different
covariate values $w$.

\subsubsection{Simulation: Continuous
  Distribution}\label{twocircsection}

Figure~\ref{fig:concentric_circle_noisy} displays data generated from our model with continuous measure $G^*$ and $\sigma = 0.5$, where $G^*$ is uniformly distributed over two concentric circles:
\begin{equation}
	G^* = 0.5 \cdot \text{Unif}\{\beta \in \mathbb{R}^2 : \|\beta\| = 1\} + 0.5 \cdot \text{Unif}\{\beta \in \mathbb{R}^2 : \|\beta\| = 2\}   
\end{equation}

Our cross-validation yielded $\hat{\sigma} = 0.6050$, which we used to
compute $\hat{G}_{\text{CV}}$ (with $K = [-10,
10]^2$). Figure~\ref{fig:concentric_circle_sigmaCV} compares $G^*$ and
$\hat{G}_{\text{CV}}$ (dot size proportional to mixing probability),
while Figure~\ref{fig:concentric_circle_sigma} compares $G^*$ with
$\hat{G}$ (NPMLE computed with known $\sigma$).

\begin{figure}[!htbp]
	\centering
	\begin{subfigure}[t]{0.32\textwidth}
		\caption{}
		\includegraphics[width = \textwidth, valign = t]{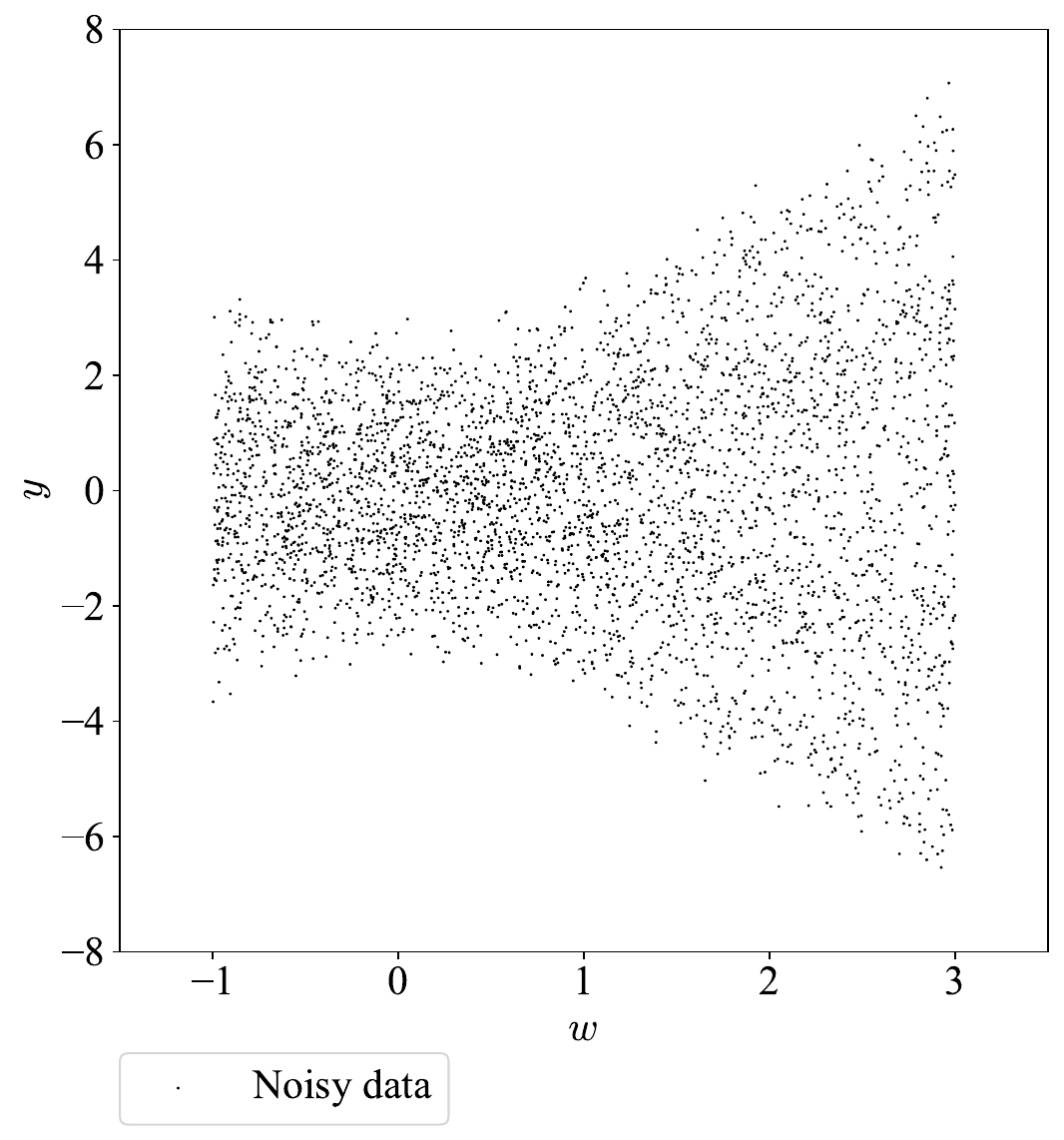}
		\label{fig:concentric_circle_noisy}
	\end{subfigure}
	\begin{subfigure}[t]{0.32\textwidth}
		\caption{}
		\includegraphics[width = \textwidth, valign = t]{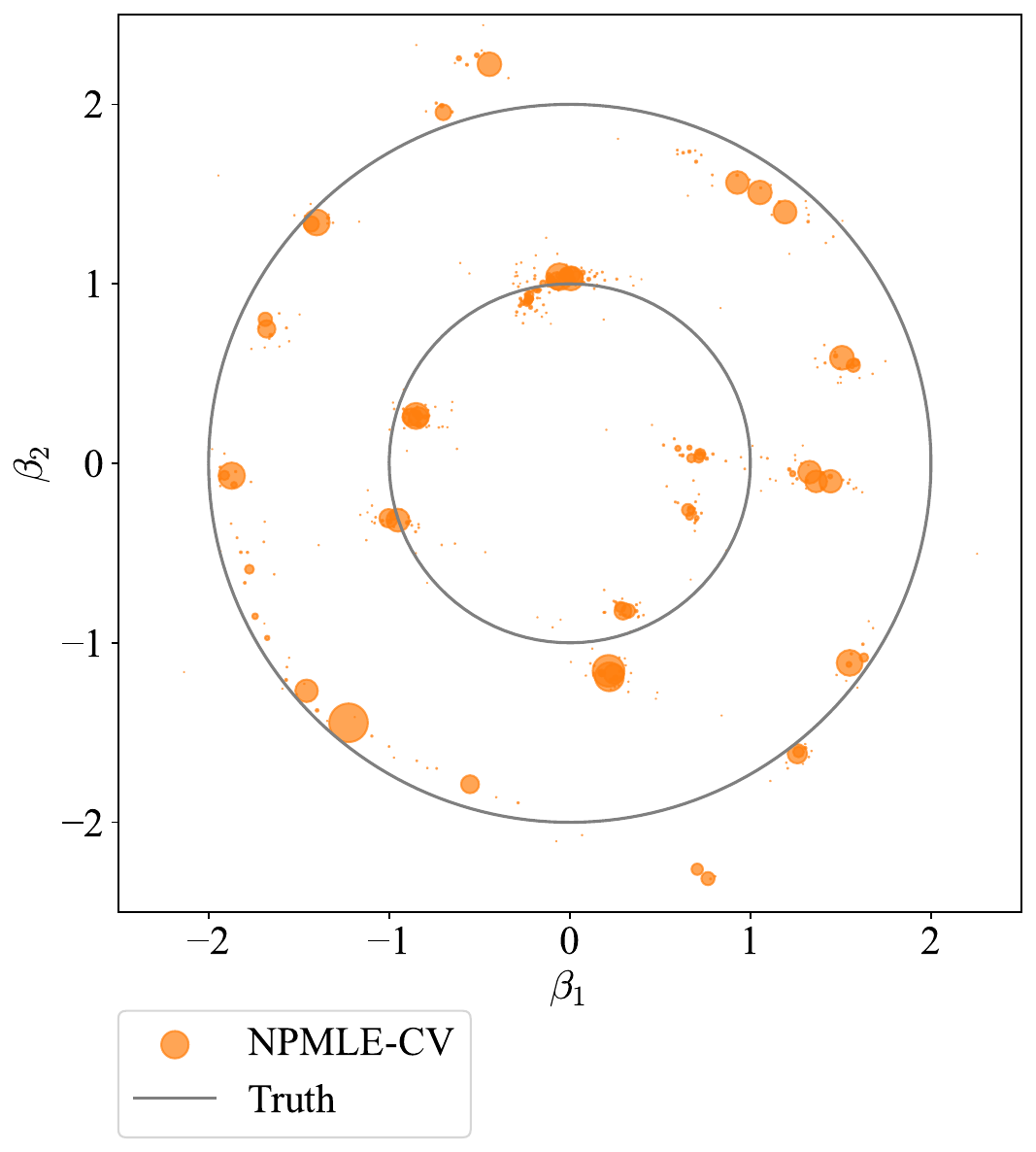}
		\label{fig:concentric_circle_sigmaCV}
	\end{subfigure}
	\begin{subfigure}[t]{0.32\textwidth}
		\caption{}
		\includegraphics[width = \textwidth, valign = t]{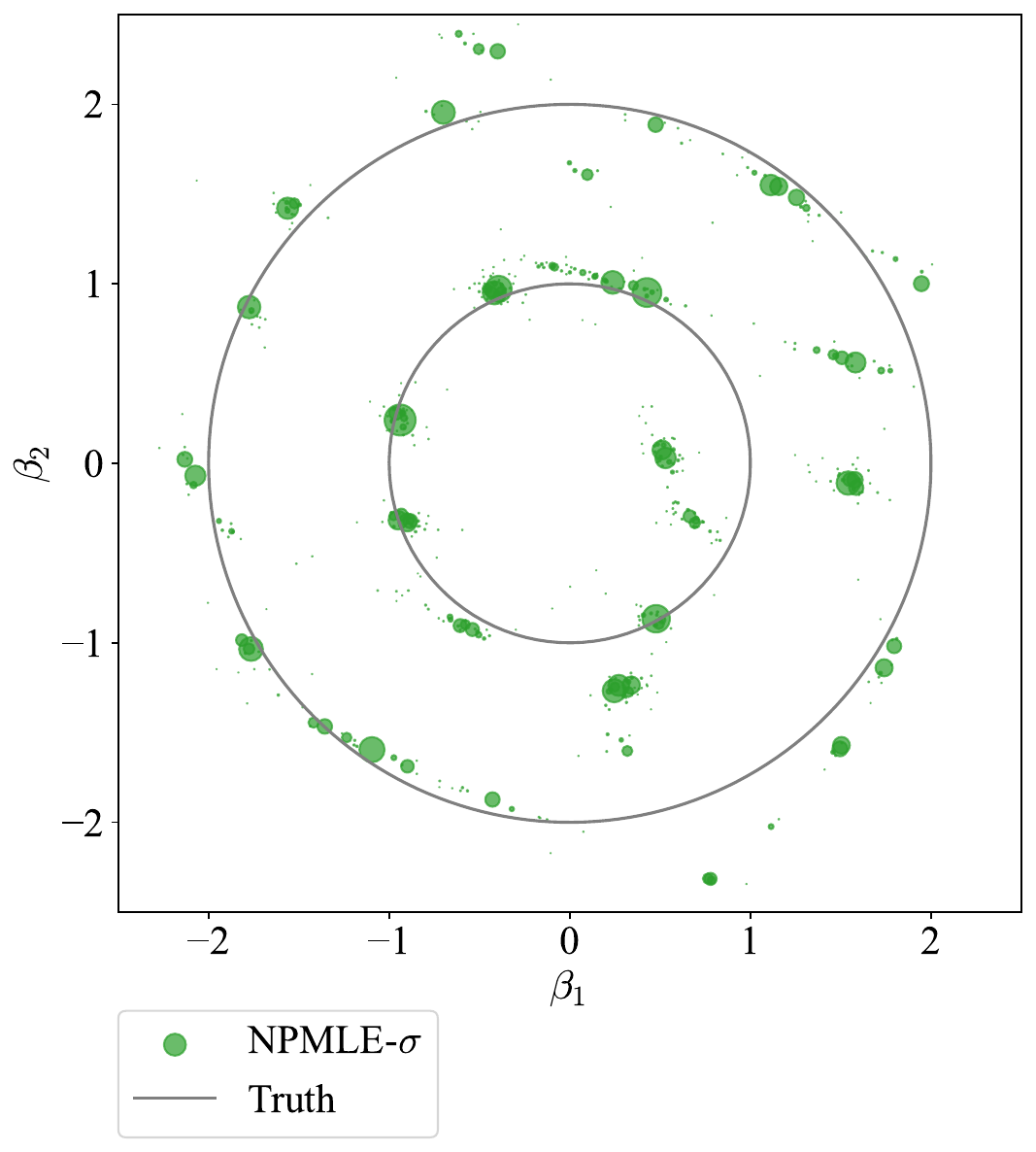}
		\label{fig:concentric_circle_sigma}
	\end{subfigure}

	\vspace{1em} 

	\begin{subfigure}[t]{0.65\textwidth}
		\caption{}
		\includegraphics[width = \textwidth, valign = t]{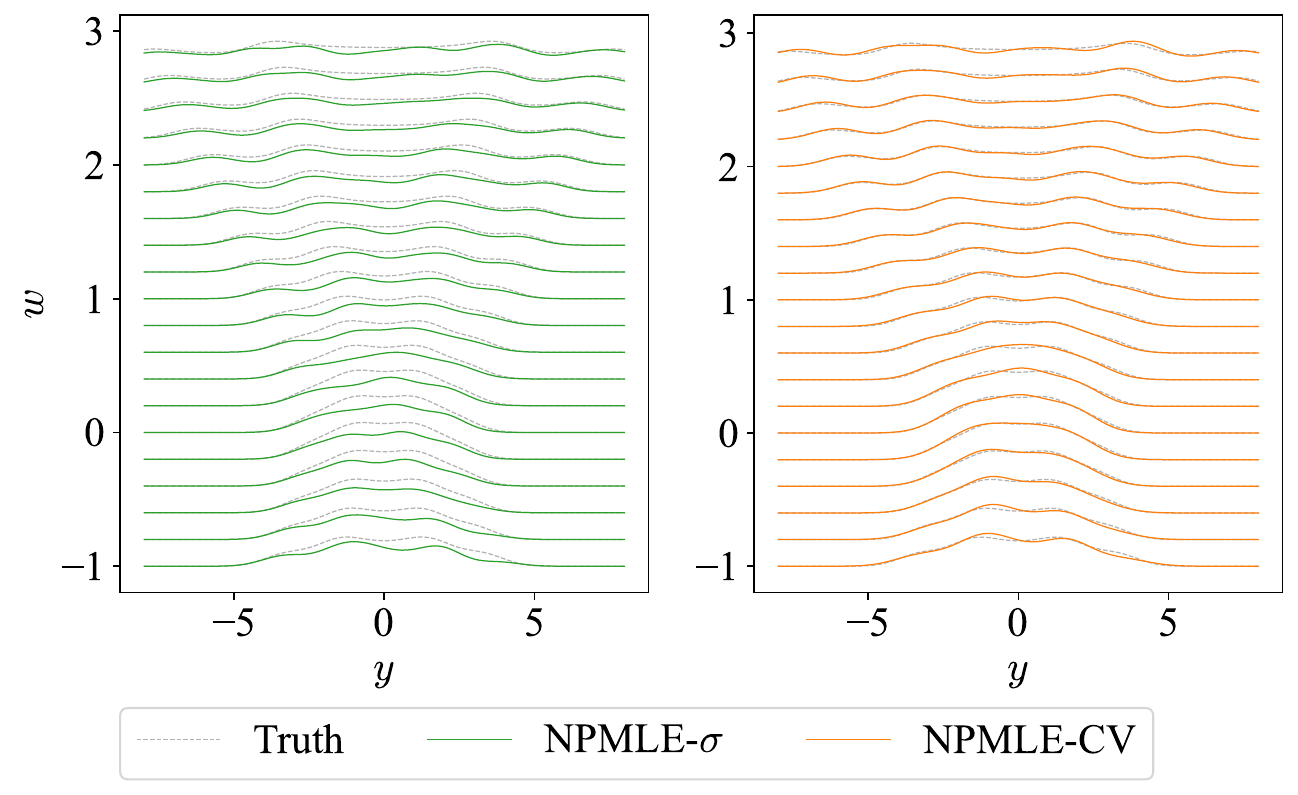}
		\label{fig:concentric_circle_ridgeline}
	\end{subfigure}
	\begin{subfigure}[t]{0.32\textwidth}
		\caption{}
		\includegraphics[width = \textwidth, valign = t]{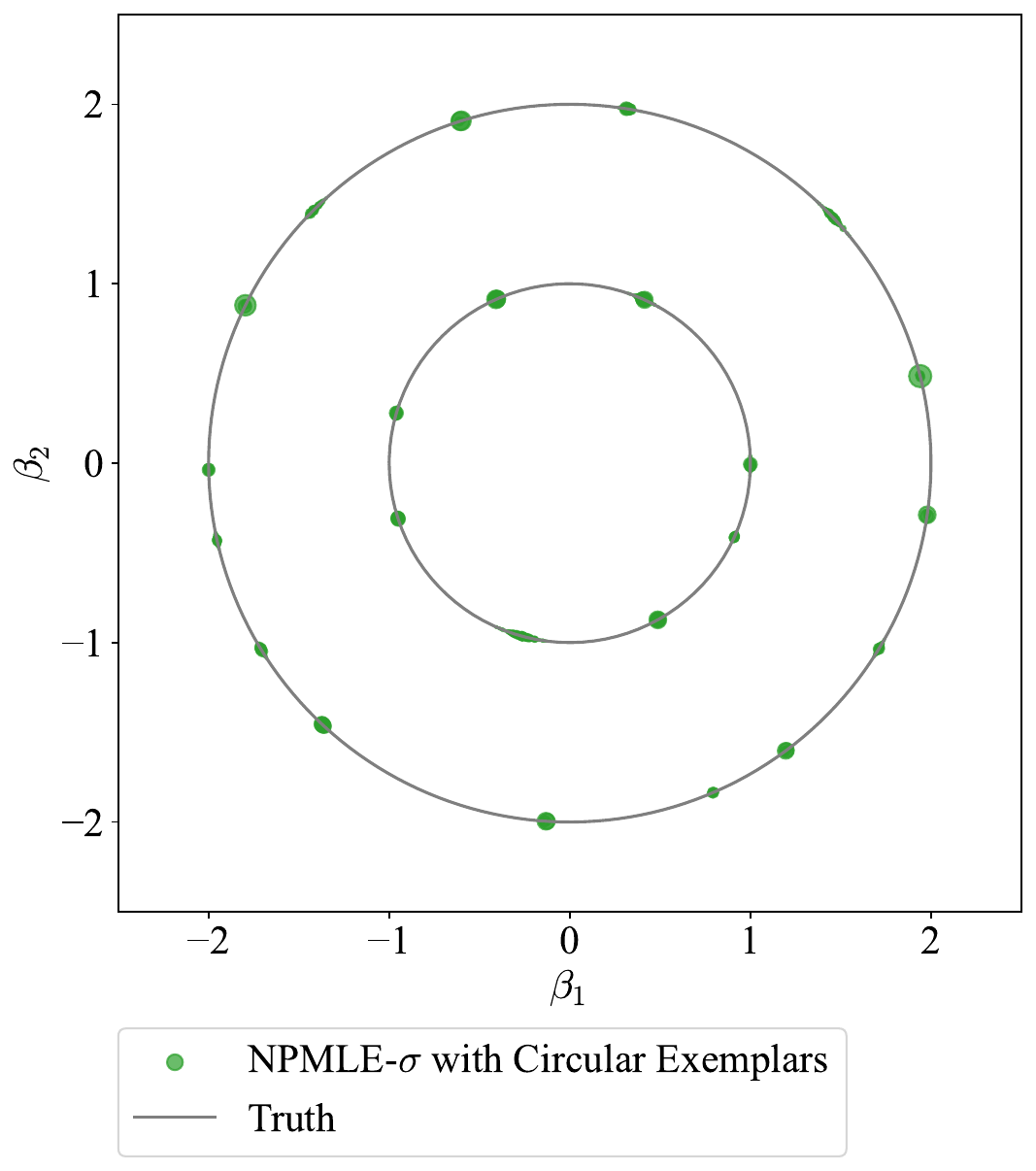} 
		\label{fig:concentric_circle_trueExemplars}
	\end{subfigure}

	\vspace{-1em} 

\caption{Continuous mixing measure $G^*$:
\subref{fig:concentric_circle_noisy} Data ($n=4000$);
\subref{fig:concentric_circle_sigmaCV} $G^*$ and $\hat{G}_{\text{CV}}$;
\subref{fig:concentric_circle_sigma} $G^*$ and $\hat{G}$;
\subref{fig:concentric_circle_ridgeline} Ridgeline plots comparing conditional densities;
\subref{fig:concentric_circle_trueExemplars} NPMLE with exemplars
uniform over support of $G^*$. 
In \subref{fig:concentric_circle_sigmaCV}, \subref{fig:concentric_circle_sigma}, and \subref{fig:concentric_circle_trueExemplars}, marker areas are proportional to probability weights.}
	\label{fig:concentric_circle_noisyANDmeasure}
\end{figure}

Since $G^*$ is continuous, $\hat{G}_{\text{CV}}$ naturally contains
many atoms that approximately trace the two circular supports of
$G^*$. $\hat{G}_{\text{CV}}$ is likely consistent (Theorem
\ref{thm:weakconsistency} shows this when $\sigma$ is known) but
previous results 
for Gaussian location mixtures suggest logarithmically slow
convergence rates, explaining the imperfect
approximation. Nevertheless,
Figure~\ref{fig:concentric_circle_ridgeline} shows that the estimated
conditional densities closely approximate the true density. 

For each observation $i = 1,\ldots,n$, our approach produces an
estimate $\hat{\beta}^{i}_{\text{EB}}$ (defined in
\eqref{postmean.est}) approximating $\hat{\beta}^i_{\text{OB}}$
(defined in \eqref{postmean}). Figure
\ref{fig:concentric_circle_EBversusOB-intro} confirms this
approximation works well by plotting their coordinates separately.

\begin{figure}[!ht]
	\centering
	\begin{subfigure}[t]{0.4\textwidth}
		\caption{}
		\includegraphics[width = \textwidth,valign = t]{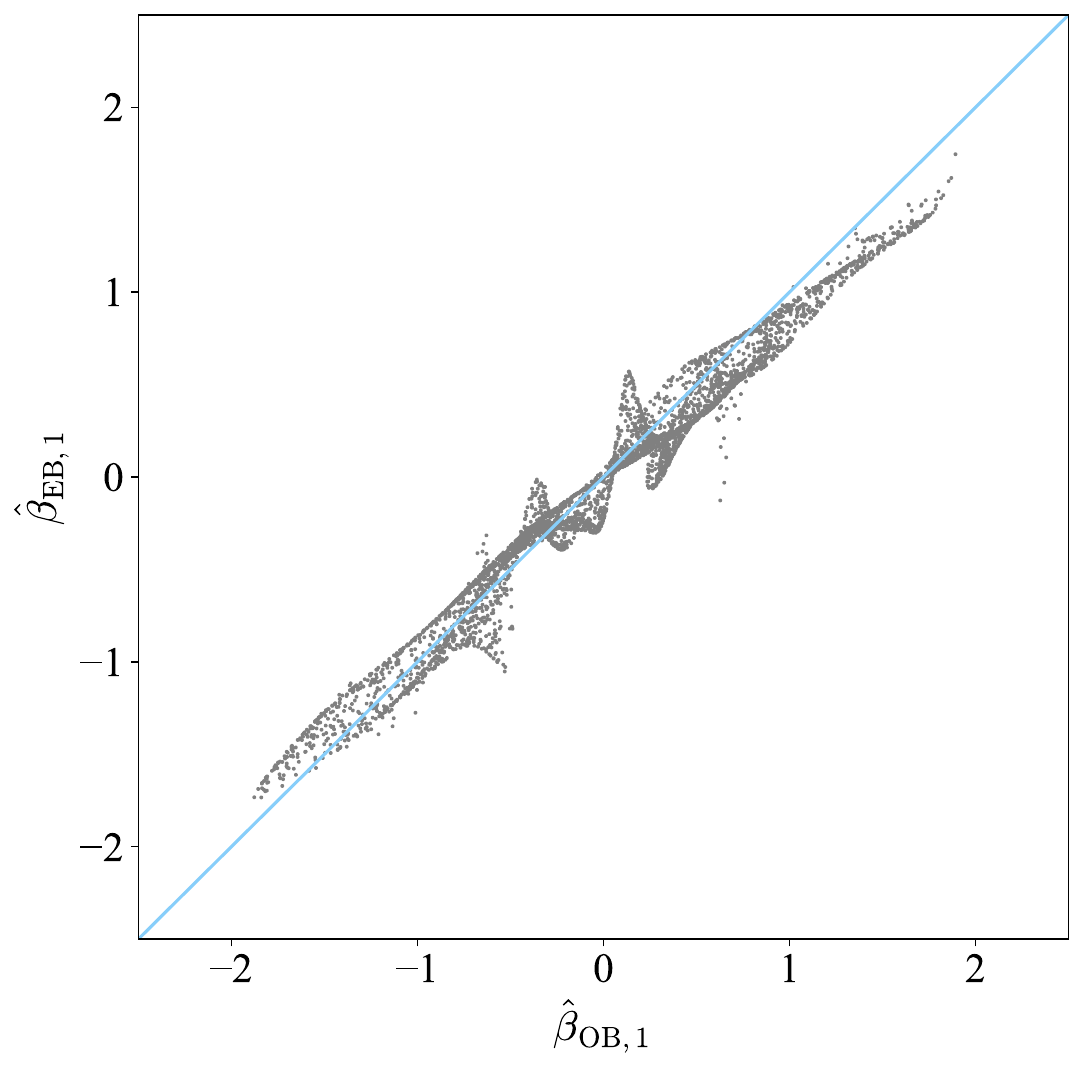}
		\label{fig:concentric_circle_EBversusOB_1-intro}
		\label{fig:eb-ob-1}
	\end{subfigure}
	\begin{subfigure}[t]{0.4\textwidth}
		\caption{}
		\includegraphics[width = \textwidth, valign = t]{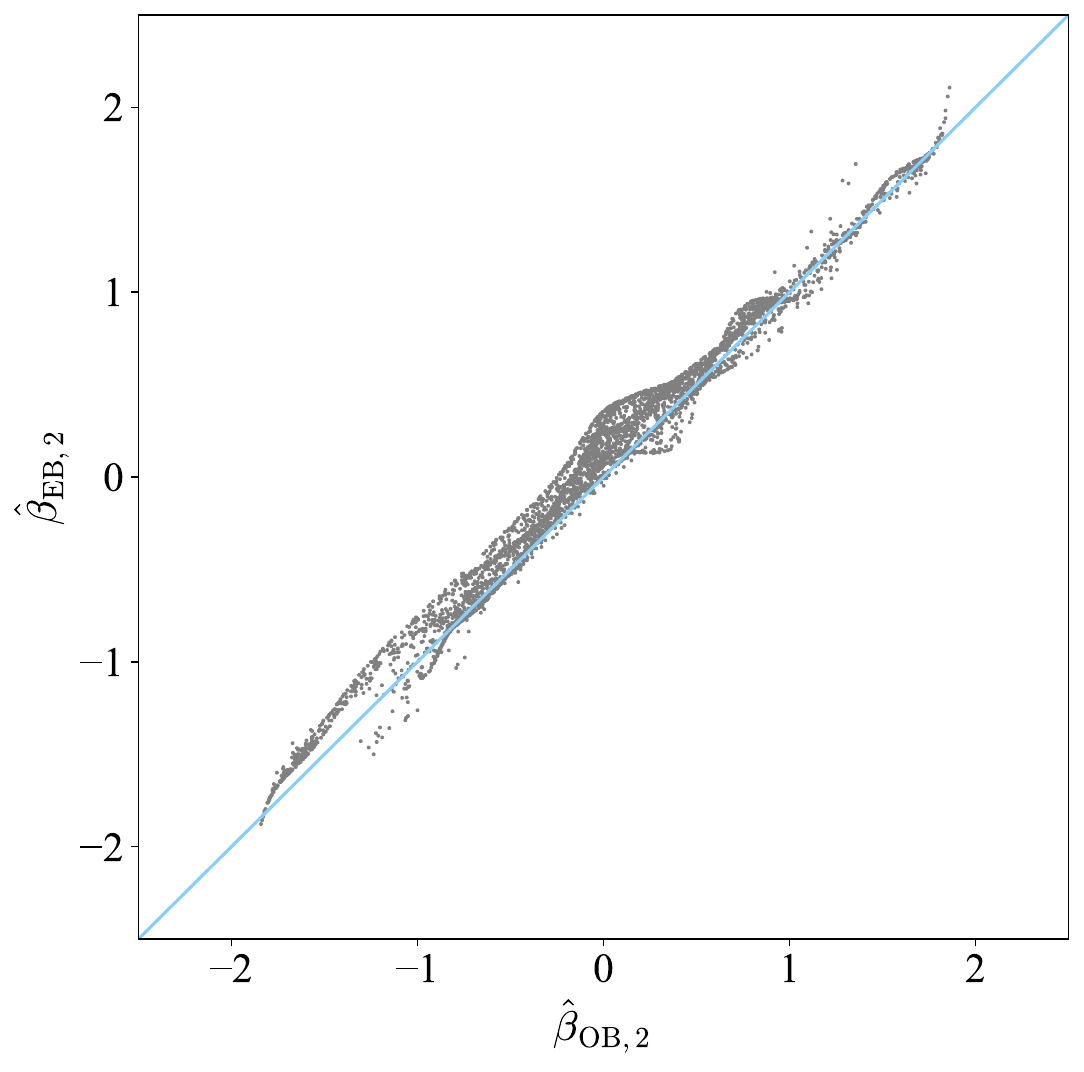}
		\label{fig:concentric_circle_EBversusOB_2-intro}
		\label{fig:eb-ob-2}
	\end{subfigure}
	\vspace{-2em}
	\caption{Plots of $\hat{\beta}_{\EB}^i$ against
          $\hat{\beta}_{\OB}^i$ for the setting in Subsection
          \ref{twocircsection}: \subref{fig:eb-ob-1} and
          \subref{fig:eb-ob-2} show the first (intercept) and the
          second (slope) component of $\hat{\beta}_{\EB}^i$,
          $\hat{\beta}_{\OB}^i$ respectively.}   
	\label{fig:concentric_circle_EBversusOB-intro}
\end{figure}

When $G^*$ is continuous, the discrete $\hat{G}_{\text{CV}}$ will not
be visually close to $G^*$ (see e.g., Figure 
\ref{fig:concentric_circle_trueExemplars} where NPMLE is shown with $M
= 4n$ exemplars regularly spaced on the \textit{true} 
support). For better estimating $G^*$ in such cases,
more information (e.g., in the form of priors) may be necessary,
as explored by \citet{chae2023likelihood} and
\citet{berenfeld2022estimating}.

\subsubsection{Simulation: Mixtures with Sinusoid Covariates}

\label{sec:sinusoid_simulation}

We now examine a higher dimensional case with $p = 7$ and $n =
10,000$. Data are generated according to model
\eqref{model:mixlin_i_1} and \eqref{model:mixlin_i_2}, with
covariates: 
\[
x_i = \left(1, \cos(2\pi f_1 w_i), \sin(2\pi f_1 w_i), \cos(2\pi f_2 w_i), \sin(2\pi f_2 w_i), \cos(2\pi f_3 w_i), \sin(2\pi f_3 w_i) \right)^\top,
\]
where $f_1 = 1, f_2 = \sqrt{5}$, $f_3 = \sqrt{11}$, and $w_i \sim
\text{Uniform}[0,1]$ independently. The data follow a mixture of $k=4$
linear regression models with noise level $\sigma = 0.75$ and equal
mixing probabilities $\pi_l = 1/4$ for $l = 1,\ldots,4$
(Figure~\ref{fig:sinusoid_true}). For each component $c = 1,2,3,4$,
all elements of the regression coefficient vector are drawn
independently from $N(0,4)$.

\begin{figure}[!htbp]
	\centering
	\begin{subfigure}[t]{0.4\textwidth}
		\subcaption{}
		\includegraphics[width = \textwidth,valign = t]{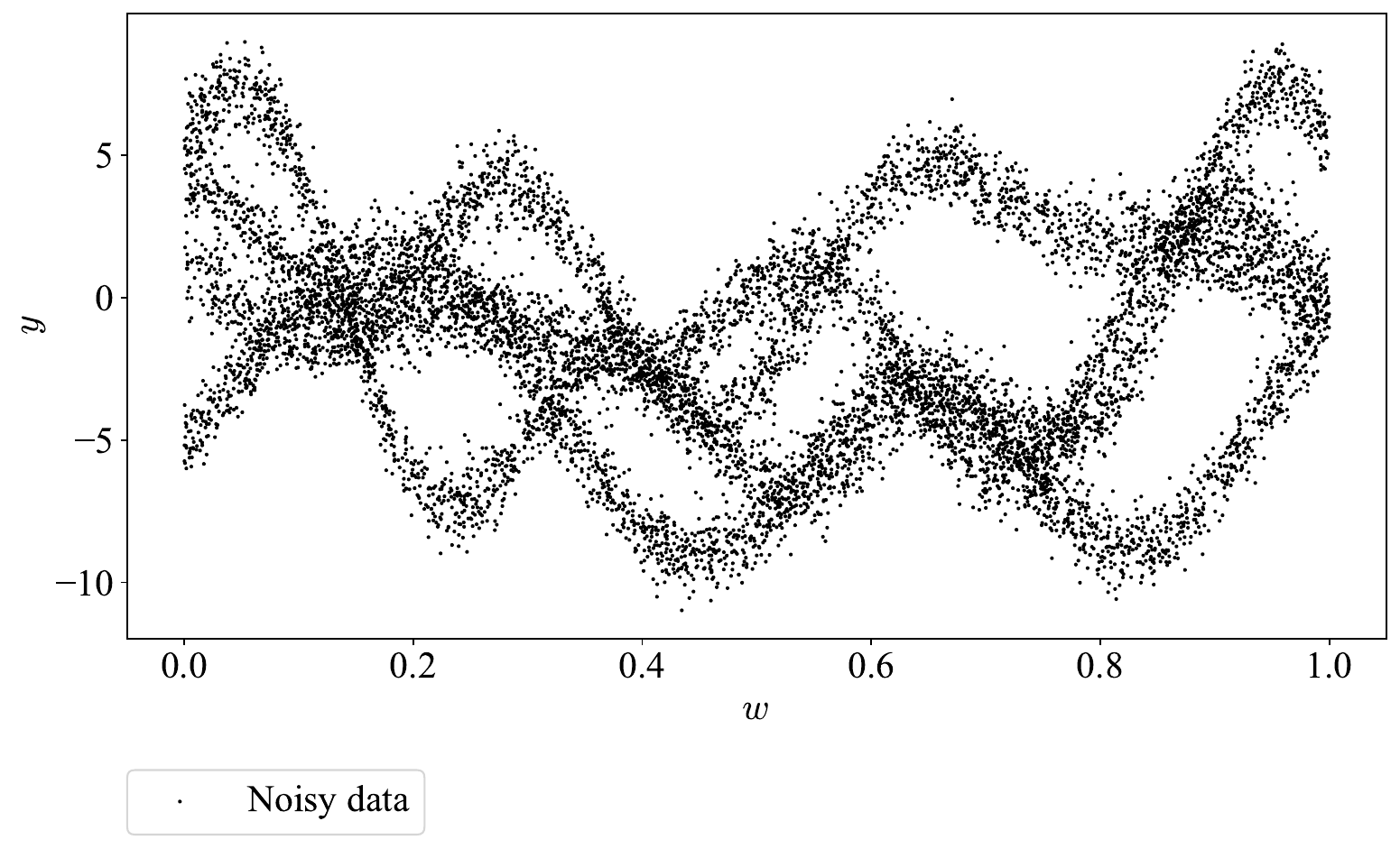}
		\label{fig:sinusoid_noisy}
	\end{subfigure}
	\begin{subfigure}[t]{0.4\textwidth}
		\subcaption{}
		\includegraphics[width = \textwidth,valign = t]{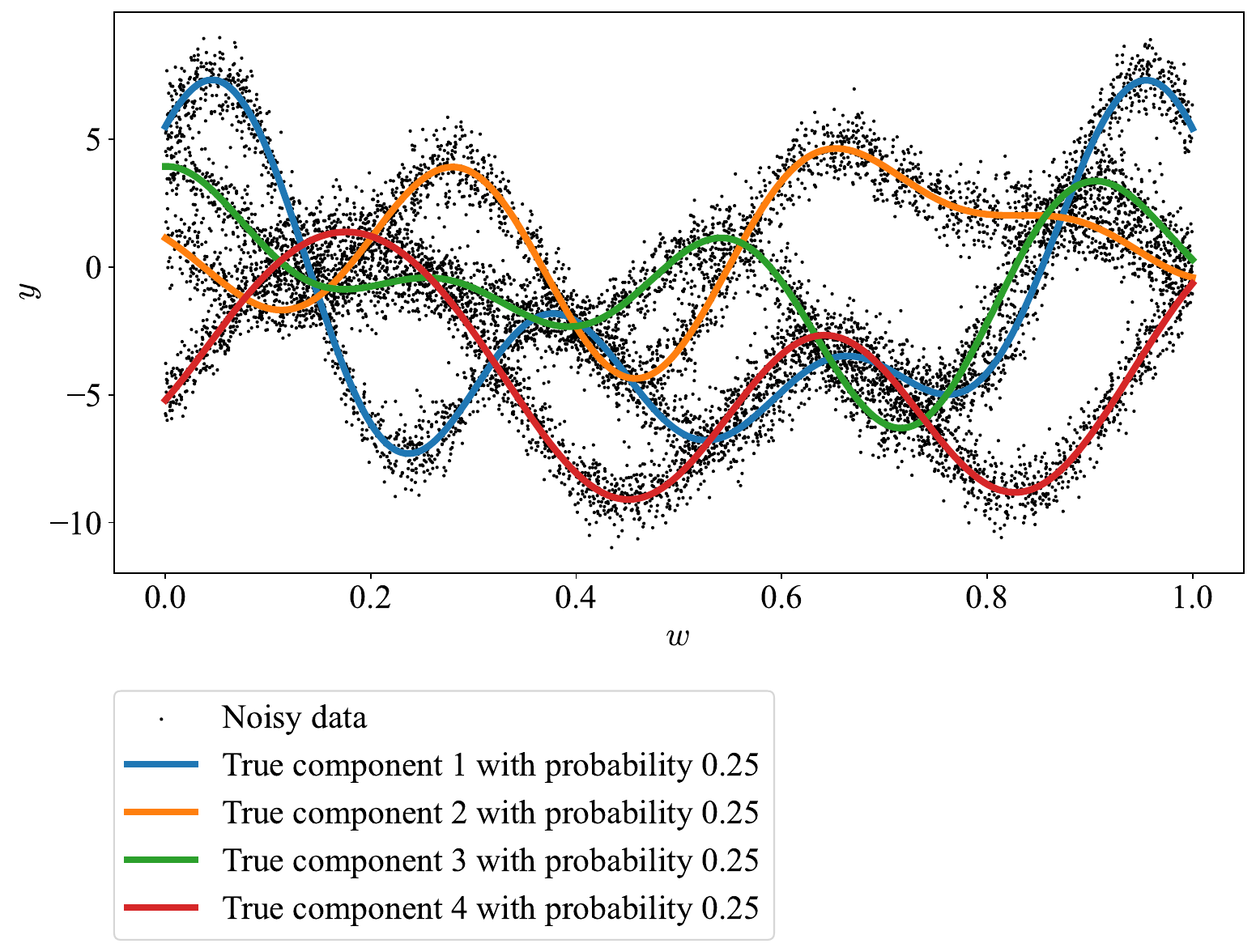}
		\label{fig:sinusoid_true}
	\end{subfigure}
	\begin{subfigure}[t]{0.4\textwidth}
		\vspace{-1em}
		\subcaption{}
		\includegraphics[width = \textwidth,valign = t]{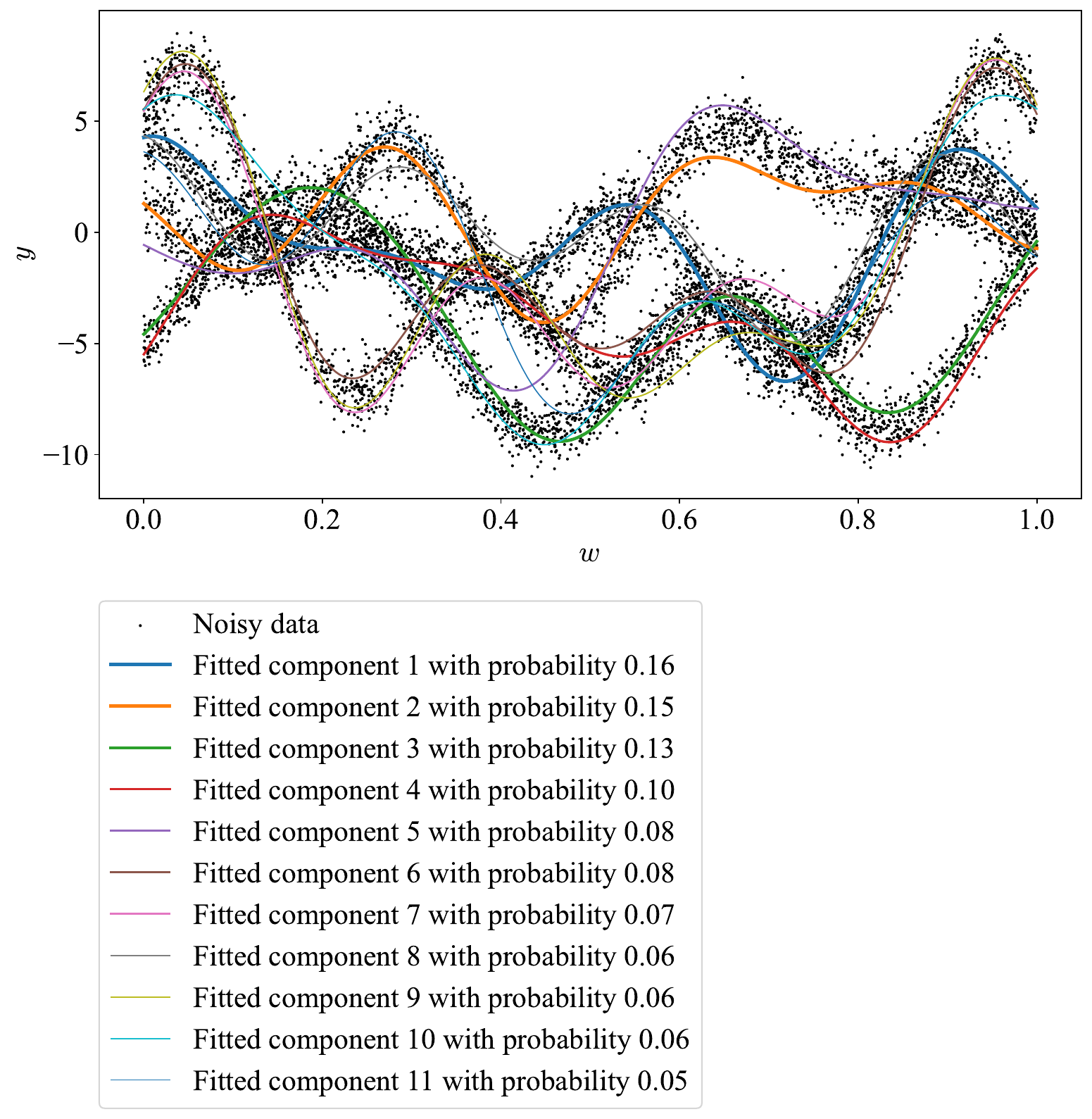}
		\label{fig:sinusoid_fitted_NPMLEsigma}
	\end{subfigure}
	\begin{subfigure}[t]{0.4\textwidth}
		\vspace{-1em}
		\subcaption{}
		\includegraphics[width = \textwidth,valign = t]{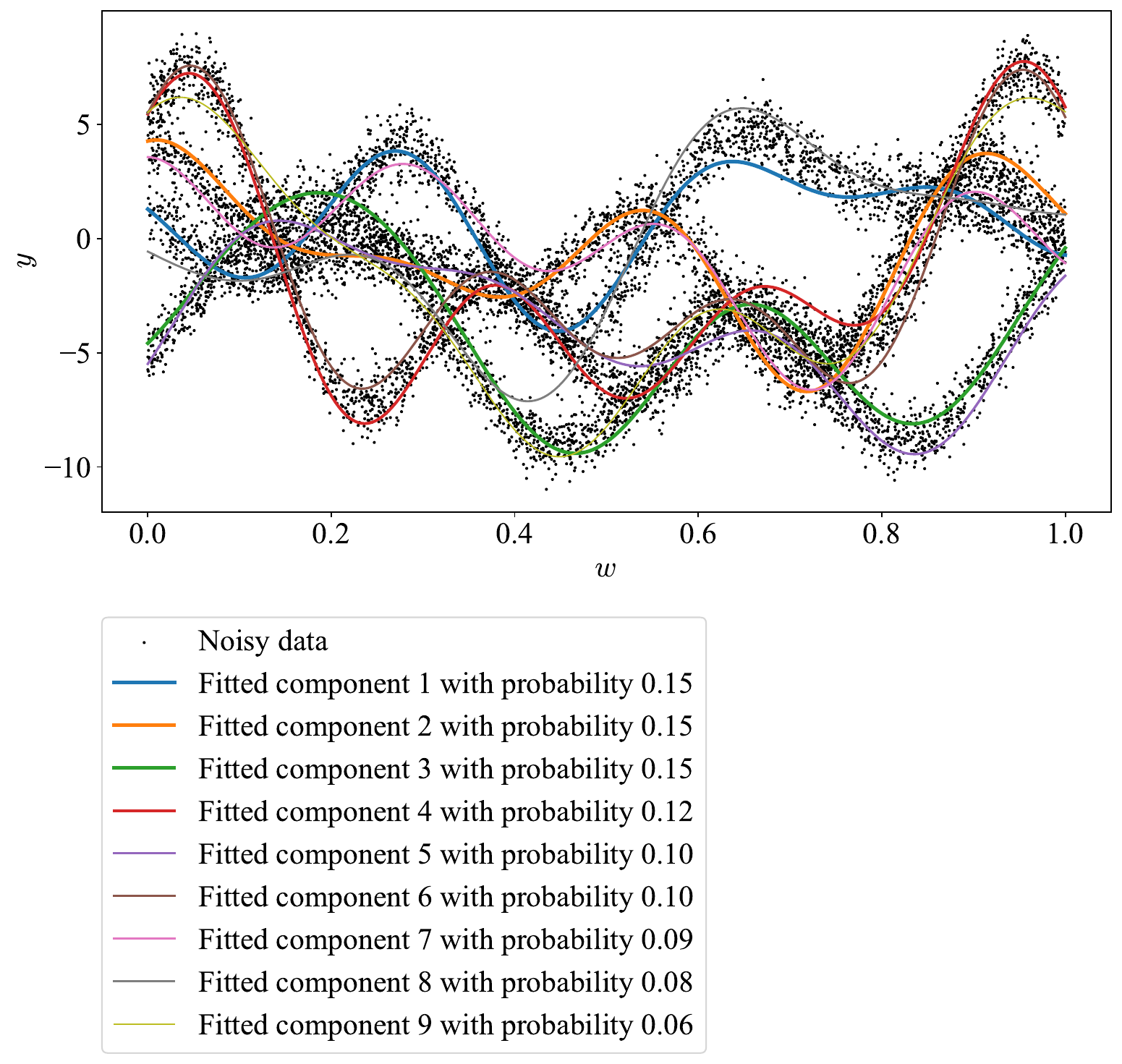}
		\label{fig:sinusoid_fitted}
	\end{subfigure}
	\vspace{-1em}
\caption{\subref{fig:sinusoid_noisy} Data ($n=10,000, \sigma=0.75$);
  \subref{fig:sinusoid_true} True components;
  \subref{fig:sinusoid_fitted_NPMLEsigma} NPMLE with true $\sigma$ and
  BIC; \subref{fig:sinusoid_fitted} NPMLE with CV-selected $\hat{\sigma}=0.9025$ and BIC.}
	\label{fig:sinusoid_fitted_all}
\end{figure}

We computed the NPMLE using Algorithm \ref{exemplaralgo} and pruned components
using BIC, with results shown in
Figure~\ref{fig:sinusoid_fitted_all} (see also Table
\ref{table:sinusoid_results} in the supplement). While the NPMLE
has several components, the BIC-pruned solution is 
parsimonious. The four components with highest mixing
probabilities in Figure~\ref{fig:sinusoid_fitted} correspond exactly
to the true components.

\subsubsection{Simulation: Mixtures with Change-Point Covariates}
\label{sec:changepoint_simulation}

We analyze data ($p = 6$, $n = 10,000$) generated from a mixture of
linear regression models with step covariates. Let $s_j = j/p$ for $j
= 1,\ldots,p-1$.  For each $i=1,\ldots,n$, we sample $w_i \sim
\text{Uniform}[0,1]$ and construct covariates as:  $x_{ij} =
\mathbbm{1}\{j = 1\} + \mathbbm{1}\{j \neq 1\} \cdot \mathbbm{1}\{w_i
\geq s_{j}\}$ for $j = 1,\ldots,p$.  The true model has $k=4$
components with equal mixing probabilities $\pi_l = 1/4$ and
regression coefficients drawn independently from
$N(0,4)$. Figure~\ref{fig:change-point_fitted-intro} (see also Table
\ref{table:change-point_results} in the supplement) demonstrates that
our method performs well on this example.     

\begin{figure}[!htbp]
	\centering
	\begin{subfigure}[t]{0.4\textwidth}
		\subcaption{}
		\includegraphics[width = \textwidth,valign = t]{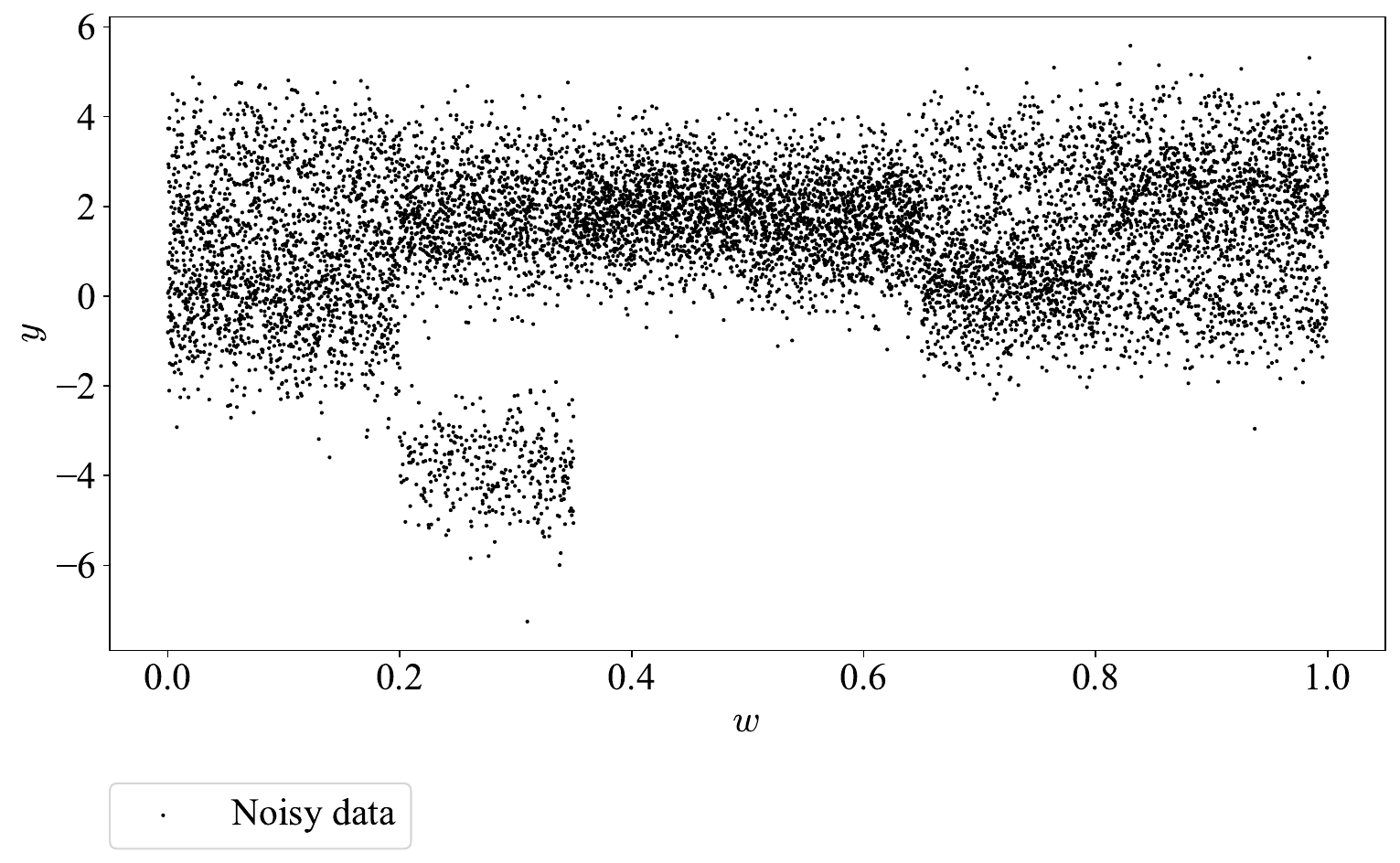}
		\label{fig:change-point_noisy}
	\end{subfigure}
	\begin{subfigure}[t]{0.4\textwidth}
		\subcaption{}
		\includegraphics[width = \textwidth,valign = t]{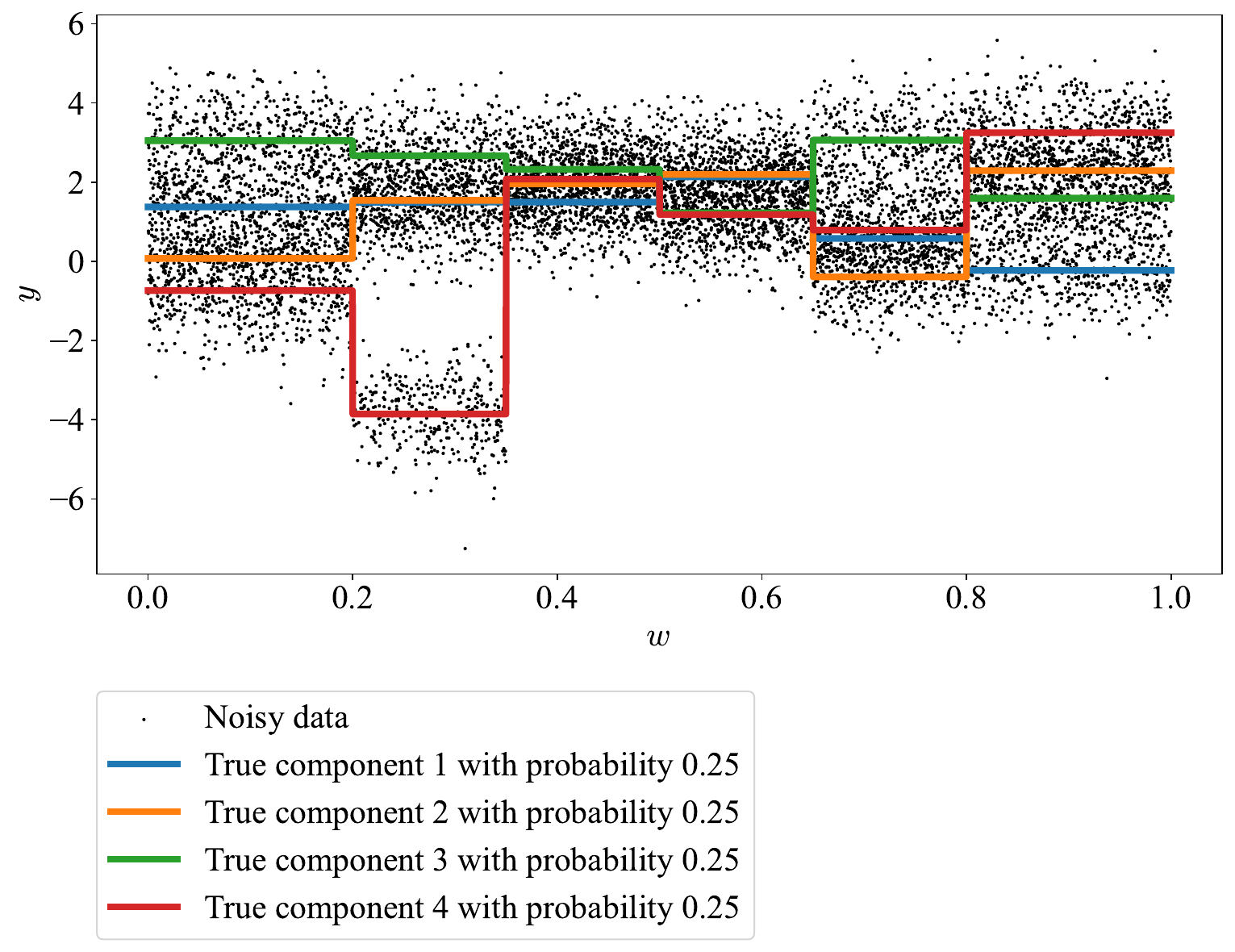}
		\label{fig:change-point_true}
	\end{subfigure}
	\begin{subfigure}[t]{0.4\textwidth}
		\vspace{-1em}
		\subcaption{}
		\includegraphics[width = \textwidth,valign = t]{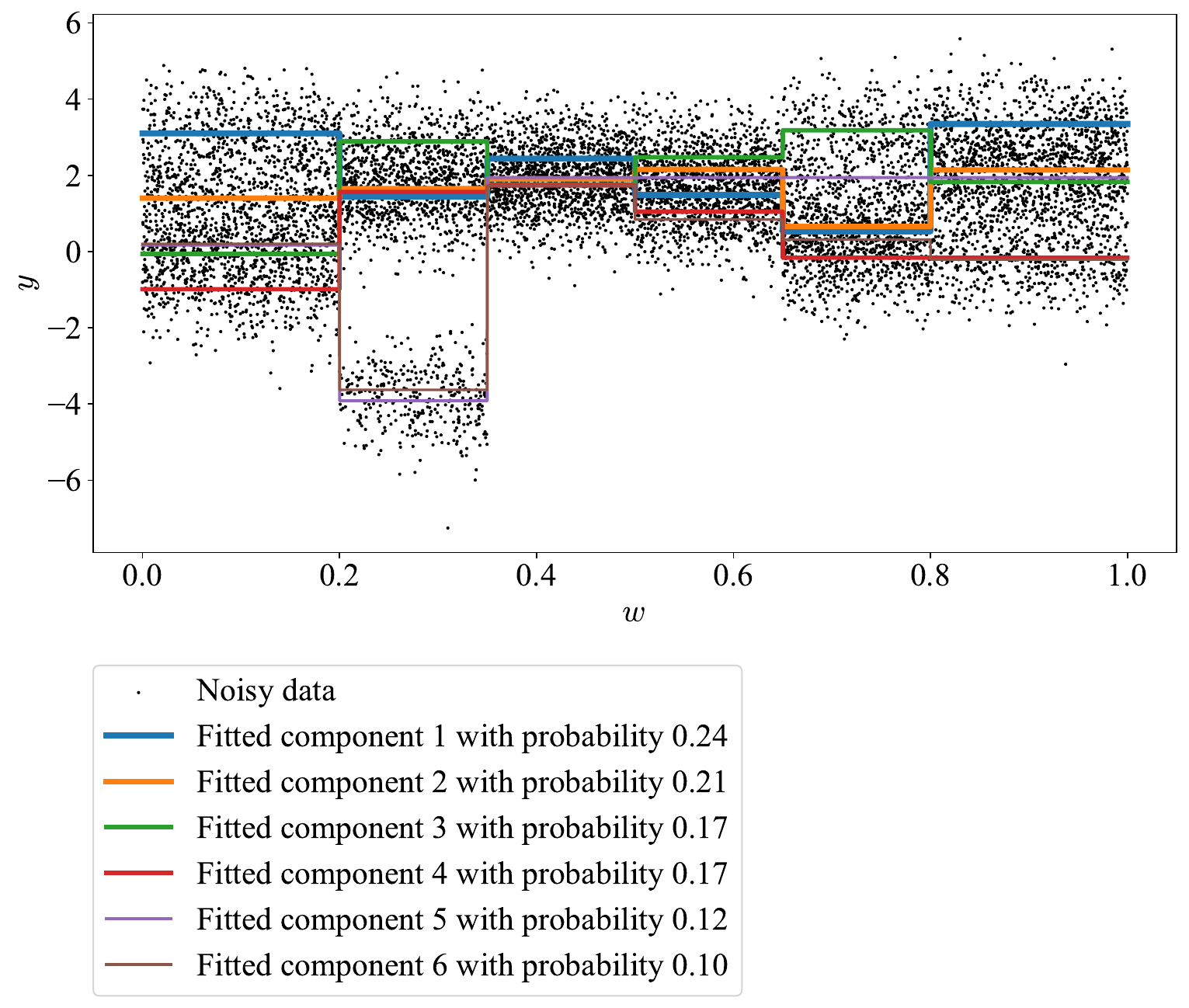}
		\label{fig:change-point_fitted_NPMLEsigma}
	\end{subfigure}
	\begin{subfigure}[t]{0.4\textwidth}
		\vspace{-1em}
		\subcaption{}
		\includegraphics[width = \textwidth,valign = t]{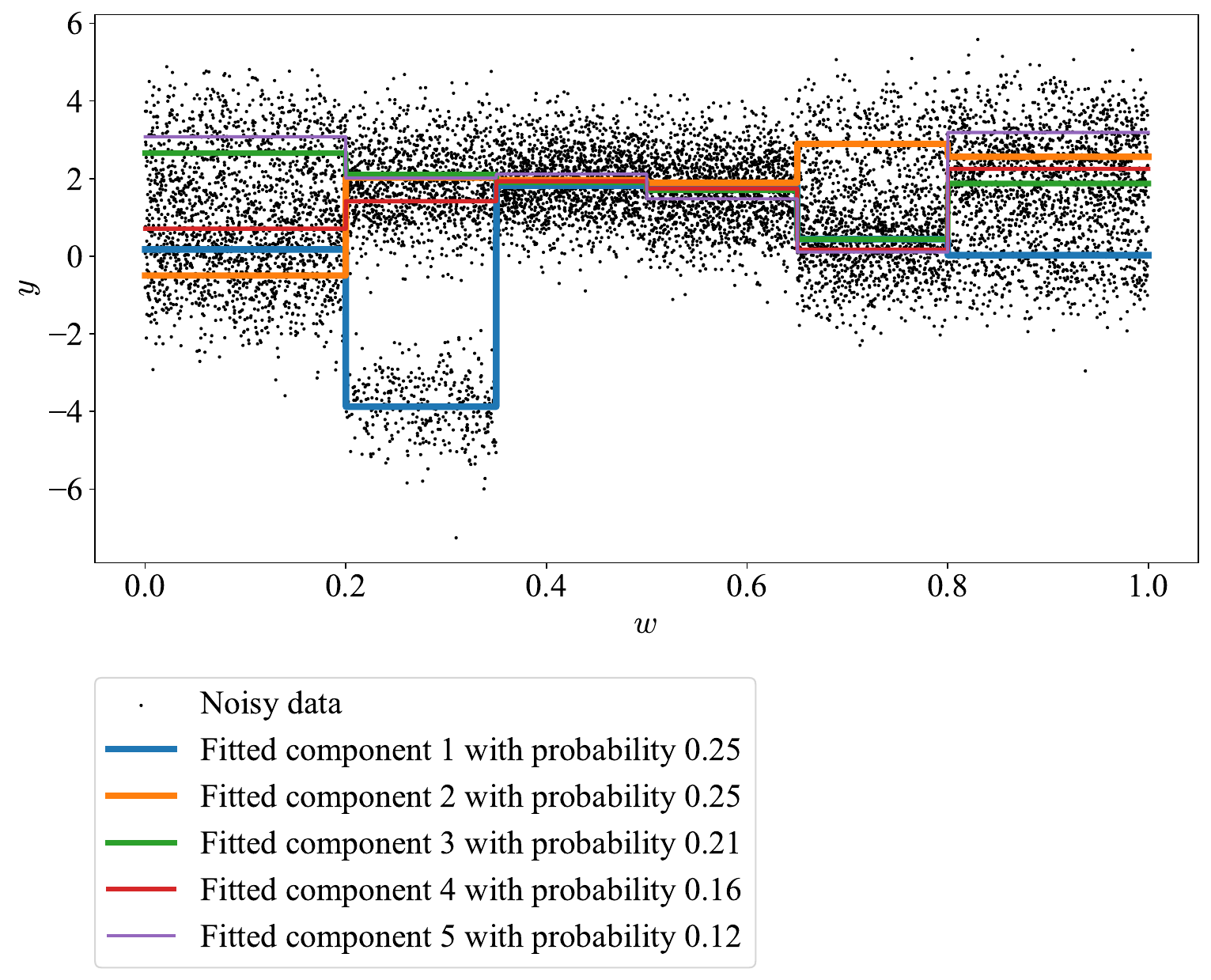}
		\label{fig:change-point_fitted}
	\end{subfigure}
	\vspace{-1em}
	\caption{\subref{fig:change-point_noisy} Data points ($n= 10,000$, $\sigma = 0.75$); \subref{fig:change-point_true} True components; \subref{fig:change-point_fitted_NPMLEsigma} Fitted mixture with true $\sigma$ and BIC; \subref{fig:change-point_fitted} Fitted mixture with $\hat{\sigma} = 0.9025$ selected by CV and BIC.}
	\label{fig:change-point_fitted-intro}
\end{figure}

Cross-validation selected $\hat{\sigma} = 0.9025$ in
both this and the sinusoid example (computed as
$\exp(\log(0.1) + 2.2)$). Despite exceeding the true $\sigma = 0.75$,
the ridgeline plots in Figures~\ref{fig:sinusoid_density} and
\ref{fig:change-point_density} show that the NPMLE with $\hat{\sigma}$
produces superior conditional density estimates compared to those
using the true $\sigma$, particularly in the change-point example.

\begin{figure}[!htb]
	\centering
	\begin{subfigure}[t]{0.48\textwidth}
		\centering
				\caption{}
		\includegraphics[width=\textwidth]{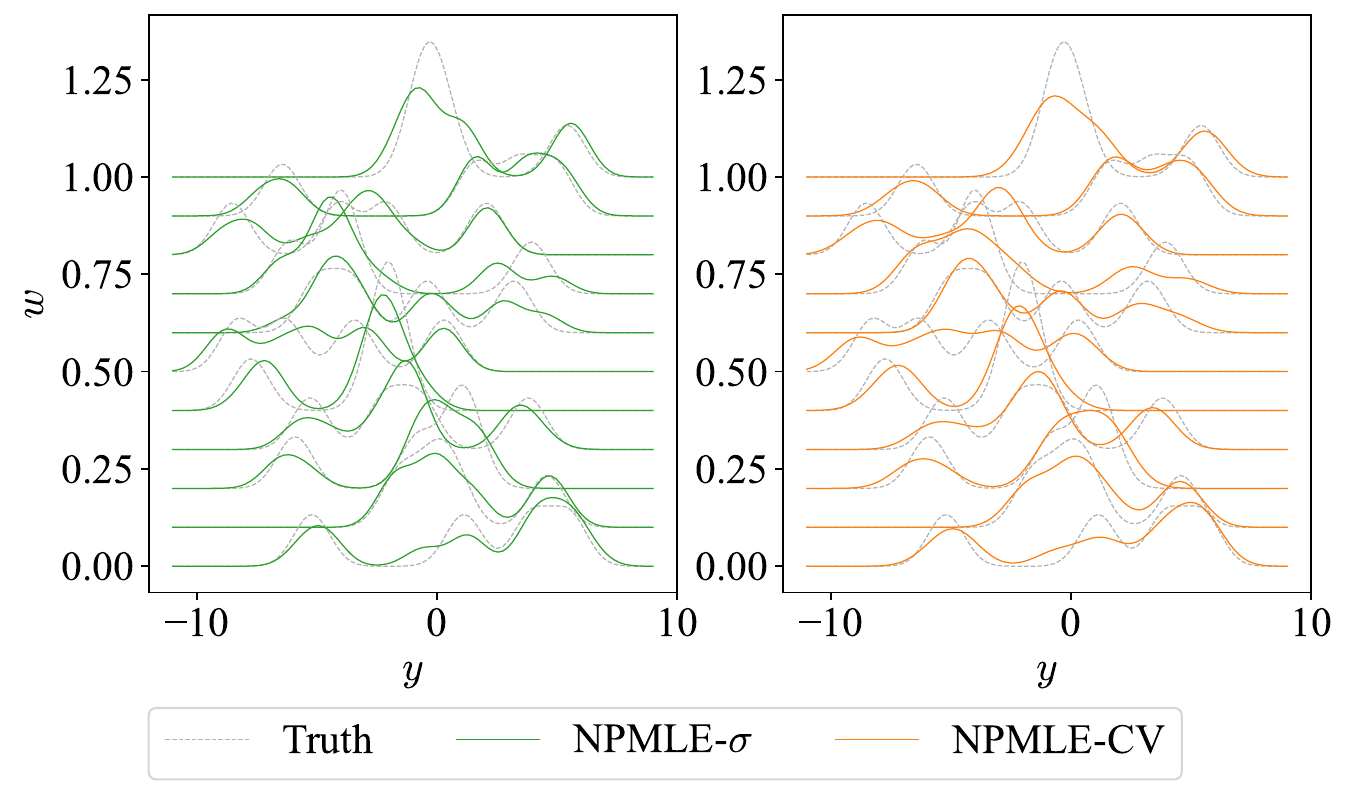}
		\label{fig:sinusoid_density}
	\end{subfigure}
	\hfill
	\begin{subfigure}[t]{0.48\textwidth}
		\centering
				\caption{}
		\includegraphics[width=\textwidth]{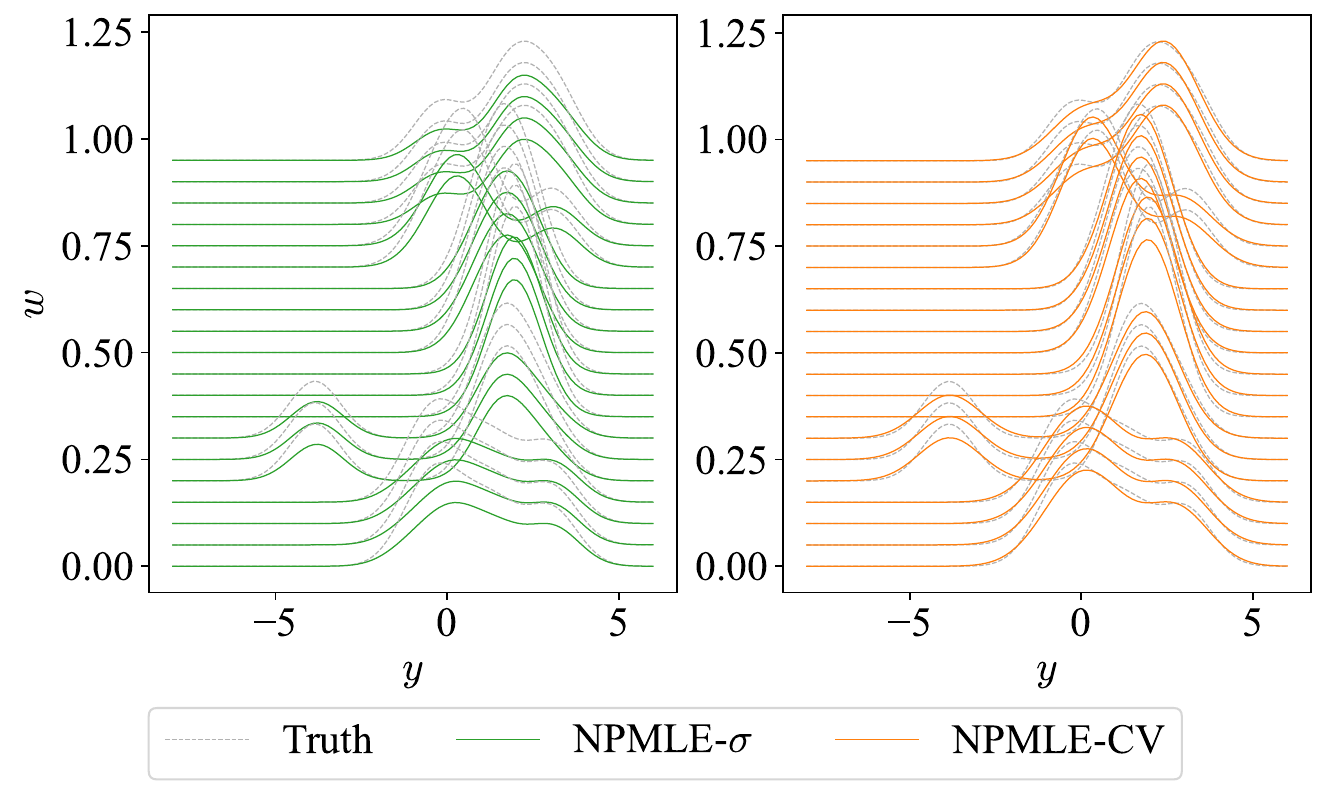}
		\label{fig:change-point_density}
	\end{subfigure}
	\vspace{-3em}
	\caption{Ridgeline plots for \subref{fig:sinusoid_density} sinusoid example and \subref{fig:change-point_density} change-point examples.}
	\label{fig:combined_ridgeline_plots}
\end{figure}

\subsection{Real Data Case Studies}

\subsubsection{Real Data: Music Tone Perception}

We apply our method to the music tone perception data collected by
\citet{cohen1980inharmonic}, which has been analyzed in several
studies \citep{de1989mixtures, viele2002modeling, yao2015mixtures} and
is available in the R package \texttt{mixtools}
\citep{R2009mixtools}. The dataset contains 150 observations from
experiments where a musician was asked to tune an adjustable tone to
the octave above a fundamental tone with stretched overtones. The
covariate $s$ represents the stretching ratio of overtones to the
fundamental tone, while the response $y$ is the ratio of the adjusted
tone to the fundamental. Two competing music perception theories
exist regarding the 
relationship between $y$ and $s$: one predicts a consistent $y = 2$
ratio, while the other suggests $y = s$. 

\begin{figure}[!htb]
	\centering
	\begin{subfigure}[t]{0.48\textwidth}
		\caption{}
		\label{fig:tonedata_withoutBIC}
		\includegraphics[width = \textwidth,valign = t]{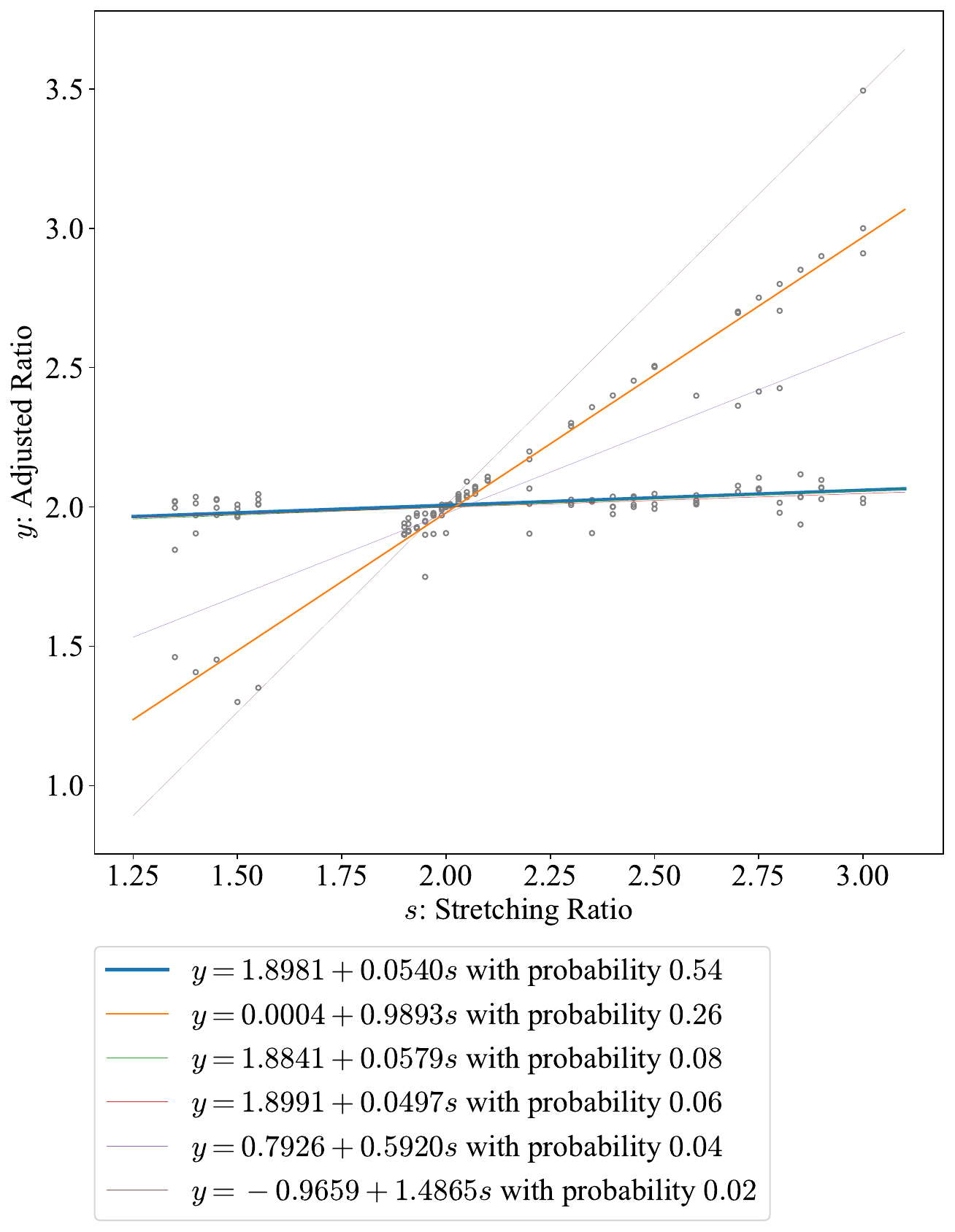}
	\end{subfigure}
	\begin{subfigure}[t]{0.48\textwidth}
		\caption{}
		\label{fig:tonedata_afterBIC}
		\includegraphics[width = \textwidth,valign = t]{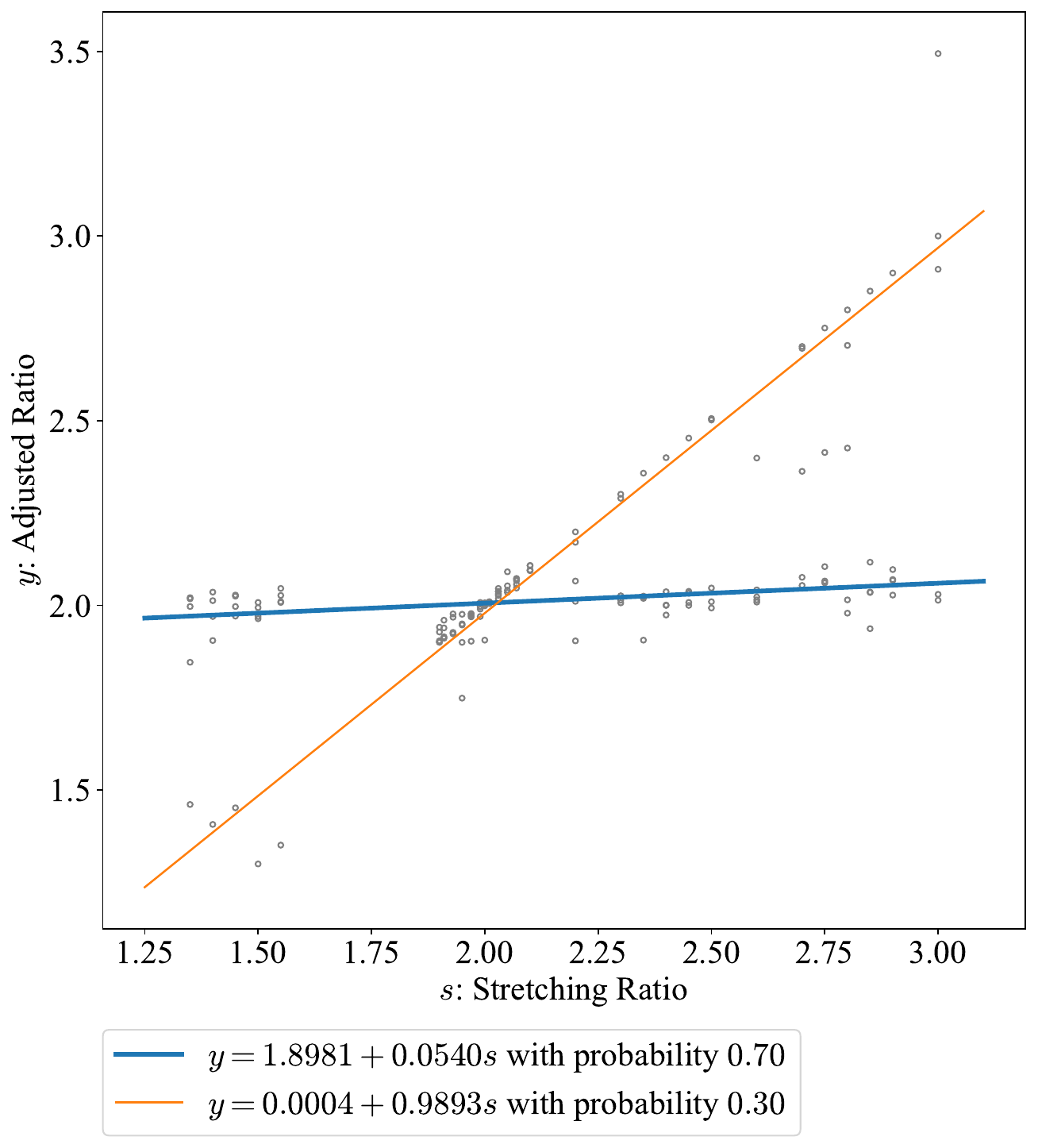}
	\end{subfigure}
	\vspace{-1em}
	\caption{Music tone perception: \subref{fig:tonedata_withoutBIC} Before BIC selection; \subref{fig:tonedata_afterBIC} After BIC selection.}
	\label{fig:tonedata}
\end{figure}

We model this data using
mixture of linear regression with covariate $x = (1, s)^\T$. After
setting $\hat{\sigma} = 0.1200$ via 10-fold cross-validation, we
compute the NPMLE and apply BIC selection to reduce overfitting.
Figure~\ref{fig:tonedata_afterBIC} shows our method identified 
two components, corresponding precisely to the two theoretical music
perception models, despite no component count prior knowledge.


\subsubsection{Real Data:
  \texorpdfstring{CO\textsubscript{2}-GDP}{CO2-GDP} Relationship}


CO\textsubscript{2} emissions, primarily from fossil fuels, are widely
considered a key driver of global warming, while GDP reflects economic
wellbeing. The relationship between these metrics is crucial for
balancing growth with sustainability. We analyze per capita
CO\textsubscript{2} emissions and GDP data from 159 countries in 2015
\citep{owideconomicgrowth}, setting CO\textsubscript{2} emissions per
capita (in 10 tons) as response $y$ and GDP per capita (in 10,000 USD)
with constant term as covariate $x = (1,g)^\T$. We selected
$\hat{\sigma} = 0.1343$ via 10-fold cross-validation.

\begin{figure}[!htb]
	\centering
	\includegraphics[width = 0.9\textwidth]{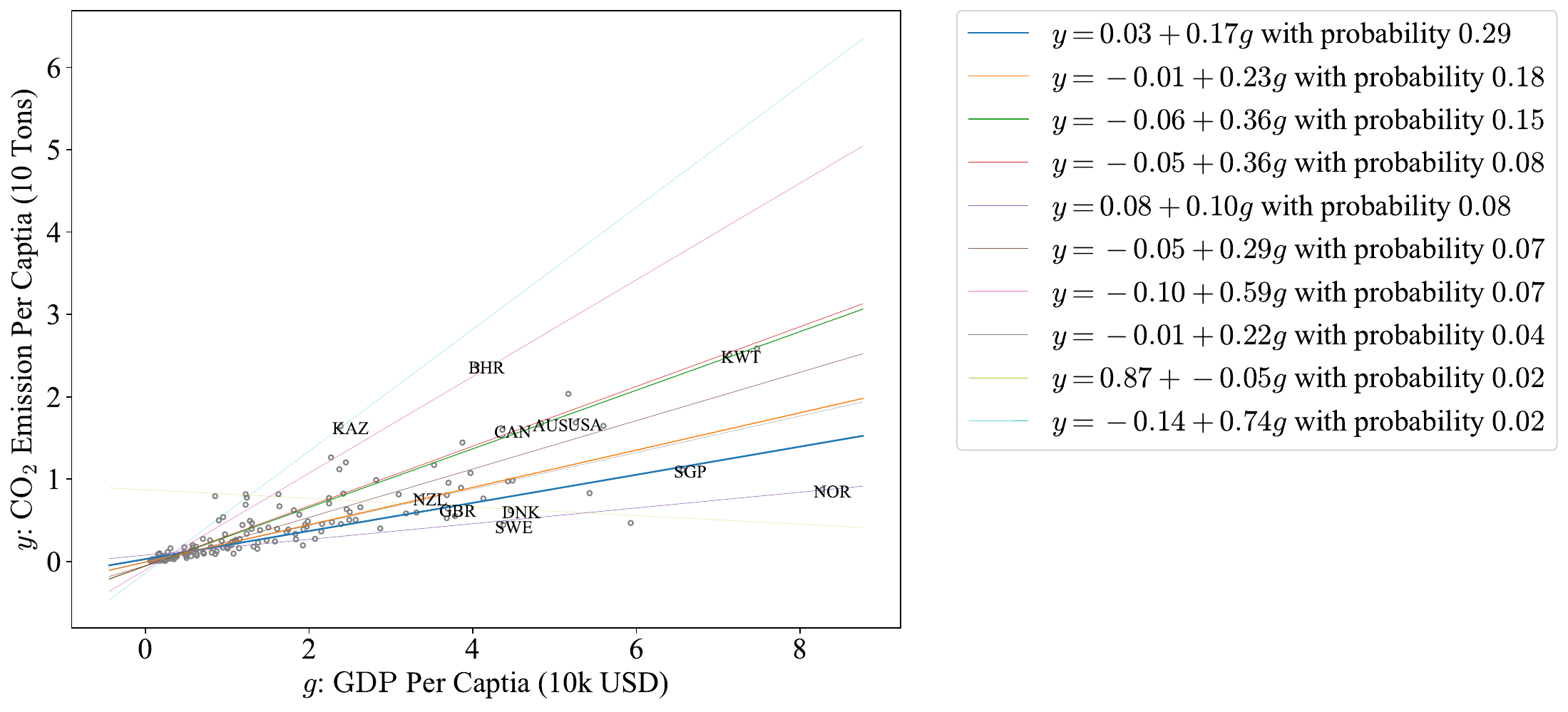}
	\caption{Fitted result of CO\textsubscript{2}-GDP data before BIC selection.}
	\label{fig:co2_gdp_withoutBIC}
\end{figure}

\begin{figure}[!htb]
	\centering
	\includegraphics[width = 0.9\textwidth]{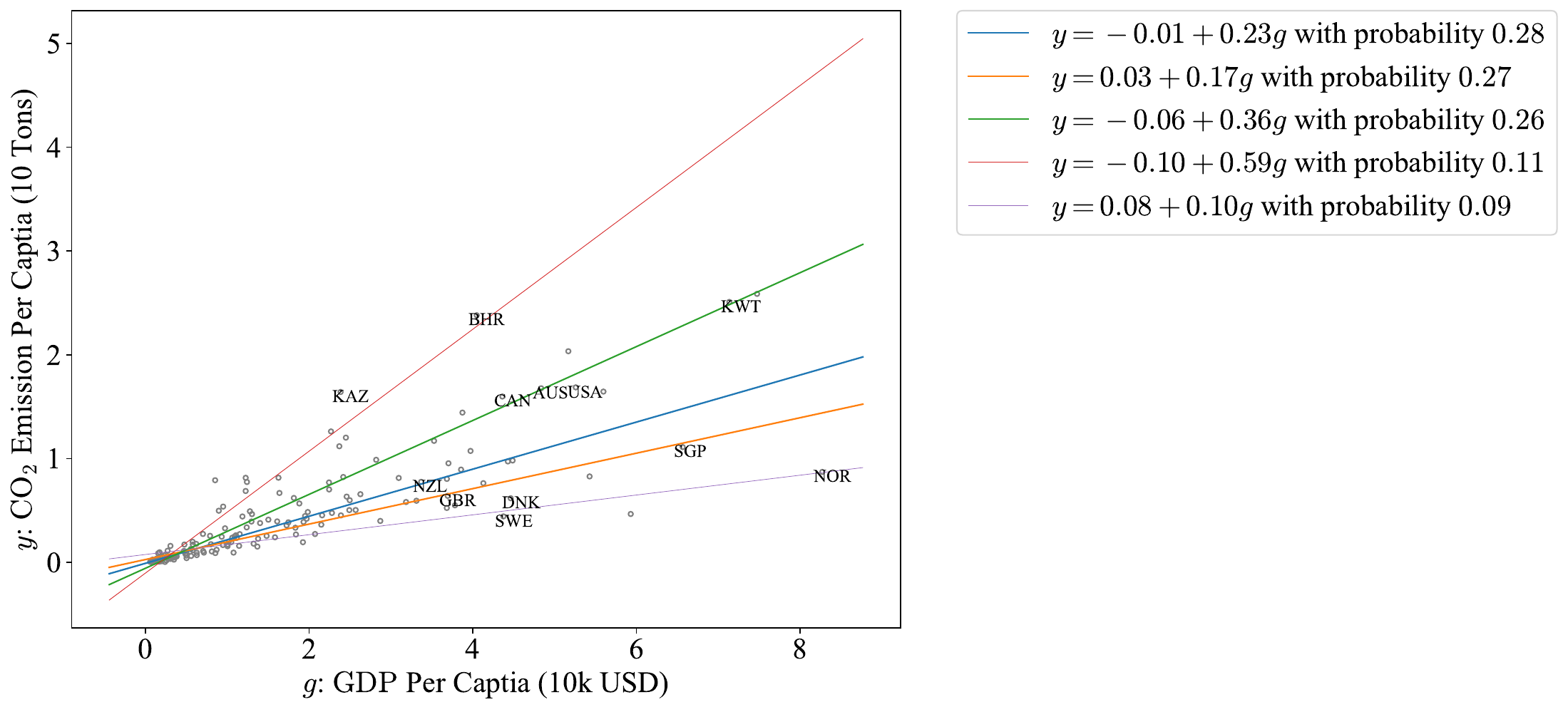}
	\caption{Fitted result of CO\textsubscript{2}-GDP data after BIC selection.}
	\label{fig:co2_gdp_afterBIC}
\end{figure}

Figures~\ref{fig:co2_gdp_withoutBIC} and \ref{fig:co2_gdp_afterBIC}
show the fitted results by NPMLE before and after BIC selection,
respectively, with select countries annotated by their three-digit
codes in Figure~\ref{fig:co2_gdp_afterBIC}. All five identified
components have intercepts close to zero but differ in slopes. 

The component with the highest slope (0.59) has very low mixing
probability and includes fossil-fuel-rich countries like Bahrain and
Kazakhstan. The second highest slope (0.36) captures developed nations
with substantial fossil fuel consumption, such as Canada, Australia,
and the United States. The lowest slope component, with high mixing
probability, includes countries like Sweden, Denmark, and Norway,
known for low emissions and strong environmental policies. Countries
within each component tend to share geographic 
proximity or similar resource/development profiles. This natural
clustering validates our mixture model approach and provides insights
into potential development paths for lower GDP countries
\citep{hurn2003estimating}, demonstrating the practical utility of our
method in identifying meaningful economic-environmental patterns.

\subsubsection{Real Data: Worker Wage}

We apply our method to analyze wage determinants using a dataset of
2,260 full-time male workers from the southern United States in 1987,
previously studied by \citet{bierens2001integrated} and available in
the R package \texttt{Sleuth3}. To reduce outlier effects, we
restricted the sample to workers with 13-18 years of
education. Following standard labor economics models
\citep{mincer1974schooling, lemieux2006mincer}, we use log weekly
earnings as response ($y$) and normalized covariates $x =
(1, ex, ex^2, ed)^\T$, where $ex$ represents years of experience and
$ed$ represents years of education.

\begin{figure}[!htb]
	\centering
	\includegraphics[width = 0.9\textwidth]{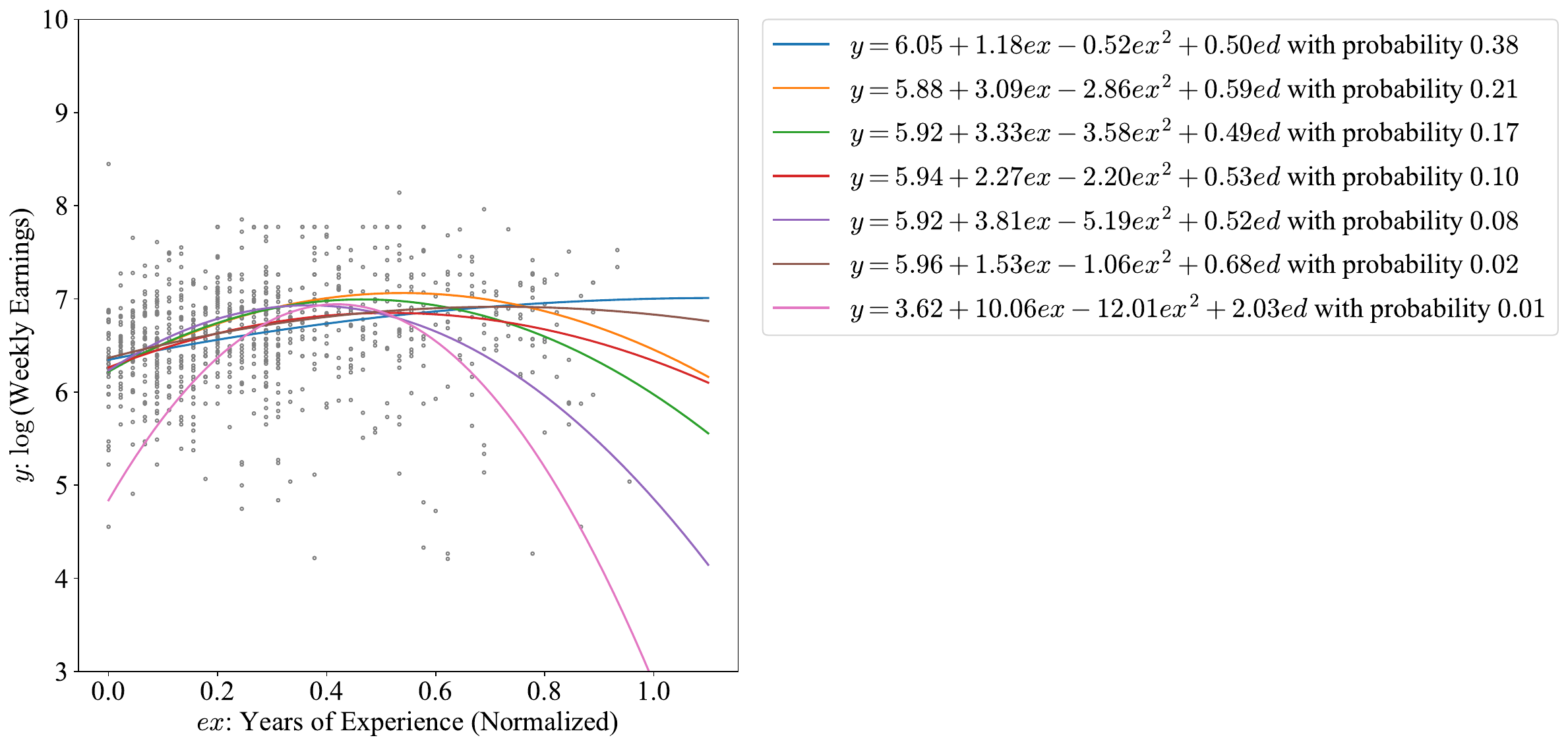}
	\caption{Fitted result of the worker wage example
          ($\hat{\sigma} = 0.4953$). For better visualization, only
          workers with $16$ years of education are
          shown.}
	\label{fig:worker_wage} 
\end{figure}

Figure~\ref{fig:worker_wage} shows the fitted results without BIC
selection. A key advantage of mixture models is the ability to compute
posterior component membership probabilities for each
individual. These posterior probabilities indicate how likely each
worker belongs to each of the identified components based on their
wage patterns. 

We examine component membership in two ways: (1) fractional
membership, where each worker contributes their posterior probability
to each component (allowing a worker to partially belong to multiple
components), and (2) integer membership, where each worker is assigned
entirely to the single component with their highest posterior
probability. 

Notably, we analyze membership patterns by race (Black vs. Non-Black),
though race was not included as a model
covariate. Table~\ref{tab:membership_race} summarizes the percentage
of workers in each racial group assigned to each component.

\begin{table}[!htb]
	\centering
	\adjustbox{max width=\textwidth}{
		\begin{tabular}{|l|*{7}{c|}}
			\hline
			\textbf{Membership Per-Component (\%)} & \text{Comp 1} & \text{Comp 2} & \text{Comp 3} & \text{Comp 4} & \text{Comp 5} & \text{Comp 6} & \text{Comp 7} \\ \hline
			\textbf{Black Workers} & 38.61 (95.97) & 20.29 (0.00)
             & 17.35 (0.34) & 10.01 (0.00) & 8.87 (0.67) & 1.95 (0.00) & \textbf{2.92 (3.02)} \\ \hline 
			\textbf{Non-Black Workers} & 38.96 (96.32) &  21.29 (1.20) & 17.81 (0.71) & 9.84 (0.00) & 8.47 (0.86) & 1.98 (0.00) & \textbf{1.65 (0.90)} \\ \hline 
		\end{tabular}
	}
	\caption{Component Membership by Race. Fractional membership
          percentage and integer membership percentage (in the
          parentheses).} 
	\label{tab:membership_race}
  \vspace{-1em}
\end{table}
While Component 1 contains the majority of workers from both racial
groups, Component 7 shows a notable disparity: Black workers (2.92\%
fractional, 3.02\% integer) are overrepresented compared to Non-Black
workers (1.65\% fractional, 0.90\% integer). Component 7,
characterized by coefficients $\hat{\beta}_7 = (3.62, 10.06, -12.01,
2.03)^\T$, has the largest negative coefficient on $ex^2$, indicating
that wages initially increase with experience but then decrease
rapidly. This suggests Black workers are more likely to experience
stronger diminishing returns to experience. This finding demonstrates
our model's ability to uncover subtle patterns in wage determinants
across different demographic groups without explicitly incorporating
race as a predictor variable.

\section{Discussion} \label{section:conclusions}

\subsection{On Computation}\label{subsection:cgm}
To approximately solve \eqref{eq:maximize_G}, Algorithm
\ref{exemplaralgo} restricts $G$ to be supported 
on finitely many exemplar points. One might attempt to 
solve the infinite-dimensional problem \eqref{eq:maximize_G} directly
via standard convex optimization 
algorithms such as the CGM (Conditional Gradient Method; see
\citet{jaggi2013revisiting}). The CGM is closely related to the Vertex
Direction Method (VDM) and the Vertex Exchange Method (VEM) which have
been historically popular for NPMLE computation in mixture models
(\citet{wu1978some,lindsay1983geometrya, bohning1986vertex,
  bohning2000computer}).   

When applied to \eqref{eq:maximize_G}, the CGM leads to the following
iterative algorithm. Initialize with $G^{(0)} =
\delta_{\{\beta^{(0)}\}}$ for some  $\beta^{(0)} \in K$. Then for each
$k \geq 0$, solve the $p$-dimensional optimization:
\begin{equation}\label{eq:subprob_beta}
  \tilde{\beta}^{(k)} \in \argmax \left\{\frac{1}{n \sigma} \sum_{i=1}^n
    \frac{1}{f^{G^{(k)}}_{x_i}(y_i)} \phi \left(\frac{y_i - x_i^{\T}
        \beta}{\sigma} \right) : \beta \in K \right\}, 
\end{equation}
and take $G^{(k+1)}$ to be the solution of \eqref{eq:maximize_G} when
$G$ is restricted to be supported on $\{\tilde{\beta}^{(0)}, \dots,
\tilde{\beta}^{(k)}\}$.  This algorithm seemingly avoids explicit
discretization but the difficulty lies in solving the non-convex problem
\eqref{eq:subprob_beta}. Naive gridding is
computationally prohibitive (even for $p =
3$) and standard black-box optimization routines are 
slow (for $p \geq 6$) with no guarantees. 

Compared to a naive grid, it makes sense to use a tailored
discretization for solving \eqref{eq:subprob_beta}. By writing the
gradient condition for  the optimization in \eqref{eq:subprob_beta},
one can see that every optimizer that is in the interior of $K$ should
satisfy the condition of Proposition \eqref{exem_lemma}. We can
therefore discretize \eqref{eq:subprob_beta} by generating a large
number of points as in Algorithm \ref{exemplaralgo}. Because these
generated points do not depend on the current iterate $G^{(k)}$, it
turns out that this overall scheme is simply implementing the CGM
algorithm on the finite-dimensional convex optimization problem
\eqref{disc.prob}. Therefore, Algorithm \ref{exemplaralgo} can also be
seen as a variant of CGM obtained by solving the subproblem
\eqref{eq:subprob_beta} using exemplars.

\subsection{Unresolved Questions}\label{subsection:unresolved}

\noindent \textbf{Uniqueness}: We do not know if $\hat{G}$ is unique
when the design matrix 
$\mathbf{X}$ is of full column rank. Some uniqueness results for
NPMLEs in the univariate case can be found in
\cite{lindsay1983geometryb}. A counterexample for uniqueness of the
NPMLE for multivariate Gaussian location mixture densities is in
\citet[Lemma 2]{soloff2021multivariate}.

\noindent \textbf{Non-compact $K$}: Our theoretical results on Hellinger accuracy assume that $K$ is
compact. We do not know if these results continue to hold for
non-compact $K$ such as when $K = \mathbb{R}^p$. 

\noindent \textbf{More rates}: While Theorem \ref{thm:weakconsistency}
shows consistency of $\hat{G}$ in random-design, corresponding convergence rates
are unknown. When $p = 1$,
$x_i \neq 0$ for each $i$ and $K = [-R, R]$, we can show existence of
$n_0$ (depending on $\min_i |x_i|$, $\max_i |x_i|$ and $R$) such that
for all $n \geq n_0$, 
\begin{equation}
  \label{W2rate}
        W_2^2(G^*, \hat{G}_n)  \lesssim  \frac{1}{\log n}
\end{equation}
 with probability at least $1- \frac{1}{n^8}$ (the constants
 underlying $\lesssim$ depend on $\min_i |x_i|$ and $\max_i
 |x_i|$). Here  $W_2^{2}(G^*, \hat{G}_n)$ is the $L_2^2$ Wasserstein  
        distance (see e.g., \cite{nguyen2013convergence}). This result
        follows from \citet[Theorem 10]{soloff2021multivariate}
        because, when $p = 1$, the mixture of regression model reduces
        to the heteroscedastic Gaussian location mixture model:   
\begin{equation}
  \label{red_het}
      \frac{y_i}{x_i} = \beta_i + \frac{z_i}{x_i},\text{ for } i = 1,\dots,n.
\end{equation}
It is unclear how to derive the rate for $p \geq 2$, as this
reduction fails in higher dimensions.

\subsection{Sparsity of $\hat{G}$} 
\label{subsection:disc_unique}

Theorem \ref{thm:existence} shows that an NPMLE $\hat{G}$ exists
with at most $n$ support points in $K$. In practice however (see e.g.,
Section
\ref{section:experimental}), approximate NPMLEs typically have far fewer support
points.  

For one-dimensional Gaussian location mixtures,
\citet{polyanskiy2020self} proved that the NPMLE has $O(\log n)$
support points under certain conditions on the true mixing
measure. Higher-dimensional analogues of this result remain elusive
and represent a challenging open problem. In our mixture of
regressions framework, rigorously establishing the sparsity of
$\hat{G}$ when $p > 1$ remains an important direction for future
research. 

For $p = 1$, however, we prove a $O(\log n)$ upper bound on the number
of support points in $\hat{G}$ under fixed design, extending
\citeauthor{polyanskiy2020self}'s result to mixture of regressions.


\begin{theorem} \label{thm:onedim_selfregu}
 Consider $p=1$ and the design points
satisfying $|x_i/x_j|\leq r_0$ for all $i,j$. Assume that $G^* \left\{\beta: \|\beta\| \le R
\right\} = 1$ and $\max_{1 \leq i \leq n} \|x_i\| \leq B$ for some
$B,R > 0$. Then for any $\tau>1$ and $n > \max\{\exp(C_0),\exp(C_1 r_0^2 B^2R^2 \sigma^{-2})\}$, every NPMLE
$\hat{G}$ has at most $\tau r_0^2  \cdot O(\log n)$ support points
with probability at least $1-n^{-\tau}$, where $O(\cdot)$ omits multiplicative constant factors and $C_0,C_1$ are constants.  
\end{theorem}

The main ingredient in proving Theorem \ref{thm:onedim_selfregu} is to
bound the number of zeros of $\nabla D(\hat{G},\beta)$ (recall, from
Proposition \ref{prop:opt_condition} that the support points of
$\hat{G}$ that are in $\text{int}(K)$ satisfy $\nabla D(\hat{G},\beta)
= 0$) by using a variant of the Jensen formula from complex
analysis.




\subsection{When $p$ is  large}\label{section:hd}

Throughout, we focused on cases with fixed dimension
$p$. Our convergence rates (Theorems \ref{fdrate} and
\ref{them:randomdesign_predictionerror}) contain logarithmic terms
$(\log n)^p$ that degrade as $p$
increases. Here we demonstrate this empirically and suggest
an alternative for high-dimensional settings. 

We extend the simulation from Figures \ref{fig:3comp_noisytrue-intro}
and \ref{fig:3comp_fitted-intro} in Section \ref{simpeg} by increasing
the sample size to $n = 500$ and systematically adding irrelevant
covariates. The original design matrix has two columns, $\mathbf{X} = [1_{n
  \times 1} : w_{n \times 1}]$, where entries of $w$ follow
$\text{uniform}(-1, 3)$. For each $\tilde{p} \in \{1, \dots, 11\}$, we
add $\tilde{p}$ spurious covariates to create $ \mathbf{X}_{\text{new}} = [1 : w : v^{(1)} : \dots : v^{(\tilde{p})}]$ where each $v^{(j)}$ contains independent $\text{uniform}(-1, 3)$ entries. 

We fit our model to $(\mathbf{X}_{\text{new}}, Y)$ using two approaches: (1)
$\hat{G}(\sigma)$ with $\sigma = 0.5$ (true value), and
(2) $\hat{G}(\hat{\sigma})$ with $\hat{\sigma}$ estimated by
CV. To evaluate prediction performance, we generate a
test dataset with 
$n_{\text{test}} = 50$ points, where each $x_i^{\text{test}}$ is
generated identically to the training data, and each
$y_i^{\text{test}}$ follows model \eqref{model:mixlin_i_1} with the
true mixture $G^*$ concentrated on $(3, -1)$, $(1, 1.5)$, and $(-1,
1.5)$ with probabilities $0.3$, $0.3$, and $0.4$ respectively. The
test score is $-\sum_{i=1}^{n_{\text{test}}} \log
  \hat{f}_{x_i^{\text{test}}}(y_i^{\text{test}})$ where
  $\hat{f}_{x_i^{\text{test}}}(y_i^{\text{test}})$ is the density
  estimated by our mixture model.

\begin{figure}[h]
    \centering
    \begin{subfigure}[t]{0.44\textwidth}
        \centering
        \caption{}
        \label{fig:NPMLEsigma_withoutBIC}
        \includegraphics[width=\textwidth]{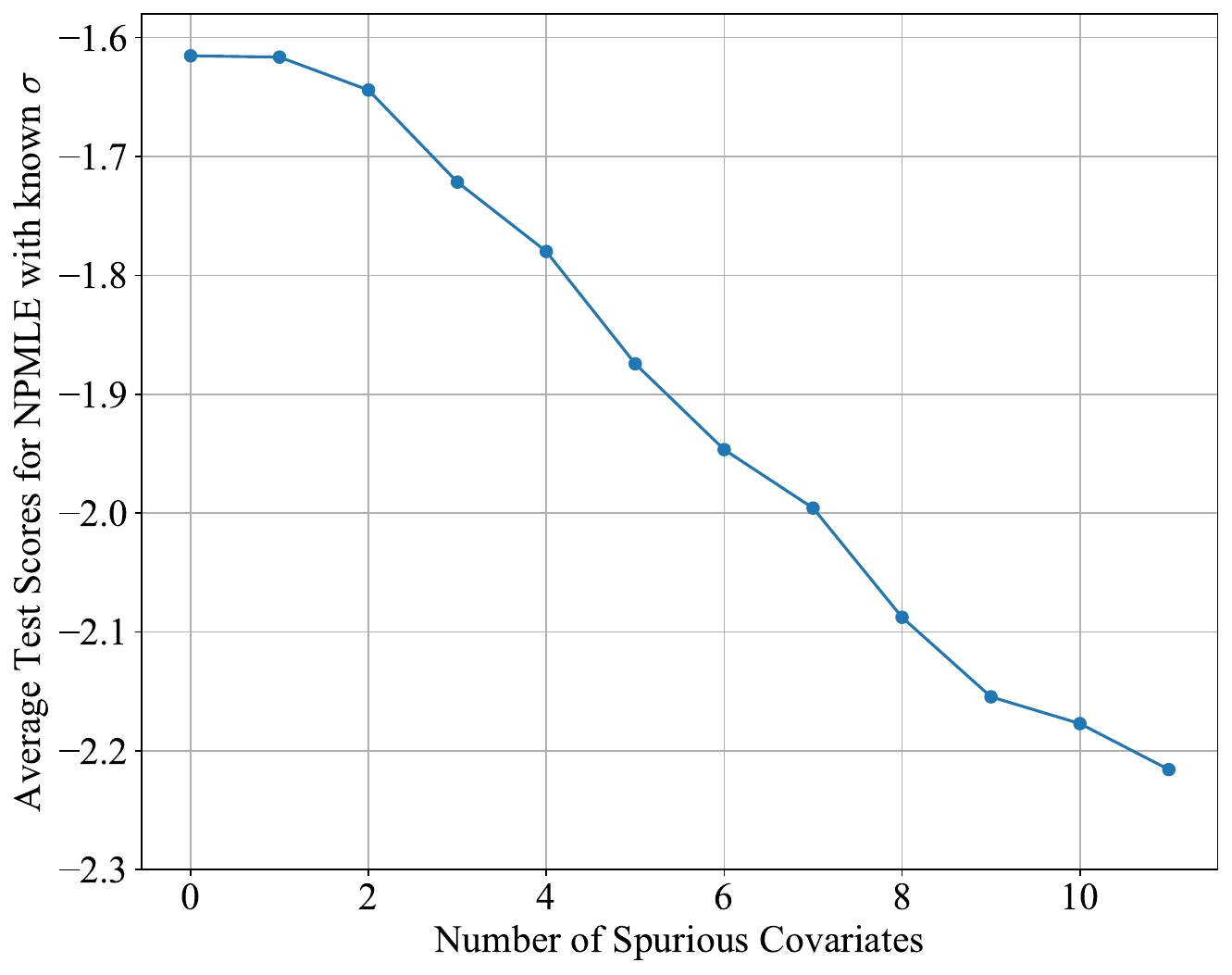}
    \end{subfigure}
    \hfill
    \begin{subfigure}[t]{0.44\textwidth}
        \centering
        \caption{}
        \label{fig:NPMLE_withoutBIC}
        \includegraphics[width=\textwidth]{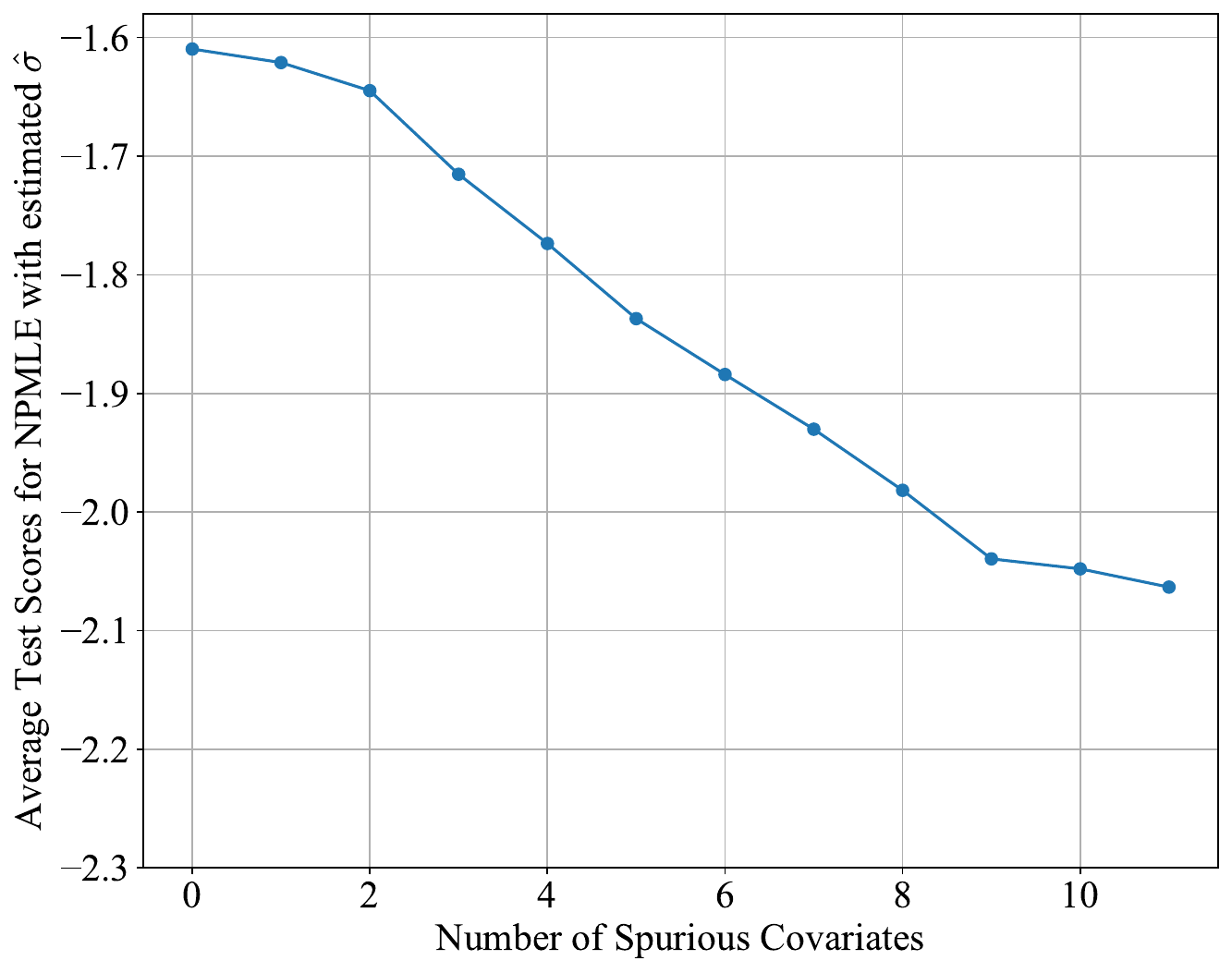}
    \end{subfigure}
    \caption{Average test scores versus the number of spurious covariates
      $\tilde{p}$: \subref{fig:NPMLEsigma_withoutBIC} NPMLE $\hat{G}(\sigma)$ with known $\sigma$; \subref{fig:NPMLE_withoutBIC} NPMLE $\hat{G}(\hat{\sigma})$ with CV-estimated $\hat{\sigma}$.}  
    \label{fig:average_scores_plots}
 \end{figure}


For each $\tilde{p}$, we average test scores across 20
repetitions. Figure \ref{fig:average_scores_plots} shows
prediction accuracy deteriorating for both $\hat{G}(\sigma)$ (left) and
$\hat{G}(\hat{\sigma})$ (right) as the number of spurious covariates
increases. The NPMLE 
makes no assumptions about the probability 
$G$ on $\mathbb{R}^p$. As $p$ grows, the space of probabilities
becomes too large, leading to overfitting. This can be addressed by
regularizing $G$ via splitting covariates into two groups. 

In our simulation, the correct model with spurious covariates is $y_i
= x_i^{\T} \beta^i_{(1)} + z_i^{\T} \beta_{(2)} + \epsilon_i$  where
$\beta^i_{(1)} \overset{\text{i.i.d}}{\sim} G^* = 0.3 
\delta_{\{(3, -1)\}} + 0.3 \delta_{\{(1, 1.5)\}} + 0.4 \delta_{\{(-1,
  0.5)\}}$ and $\beta_{(2)} = \mathbf{0}$. This suggests the following
general model with partitioned covariates:  
\begin{equation}\label{splitmodel}
  y_i = x_i^{\T} \beta^i_{(1)} + z_i^{\T} \beta_{(2)} + \epsilon_i,
\end{equation}
where $\beta^i_{(1)} \overset{\text{i.i.d}}{\sim} G^*$ and
$\beta_{(2)}$ is \textit{fixed}. Unlike our original optimization,
maximizing the log-likelihood in Model \eqref{splitmodel} is
non-convex in both $G$ and $\beta_{(2)}$. We propose an alternating
scheme: (a) For fixed $\beta_{(2)}$, apply our standard algorithm to data
$(x_i, y_i - z_i^T \beta_{(2)})$, and (b) For fixed $G$, use numerical
optimization with multiple starting points to estimate
$\beta_{(2)}$. After obtaining $\hat{G}$ and $\hat{\beta}_{(2)}$, we
estimate $\sigma$ via cross-validation.  

Applied to our simulation (with $x$ as correct covariates and $z$ as
spurious covariates), results appear in
Figures~\ref{fig:higher_p_withoutCV}, \ref{fig:higher_p}, and
\ref{fig:higher_p_ridgeline}. The estimated $\hat{\beta}_{(2)}$
vectors (with elements approximately between -0.06 and 0.01) closely
match the ground truth $\mathbf{0}$.

\begin{figure}[h]
	\centering
	\begin{subfigure}[t]{0.26\textwidth}
		\centering
		\caption{}
		\label{fig:higher_p_withoutCV}
		\includegraphics[width=\textwidth]{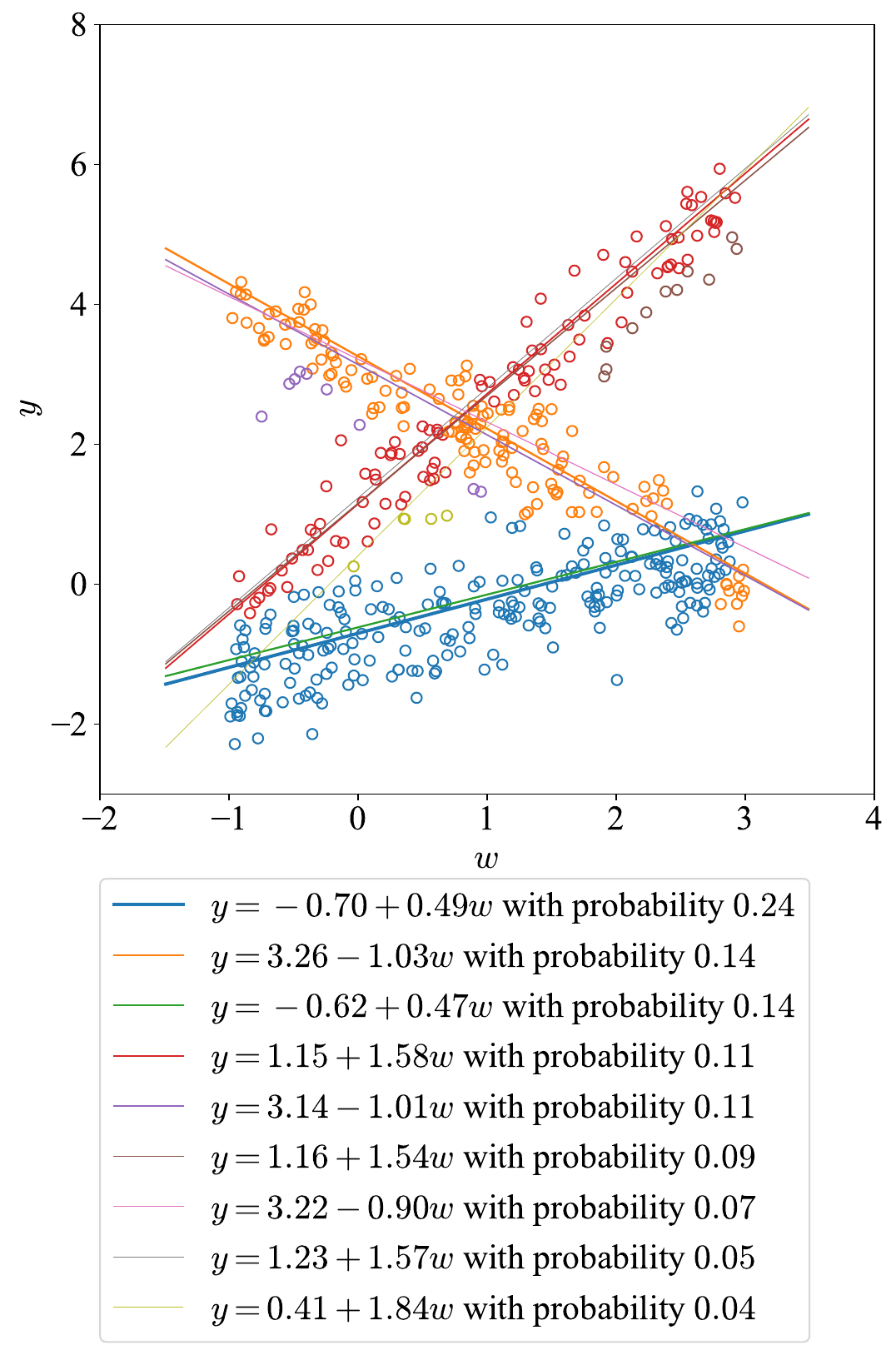}
	\end{subfigure}
	\begin{subfigure}[t]{0.26\textwidth}
		\centering
		\caption{}
		\label{fig:higher_p}
		\includegraphics[width=\textwidth]{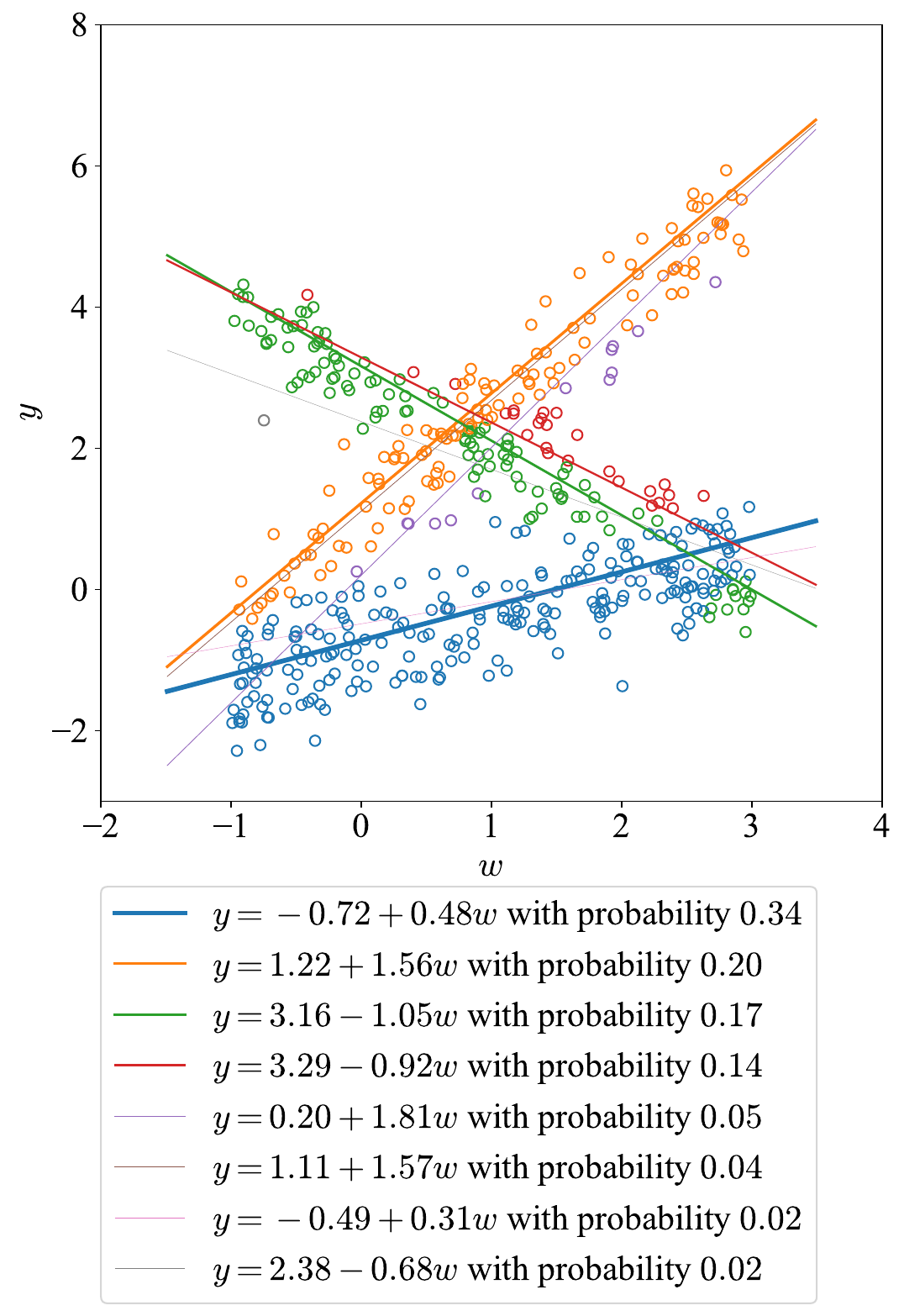}
	\end{subfigure}
	\begin{subfigure}[t]{0.45\textwidth}
		\centering
		\caption{}
		\label{fig:higher_p_ridgeline}
		\includegraphics[width=\textwidth]{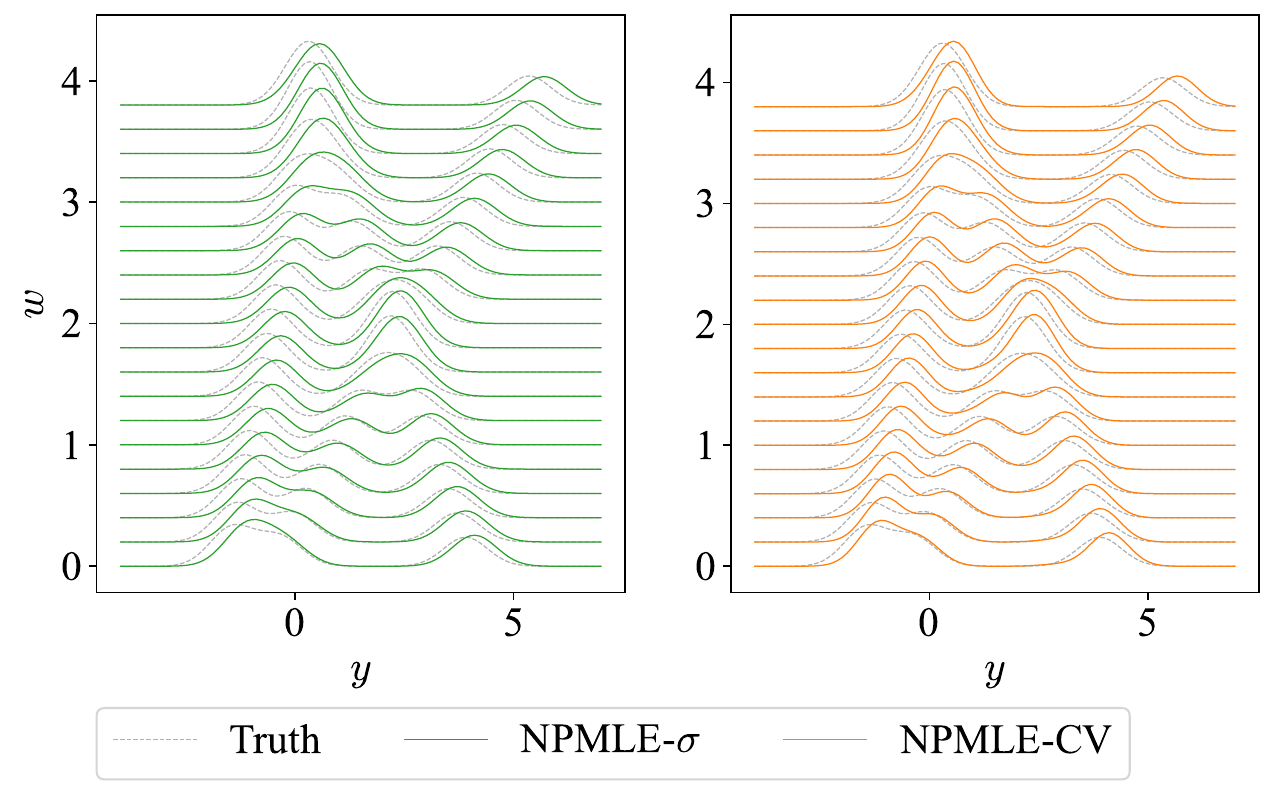}
	\end{subfigure}
	\caption{Fitted results by the alternating approach: \subref{fig:higher_p_withoutCV} With true $\sigma$; \subref{fig:higher_p} With $\hat{\sigma}$ selected by cross-validation; \subref{fig:higher_p_ridgeline} Ridgeline plots of conditional densities.}  
	\label{fig:fitted_higher_p}
\end{figure}

The average test scores ($-1.6591$ with known $\sigma$ and $-1.6454$
with CV-selected $\hat{\sigma}$) significantly outperform the full
NPMLE with $\tilde{p} = 11$, demonstrating the effectiveness of our
regularization strategy for higher-dimensional settings. To implement
this procedure, we need to know which covariates 
belong to $x$ and which to $z$ -- information typically
unavailable. When only $p_0$ (the number of covariates in $x$) is
known, we can consider all ${p \choose p_0}$ possible partitions of
covariates and select the best split based on likelihood 
maximization. This is computationally feasible for moderate
$p$ (e.g., $p \leq 15$) and small $p_0$ (e.g., $p_0 = 2$ or 
$3$). A detailed study is left for future
work.  

The model \eqref{splitmodel} resembles the ``partial linear model" of
\citet{jiang2010empirical}, who studied the special case where $x$
contains only the intercept (i.e., $G^*$ is one-dimensional). They
also employed alternating maximization to  estimate $G^*$ and
$\beta_{(2)}$, though without a $\sigma$ parameter as their setting had known
(possibly heteroscedastic) standard deviations.


\bibliographystyle{chicago}
\bibliography{NPMLE}


%
%
%
%

%

\newpage

\begin{appendices}
Proofs of all our theorems are given in Section
\ref{allproofs}. Section~\ref{subsec:detailed_coef} contains numerical
results for the simulations in Subsections
\ref{sec:sinusoid_simulation} and \ref{sec:changepoint_simulation}.

\section{Proofs}\label{allproofs}

\subsection{Proof of Theorem
	\ref{thm:existence}}\label{appendix:existence}
The notation for $\fv^G$, $\fv^{\beta}$ and
$\PS_K$ introduced at the beginning of Section
\ref{section:existenceandcomputation} will be used in the proof below.  

\begin{proof}[Proof of Theorem \ref{thm:existence}]
	The objective function in the optimization problem \eqref{eq:maximize_G}
	only depends on $G$ through the vector $\fv^G$. As a result,
	\eqref{eq:maximize_G} is equivalent to
	\begin{equation}\label{tempopt}
		\argmax \left\{\frac{1}{n} \sum_{i=1}^n \log \fv(i) : \fv \in \QS_K \right\}
	\end{equation}
	where $\fv(i)$ denotes the $i^{th}$ element of the vector $\fv \in
	\mathbb{R}^n$ and 
	\[
	\QS_K=\{\fv^G \,:\, G \text{ is a probability measure supported on } K\}.
	\]
	\begin{revise}We note that the set $\QS_K$ need not be compact
          in general, even in the case that $K = \RS^p$. This is
          because for any probability measure supported on $\RS^p$,
          each component of $\fv^{G}$ must be strictly positive. On the other
          hand, if we consider one non-zero design point, say $x_i
          \neq \vec{0} \in \RS^p$, then for the Dirac measure $G_s=
          \delta_{\{s x_1\}}$ indexed by a positive scalar $s$, we
          have that the probability $\fv^{G_s}(i) :=
          f_{x_i}^{G_s}(y_i) = \frac{1}{\sqrt{2\pi} \sigma}
          \exp\left\{-\frac{(y_i - s \Vert x_i\Vert)^2}{2\sigma^2}
          \right\}$ goes to $0$ as $s$ goes to infinity. This implies
          that the boundary points of $\mathcal{Q}_K$ may have
          components equal to $0$, and therefore such boundary points
          are strictly outside $\mathcal{Q}_K$. It follows that
          $\mathcal{Q}_K$ is not closed and thus not
          compact. \end{revise}

	 We claim that
	\begin{equation}\label{claim1}
		\QS_K \subseteq \mathrm{conv}(\mathrm{cl}(\PS_K))
	\end{equation}
	where $\mathrm{cl}$ denotes closure and $\mathrm{conv}$ denotes convex
	hull. This claim will be proved later. The set
	$\mathrm{conv}(\mathrm{cl}(\PS_K))$ is compact because
	$\mathrm{cl}(\PS_K)$ is compact (as $\PS_K \subseteq [0,
	1/(\sqrt{2\pi} \sigma)]^n$ is bounded) and as the convex hull of a 
	compact set in Euclidean space is compact (see, for example,
	\citet[Proposition 1.3.2]{bertsekas2003convex}). Therefore a solution
	$\hat{\fv} \in \mathrm{conv}(\mathrm{cl}(\PS_K))$ exists for the
	optimization problem  
	\begin{equation}\label{tempopt2}
		\argmax \left\{\frac{1}{n} \sum_{i=1}^n \log \fv(i) : \fv \in
		\mathrm{conv}(\mathrm{cl}(\PS_K)) \right\}. 
	\end{equation}
	Further the solution $\hat{\fv}$ is unique as the objective function
	$L(\fv) := \frac{1}{n} \sum_{i=1}^n \log \fv(i)$ is strictly
	concave. Moreover $\hat{\fv}$ lies in the boundary of the set
	$\mathrm{conv}(\mathrm{cl}(\PS_K))$ because otherwise  $\nabla
	L(\hat{\fv})^\T = ({1}/{\hat{\fv}(1)}, \dots,
	{1}/{\hat{\fv}(n)})$ would have to be zero which is impossible. As a
	result, by the the Carath\'eodory 
	theorem (see, for example, \citet[Appendix 2]{silvey1980optimal}), 
	$\hat{\fv}$ can be written as a convex 
	combination of at most $n$ points in
	$\mathrm{cl}(\PS_K)$ i.e., $\hat{\fv} = \sum_{j=1}^N \pi_j g_j$ for
	some $N 
	\leq n$, $g_j \in \mathrm{cl}(\PS_K)$ and $\pi_j > 0$ with $\sum_{j}
	\pi_j = 1$. \begin{revise} We are able to claim $n$ points (as
          opposed to $n+1$) due to the fact that
          $\hat{\fv}$ lies in the boundary of the set  
	$\mathrm{conv}(\mathrm{cl}(\PS_K))$ (see  \citet[Appendix
        2]{silvey1980optimal}).\end{revise}   We now 
      claim that under the  
	assumptions on $K$ given in the statement of Theorem
	\ref{thm:existence}, for every $g \in \mathrm{cl}(\PS_K)$, there
	exists $\beta \in K$ such that 
	\begin{equation}\label{claim2}
		g(i) I\{g(i) > 0\} = \fv^{\beta}(i) I\{g(i) > 0\}. 
	\end{equation}
	Assuming the validity of this claim (which will be proved later),
	there exists $\beta_j \in K$ for which
	\begin{equation}\label{claim2imp}
		g_j(i) I\{g_j(i) > 0\} = \fv^{\beta_j}(i) I\{g_j(i) > 0\} \qt{for $j
			= 1, \dots, N$}. 
	\end{equation}
	Now if $g_j(i) = 0$ for some $j$ and $i$, we would have $\sum_{j}
	\pi_j \fv^{\beta_j}$  having a higher objective value compared to
	$\hat{\fv} = \sum_j \pi_j \fv^{\beta_j}$ (note that all components of
	$\fv^{\beta}$ are all strictly positive for every $\beta$) which would
	contradict the fact that $\hat{\fv}$ is the unique solution to
	\eqref{tempopt2}.  We thus have $g_j(i) > 0$ for all $i$ which
	implies, by \eqref{claim2imp}, that $g_j = \fv^{\beta_j}$ for every
	$j$. This obviously implies that $g_j \in \PS_K$ so that $\hat{\fv}
	\in \mathrm{conv}(\PS_K)$. Also
	\begin{equation*}
		\hat{\fv} = \sum_{j=1}^N \pi_j \fv^{\beta_j} = \fv^{\hat{G}} \in
		\mathrm{conv}(\PS_K) 
		\qt{where $\hat{G} = \sum_{i=1}^N \pi_j \delta_{\{\beta_j\}}$}. 
	\end{equation*}
	As a result $\hat{\fv} \in \QS_K$ which shows that $\hat{\fv}$ is the
	unique solution to \eqref{tempopt} and this completes the proof of
	Theorem \ref{thm:existence}. We only need to prove the two claims
	\eqref{claim1} and \eqref{claim2}. 
	
	For \eqref{claim1}, take $\fv^G \in \QS_K$ where $G$ is a probability
	measure on $K$. By \citet[Theorem
	6.3]{parthasarathy2005probability}, there exist discrete probability measures $\{
	\mu_m\}_{m=1}^\infty$ with finite supports converging weakly to $G$ as
	$m \rightarrow \infty$ and this implies $f^{\mu_m}_{x_i}(y_i)
	\rightarrow f^G_{x_i}(y_i)$ for $i = 1, \dots, n$. As a result,
	$\fv^{\mu_m}  \rightarrow \fv^G$ as $m \rightarrow \infty$. This
	implies that
	\[
	\QS_K \subseteq \mathrm{cl}(\mathrm{conv}(\PS_K)).
	\]
	because each $\fv^{\mu_m} \in \mathrm{conv}(\PS_K)$. To complete the
	proof of \eqref{claim1}, it is enough to show that
	\begin{equation}
		\label{subclaim1}
		\mathrm{conv}(\mathrm{cl}(\PS_K)) =
		\mathrm{cl}(\mathrm{conv}(\PS_K)). 
	\end{equation}
	For \eqref{subclaim1}, first note that $\PS_K\subseteq
	\mathrm{conv}(\PS_K)$ which implies $\mathrm{cl}(\PS_K) \subseteq 
	\mathrm{cl}(\mathrm{conv}(\PS_K))$. $\mathrm{cl}(\mathrm{conv}(\PS_K))$
	is convex, and $\mathrm{conv}(\mathrm{cl}(\PS_K))$ is the smallest
	convex set that contains $\mathrm{cl}(\PS_K) $, so 
	\[
	\mathrm{conv}(\mathrm{cl}(\PS_K)) \subseteq \mathrm{cl}(\mathrm{conv}(\PS_K)).
	\]
	For the other inclusion, observe that, as noted earlier,
	$\mathrm{conv}(\mathrm{cl}(\PS_K))$ is compact so that
	$$\mathrm{conv}(\mathrm{cl}(\PS_K)) = 
	\mathrm{cl}(\mathrm{conv}(\mathrm{cl}(\PS_K))) \supseteq
	\mathrm{cl}(\mathrm{conv}((\PS_K))).$$ This proves \eqref{subclaim1}
	and consequently \eqref{claim1}. 
	
	We next prove \eqref{claim2}. Fix $g \in \mathrm{cl}(\PS_K)$. If $K$
	is compact, then $\PS_K$ is also compact so that $g \in \PS_K$ which
	means that $g = \fv^{\beta}$ for some $\beta \in K$ and this proves
	\eqref{claim2}. So let us assume that $K$ is not necessarily compact
	and that the second assumption in the statement of Theorem
	\ref{thm:existence} holds.
	
	Let $I := \{1 \leq i \leq n : \gv(i) >0\}$ and let  $V$ be the linear
	subspace of $\RS^p$ spanned by $\{x_i| i \in I\}$ (recall that $x_1,
	\dots, x_n$ are the observed covariate vectors). Because $g \in
	\mathrm{cl}(\PS_K)$, we can write $g = \lim_{l \rightarrow \infty}
	\fv^{\beta_l}$ for some sequence $\{\beta_l\}$ in $K$.  For $l \geq
	1$, let  $\alpha_l$ denote the projection of $\beta_l$ onto $V$ 
	so that $x_i^\T \alpha_l = x_i^\T \beta_l$ for all $i \in I$ and all
	$l \geq 1$. Also by our assumption on $K$, we have $\alpha_l \in
	K$. We will show that $\{\alpha_l\}_{l=1}^\infty$ is bounded.  
	
	For $i \in I$, $\gv(i)>0$ thus $\{x_i^\T \beta_l\}_{l=1}^\infty$ is
	bounded and $\lim_{l \rightarrow \infty} x_i^\T \beta_l$ exists. Since $x_i^\T
	\alpha_l = x_i^\T \beta_l$, $\{x_i^\T \alpha_l\}_{l=1}^\infty$ is also
	bounded. Take an orthonormal basis of $V$ as $r_1,r_2,\dots,r_v$. For
	any $j = 1,\dots,v$, since $V$ is spanned by $\{x_i| i \in I\}$, $r_j$
	is a linear combination of $\{x_i|i \in I\}$. Therefore, as a linear
	combination of $\{x_i^\T \alpha_l\}_{l=1}^\infty$, $\{r_j^\T
	\alpha_l\}_{l=1}^\infty$ is bounded (noting that the linear
	combination coefficients do not depend on $l$). Because  
	\[
	\alpha_l^\T \alpha_l = \sum_{j=1}^v (r_j^\T \alpha_l)^2,
	\]
	it follows that $\{\alpha_l^\T\alpha_l\}_{l=1}^\infty$ is also bounded.
	Now we can take a convergent subsequence of
	$\{\alpha_l\}_{l=1}^\infty$. The limit of the subsequence, denoted by
	$\beta$,  also belongs to $K$ because $K$ is assumed to be
	closed. For $i \in I$, $x_i^\T \beta = \lim_{l\rightarrow \infty}
	x_i^\T \beta_l$. Let $\fv^{\beta}$ denote the atomic likelihood
	vector with respect to $\beta$, then $\fv^{\beta}(i) = \gv(i)$ for all
	$i \in I$. This proves \eqref{claim2} and thereby completes the 
	proof of Theorem \ref{thm:existence}.
\end{proof}

\subsection{Proof of Proposition~\ref{proposition:nonunique}}

\begin{proof}[Proof of Proposition~\ref{proposition:nonunique}] 
Since $\mathbf{X}$ does not have
full rank, there exists a nonzero vector $v$ in its null space, i.e.,
$x_i^\T v = 0$ for all $i = 1,\dots,n$. 
Suppose $\hat{G} = \sum_{j=1}^K \delta_{\{\beta_l\}}$ is an NPMLE,
then  $\hat{G}' = \sum_{j=1^K} \delta_{\{\beta_l + v\}}$ is  also an
NPMLE since $f_{x_i}^{\hat{G}}(y_i) = \fv_{x_i}^{\hat{G}'}(y_i)$ for
all $i = 1,\dots,n$. Because $\hat{G}$ is not equal to $\hat{G}'$ when
$v$ is nonzero, we know that NPMLE is not unique in this case. 
\end{proof}

\subsection{Proof of Proposition \ref{prop:opt_condition}}

\begin{proof}[Proof of Proposition~\ref{prop:opt_condition}]
  If $\hat{G}$ solves \eqref{eq:maximize_G}, then for every $\alpha
  \in (0, 1)$ and every probability measure $G$ supported on $K$, we
  have 
  \begin{align*}
    0 &\geq \frac{1}{\alpha} \sum_{i=1}^n\left\{ \log f^{(1 - \alpha) \hat{G}
        + \alpha G}_{x_i}(y_i) - \log f^{\hat{G}}_{x_i}(y_i) \right\} \\
    &= \frac{1}{\alpha} \sum_{i=1}^n \left\{ \log \left((1 - \alpha)
      f^{\hat{G}}_{x_i}(y_i) + \alpha f^G_{x_i}(y_i) \right)  - \log
      f^{\hat{G}}_{x_i}(y_i) \right\}. 
  \end{align*}
  Taking the limit of the right hand side as $\alpha \downarrow 0$, we
  get
  \begin{align}\label{G_con}
\frac{1}{n}    \sum_{i=1}^n
    \frac{f_{x_i}^G(y_i)}{f^{\hat{G}}_{x_i}(y_i)} - 1 \leq 0. 
  \end{align}
Since this is true for every $G$ that is supported on $K$, the above
is equivalent to
\begin{align}\label{beta_con}
  \sup_{\beta \in K} \frac{1}{n} \sum_{i=1}^n
  \frac{f^{\beta}_{x_i}(y_i)}{f^{\hat{G}}_{x_i}(y_i)} \le 1
\end{align}
which is the same as \eqref{eq:D_G_beta}.

Conversely if $\hat{G}$ satisfies \eqref{beta_con} (and consequently
\eqref{G_con}), then (below we use $\log x \leq x - 1$):
\begin{align*}
  \sum_{i=1}^n \log f^{G}_{x_i}(y_i) -   \sum_{i=1}^n \log
  f^{\hat{G}}_{x_i}(y_i) = \sum_{i=1}^n \log
  \frac{f^G_{x_i}(y_i)}{f^{\hat{G}}_{x_i}(y_i)} \leq \sum_{i=1}^n
  \left(     \frac{f_{x_i}^G(y_i)}{f^{\hat{G}}_{x_i}(y_i)} - 1 \right) \leq 0
\end{align*}
for every $G$ supported on $K$. This clearly shows that $\hat{G}$
maximizes \eqref{eq:maximize_G}. 

The integral of the term inside the supremum in \eqref{beta_con} with
respect to $\beta \in \hat{G}$ is clearly one. From this, it
immediately follows that
\begin{align*}
   \frac{1}{n} \sum_{i=1}^n
  \frac{f^{\beta}_{x_i}(y_i)}{f^{\hat{G}}_{x_i}(y_i)} = 1 \qt{for
  $\beta$ a.s $\hat{G}$}
\end{align*}
which proves \eqref{eq:D_G_as}. This implies that almost every $\beta$
(with respect to $\hat{G}$) maximizes the left hand side above over
$\beta \in K$. Thus if $\hat{G}$  is discrete and $\tilde{\beta}$ is a
support 
point of $\hat{G}$ that is also in the interior of $K$, then the
gradient of the left hand side above (w.r.t $\beta$) should equal zero
at $\tilde{\beta}$. This proves the last claim of Proposition
\ref{prop:opt_condition}. 
\end{proof}

\subsection{Proof of Proposition \ref{exem_lemma}}

\begin{proof}[Proof of Proposition \ref{exem_lemma}]
  By the last claim of Proposition \ref{prop:opt_condition}, we have
  \begin{equation*}
    \nabla \left(\frac{1}{n} \sum_{i=1}^n
      \frac{f^{\beta}_{x_i}(y_i)}{f^{\hat{G}}_{x_i}(y_i)} - 1\right) = \mathbf{0} ,
  \end{equation*}
  where the gradient $\nabla$ is with respect to $\beta$ and is
  evaluated at $\beta = \tilde{\beta}$. Explicitly calculating the
  gradient, we get
  \begin{equation*}
    \frac{1}{n} \sum_{i=1}^n w_i(\tilde{\beta}) \left(x_i x_i^T
      \tilde{\beta} - x_i 
      y_i \right) = \mathbf{0} \qt{with $w_i(\tilde{\beta}) \propto
      \frac{f^{\tilde{\beta}}_{x_i}(y_i)}{f_{x_i}^{\hat{G}}(y_i)}$}. 
  \end{equation*}
  In other words, there exists a probability vector $(w_1, \dots,
  w_n)$ which satisfies
  \begin{equation*}
    \sum_{i=1}^n w_i x_i (y_i - x_i^{\T} \tilde{\beta}) = \mathbf{0},
  \end{equation*}
which implies that $\tilde{\beta} \in S(w)$, where $S(w)$ is defined in
\eqref{sw}. The above condition is equivalent to
\begin{equation*}
 \mathbf{0} \in \text{conv} \left\{x_1(y_1 - x_1^{\T} \tilde{\beta}), \dots,
   x_n(y_n - x_n^{\T} \tilde{\beta}) \right\}. 
\end{equation*}
As the right hand side above is a convex hull in $\mathbb{R}^p$,
Carath{\'e}odory's theorem guarantees the existence of a probability
vector $(w_1, \dots, w_n)$ with at most $p+1$ non-zero entries such
that
\begin{equation*}
  \mathbf{0} = \sum_{i=1}^n w_i x_i (y_i - x_i^{\T} \tilde{\beta}),
\end{equation*}
which is equivalent to $\tilde{\beta} \in S(w)$. This completes the
proof of Proposition \ref{exem_lemma}.

\end{proof}

\subsection{Proof of Theorem \ref{mainthem}} \label{subsection_mainthem}
The proof of Theorem \ref{mainthem} given below uses  the notion of
covering numbers and metric entropy which are defined as follows. Let
$T$ be a subset of a metric  
space with metric $\mathfrak{d}$. For $\eta > 0$, we say that a set
$S$ is an $\eta$-covering of $T$ if $\sup_{t \in T} \inf_{s \in S}
\mathfrak{d}(s, t) \leq \eta$. The smallest possible cardinality of an
$\eta$-covering of $T$ is known as the $\eta$-covering number of $T$
under the metric $\mathfrak{d}$ and this is denoted by $N(\eta, T,
\mathfrak{d})$. The logarithm of $N(\eta, T, \mathfrak{d})$ is called
the $\eta$-metric entropy of $T$ under $\mathfrak{d}$. When $T$ is a
subset of $\RS^p$  and the metric $\mathfrak{d}$ is the usual
Euclidean metric on $\RS^p$, we shall denote $N(\eta, T,
\mathfrak{d})$ by simply $N(\eta, T)$.

The proof of Theorem \ref{mainthem} given below is based on ideas
similar to those used in \citet{jiang2009general} and
\citet{saha2020nonparametric}. A key ingredient is the metric entropy
result stated as Theorem \ref{theorem:metric_entropy}.  Theorem
\ref{theorem:metric_entropy} is stated for the more general case of
possibly nonlinear regression functions $r(x, \beta)$. We take $r(x,
\beta) = x^{\T} \beta$ while applying Theorem
\ref{theorem:metric_entropy} in the proof below. 

\begin{proof}[Proof of Theorem~\ref{mainthem}]
	Let $S_0 := \{x : \|x\| \leq B\}$ so that $S_0$ contains all the
	design points $x_1, \dots, x_n$. Let 
	\begin{equation}
		\mathcal{M}_R  = \{ f^G_x(y) : \text{ any probability measure } G
		\text{ supported on } \ball_p(0,R) \}, 
		\label{eq:M_defn}
	\end{equation}
	where $\ball_p(0, R) := \{\beta \in \mathbb{R}^p : \|\beta\| \leq
	R\}$. Let $\|\cdot\|_{\infty}$ be the pseudometric on $\mathcal{M}_R$
	given by 
	\begin{equation*}
		(f^G, f^{G'}) \mapsto \sup_{x \in S_0, y \in \mathbb{R}} \left|f_x^G(y) - f_x^{G'}(y) \right|. 
	\end{equation*}
	
	Theorem \ref{theorem:metric_entropy}, which will be crucially used in
	this proof, gives an upper bound on the
	$\eta$-covering number $N(\eta, \mathcal{M}_{R}, \|\cdot\|_{\infty})$
	of  $\mathcal{M}_R$ under the pseudometric $\|\cdot\|_{\infty}$. For a
	fixed $\eta > 0$, let $\{ h^1, \dots, h^N\} \subseteq \mathcal{M}_R$ be an
	$\eta$-covering set of $\mathcal{M}_R$ under  $\| \cdot \|_{\infty}$
	where $N = N(\eta, \mathcal{M}_R, \| \cdot \|_{\infty})$. This ensures  
	\begin{equation}
		\underset{h \in \mathcal{M}_R}{\sup} \, \underset{1\leq j \leq N}{\inf} \| h - h^j\|_{\infty} \leq \eta.
		\label{eq:metric}
	\end{equation}
	For a fixed sequence $\{\gamma_n\}_{n \geq 1}$ and $t > 0$, let us now bound $\PP \{\mathfrak{H}_n(f^{\hat{G}},f^{G^*})\geq t \gamma_n\}$ (the precise form for $\gamma_n$ will be given later in the proof; it will equal a constant multiple of $\epsilon_n$).   
	
	We define a set $J\subseteq \{1,\dots,N\}$. Let $J$ be composed of all index $j \in \{ 1,\dots, N\}$ for which there exists $h^{0j}\in \mathcal{M}_R$ satisfying
	\begin{equation}
		\| h^{0j} - h^j \|_{\infty, S_0\times \RS } \leq \eta \quad \text{ and } \quad \mathfrak{H}_n(h^{0j}, f^{G^*}) \geq t\gamma_n.
		\label{eq:J_defn}
	\end{equation}
	Let $j \in \{1, \dots, N\}$ be such that $\| h^j - {f}^{\hat{G}} \|_{\infty} \leq \eta$ (such a $j$ clearly exists because $h^1, \dots, h^N$ form an $\eta$-covering set of $\mathcal{M}_R$). Now if $\mathfrak{H}_n(f^{\hat{G}}, f^{G^*}) \geq t\gamma_n$, then $j \in J$ and consequently $\| {f}^{\hat{G}} - h^{0j}\|_{\infty} \leq 2\eta$ which implies that 
	\[
	f^{\hat{G}}_{x_i}(y) \leq h^{0j}_{x_i}(y) + 2\eta \qt{for all $i = 1, \dots, n$ and $y \in \mathbb{R}$}. 
	\]
	Therefore, we have
	\[
	\begin{split}
		\prod_{i=1}^n f^{{G}^*}_{x_i}(\Y_i) \leq \prod_{i=1}^n f^{\hat{G}}_{x_i}(\Y_i) &\leq \prod_{i = 1}^n\{ h^{0j}_{x_i}(\Y_i) + 2\eta\} \leq \max_{j \in J} \prod_{i = 1}^n\{ h^{0j}_{x_i}(\Y_i) + 2\eta\},
	\end{split}
	\]
	where the first inequality follows from the fact that $\hat{G}$ maximizes the likelihood. We thus get 
	\begin{equation*}
		\begin{split}
			\PP(\mathfrak{H}_{\mathrm{fixed}}(f^{\hat{G}}, f^{G^*}) \geq t
			\gamma_n)   
			&\leq  \PP\left\{ \max_{j\in J} \prod_{i=1}^n \frac{h^{0j}_{x_i}(\Y_i) + 2 \eta }{f^{G^*}_{x_i}(\Y_i)} \geq 1 \right\} \\
			&\leq 
			\sum_{j \in J} \PP \left\{\prod_{i=1}^n \frac{h^{0j}_{x_i}(\Y_i) + 2\eta }{f^{G^*}_{x_i}(\Y_i)} \geq 1 \right\} \\
			& \leq \sum_{j\in J} \E \prod_{i=1}^n \sqrt{ \frac{h^{0j}_{x_i}(\Y_i) + 2 \eta}{f^{G^*}_{x_i}(\Y_i)}}  = \sum_{j\in J}\prod_{i=1}^n \E \sqrt{ \frac{h^{0j}_{x_i}(\Y_i) + 2 \eta}{f^{G^*}_{x_i}(\Y_i)}},
		\end{split}
	\end{equation*}
	where we used the union bound in the second line and Markov's inequality (followed by the independence of $Y_1, \dots, Y_n$) in the third line. For each $j \in J$, 
	\[
	\begin{split}
		\prod_{i=1}^n  \E \sqrt{ \frac{h^{0j}_{x_i}(\Y_i) + 2 \eta}{f^{G^*}_{x_i}(\Y_i)}} & = \exp\left( \sum_{i=1}^n \log \E \sqrt{ \frac{h^{0j}_{x_i}(\Y_i) + 2 \eta}{f^{G^*}_{x_i}(\Y_i)}} \right)\\
		& \leq \exp\left(\sum_{i=1}^n  \E \sqrt{ \frac{h^{0j}_{x_i}(\Y_i) + 2 \eta}{f^{G^*}_{x_i}(\Y_i)}} - n \right) \\
		& = \exp\left(\sum_{i=1}^n \int \sqrt{(h^{0j}_{x_i} + 2\eta) f^{G^*}_{x_i}} - n \right),
	\end{split}
	\]
	where we used the inequality $\log a \leq a -1$ in the second line, and the last equality follows from the fact that $Y_i$ has density $f_{x_i}^{G^*}$. The simple inequality $\sqrt{a+b} \leq \sqrt{a} + \sqrt{b}$ now gives, for each $1 \le i \le n$, 
	\[
	\begin{split}
		\int \sqrt{(h^{0j}_{x_i}+2\eta) f^{G^*}_{x_i}} &\leq 
		\int \sqrt{h^{0j}_{x_i} f^{G^*}_{x_i}}+ \sqrt{2 \eta} \int \sqrt{f^{G^*}_{x_i}} \\
		&\leq 1 - \frac{1}{2} \mathfrak{H}^2(h^{0j}_{x_i}, f^{G^*}_{x_i}) + \sqrt{2 \eta} \sqrt{\int f^{G^*}_{x_i}} = 1 - \frac{1}{2} \mathfrak{H}^2(h^{0j}_{x_i}, f^{G^*}_{x_i}) + \sqrt{2 \eta}. 
	\end{split}
	\]
	As a result, we deduce
	\[
	\sum_{i=1}^n \int \sqrt{(h^{0j}_{x_i}+2 \eta ) f^{G^*}_{x_i}} \leq n - \frac{1}{2} \sum_{i=1}^n \mathfrak{H}^2(h^{0j}_{x_i}, f^{G^*}_{x_i})+ n \sqrt{2  \eta} .
	\]
	As we have assumed that for every $j \in J$, 
	\[
	\sum_{i=1}^n \mathfrak{H}^2(h^{0j}_{x_i}, f^{G^*}_{x_i}) = n
	\mathfrak{H}^2_{\mathrm{fixed}}(h^{0j}, f^{G^*}) \geq n t^2
	\gamma_n^2, 
	\]
	we obtain
	\[
	\sum_{i=1}^n \int \sqrt{(h^{0j}_{x_i}+2v_i) f^{G^*}_{x_i}} \leq n - \frac{n}{2} t^2 \gamma_n^2+ n\sqrt{2\eta}.
	\]
	We have thus proved
	\[
	\begin{split}
		\prod_{i=1}^n  \Expt \sqrt{ \frac{h^{0j}_{x_i}(\Y_i) + 2 \eta}{f^{G^*}_{x_i}(\Y_i)}} & \leq \exp\left(\sum_{i=1}^n \int \sqrt{(h^{0j}_{x_i} + 2\eta) f^{G^*}_{x_i}} - n \right) \leq \exp(- \frac{n}{2} t^2 \gamma_n^2+ n\sqrt{2\eta}),
	\end{split}
	\]
	which gives (note that $|J| \leq N$)
	\begin{align}
		\PP \left\{\mathfrak{H}_{\mathrm{fixed}}(f^{\hat{G}}, f^{G^*}) \geq t
		\gamma_n \right\}  
		&\leq  |J| \cdot \exp\left( - \frac{n}{2} t^2\gamma_n^2 + n \sqrt{2 \eta}\right) \nonumber \\
		& \leq  \exp\left(\log N- \frac{n}{2} t^2\gamma_n^2 + n \sqrt{2 \eta}\right). \label{hellbound}
	\end{align} 
	We now use the metric entropy result in
	Theorem~\ref{theorem:metric_entropy} to bound $\log N$. Setting $S_0 =
	\{x: \|x\| \leq B\}$ and $K =  \left\{\beta \in \mathbb{R}^p :
	\|\beta\| \leq R \right\}$  in Theorem~\ref{theorem:metric_entropy},
	we get  
	\[
	\log N(\eta, \mathcal{M}_R, \| \cdot \|_{\infty}) \leq C_p \zeta^p N(\{2 \log(3\sigma^{-1}\eta^{-1})\}^{1/2} \sigma/\lipsz,\{\beta : \|\beta\| \leq R \}) \{\log(\sigma^{-1}\eta^{-1})\}^{p+1},
	\]
	where $\lipsz = \sup_{x \in S_0} \lipsz(x)$ and $\lipsz(x)$ is
	defined in 
	\eqref{eq:lipsz_defn_appendix}. It is clear that for the linear model,
	$\zeta = 1$ and $\lipsz(x) \leq \|x\| \leq B$ (note that we have
	made the assumption $\max_{1 \leq i \leq n} \|x_i\| \leq B$). The
	Euclidean covering number $N(\{2 \log(3\sigma^{-1}\eta^{-1})\}^{1/2}
	\sigma/\lipsz,\{\beta : \|\beta\| \leq R \})$ is bounded in the
	following way. It is well-known that  
	\[
	N\left(\epsilon, \{\beta \in \mathbb{R}^p : \|\beta\| \leq R\} \right) \leq \left( 1+ \frac{2R}{\epsilon}  \right)^p \qt{for all $ \epsilon > 0$},
	\]
	and consequently
	\[
	\begin{split}
		N(\{2 \log(3\sigma^{-1}\eta^{-1})\}^{1/2} \sigma/\lipsz, \{\beta : \|\beta\| \leq R\}) & \leq 
		\left(  1+ \frac{2R\lipsz}{\{2 \log(3\sigma^{-1}\eta^{-1})\}^{1/2} \sigma}\right)^p. 
	\end{split}
	\]
	This and the fact that $\lipsz \leq B$ lead to   
	\begin{align}
		\log N &=  \log N(\eta, \mathcal{M}_R, \| \cdot \|_{\infty})
		\nonumber\\ &\leq C_p  \left(1+ \frac{2 R B}{\{2
			\log(3\sigma^{-1}\eta^{-1})\}^{1/2}
			\sigma}\right)^p
		\{\log(\sigma^{-1}\eta^{-1})\}^{p+1} \nonumber \\ 
		&\leq C_p \{\log(\sigma^{-1}\eta^{-1})\}^{p+1} + C_p
		\left(\frac{RB}{\sigma} \right)^p \{\log (3 \sigma^{-1} \eta^{-1})
		\}^{p/2+1}, \label{conv.entbound}
	\end{align}
	where $C_p$ absorbs a coefficient $2^p$ in the last line. Using the above in
	\eqref{hellbound}, we obtain  
	\begin{equation*}
		\begin{split}
			\PP \left\{\mathfrak{H}_{\mathrm{fixed}}(f^{\hat{G}}, f^{G^*})
			\geq t \gamma_n \right\} &\leq \exp \left(  C_p
			\{\log(\sigma^{-1}\eta^{-1})\}^{p+1} \right. \\ &\left. + C_p
			\left(\frac{RB}{\sigma} \right)^p \{\log (3 \sigma^{-1}
			\eta^{-1}) \}^{p/2+1} - \frac{n}{2} t^2\gamma_n^2 + n \sqrt{2
				\eta}\right).   
		\end{split}
	\end{equation*}
	We shall now take $\gamma_n$  and $\eta$ so that
	\begin{equation}\label{gamma.crit}
		n \gamma_n^2 \geq 12 \max \left(C_p \{\log(\sigma^{-1}\eta^{-1})\}^{p+1}, C_p \left(\frac{RB}{\sigma} \right)^p \{\log (3 \sigma^{-1} \eta^{-1})  \}^{p/2+1}, n \sqrt{2 \eta} \right). 
	\end{equation}
	This will ensure that, for $t \geq 1$,
	\begin{equation}\label{hell.tail}
		\PP \left\{\mathfrak{H}_{\mathrm{fixed}}(f^{\hat{G}}, f^{G^*}) \geq
		t \gamma_n \right\}  \leq \exp \left(\frac{n \gamma_n^2}{4} (1 - 2
		t^2) \right)  \leq \exp \left(-\frac{n t^2 \gamma_n^2}{4} \right).  
	\end{equation}
	To satisfy \eqref{gamma.crit}, we first take $\eta := \gamma_n^4/288$ (so that $12 n \sqrt{2 \eta} = n \gamma_n^2$). The quantity $\gamma_n$  will then have to satisfy the two inequalities:
	\begin{equation}\label{cond.1}
		n \gamma_n^2 \geq 12 C_p \left(\log \frac{288}{\sigma \gamma_n^4} \right)^{p+1},
	\end{equation}
	and
	\begin{equation}\label{cond.2}
		n \gamma_n^2 \geq 12 C_p \left(\frac{RB}{\sigma} \right)^{p} \left(\log \frac{864}{\sigma \gamma_n^4} \right)^{p/2 + 1}. 
	\end{equation}
	It is now elementary to check that \eqref{cond.1} is satisfied whenever
	\begin{equation*}
		\gamma_n \geq \sqrt{\frac{12C_p}{n}} \left(\Log  \frac{2n^2}{\sigma C_p^2} \right)^{(p+1)/2} 
	\end{equation*}
	and \eqref{cond.2} is satisfied whenever
	\begin{equation*}
		\gamma_n \geq \sqrt{\frac{12 C_p}{n}} \left(\frac{R B}{\sigma} \right)^{p/2} \left(\Log \frac{6n^2 \sigma^{2p}}{\sigma C_p^2 (R B)^{2p}} \right)^{(p/4) + (1/2)} ,
	\end{equation*}
	where we used the notation $\Log x := \max(1, \log x)$.
	
	We may now assume $C_p \geq \sqrt{6}$. It is then easy to see that both the above inequalities and consequently both \eqref{cond.1} and \eqref{cond.2} are satisfied whenever  
	\begin{equation*}
		\gamma_n \geq \sqrt{\frac{12 C_p}{n}} \max \left(\left(\Log  \frac{n^2}{\sigma} \right)^{\frac{p+1}{2}} , \left(\frac{R B}{\sigma} \right)^{\frac{p}{2}} \left(\Log \frac{n^2 \sigma^{2p}}{\sigma (R B)^{2p}} \right)^{\frac{p}{4} + \frac{1}{2}}  \right) .
	\end{equation*}
	Using $\Log x^2 \leq 2 \Log x$ and absorbing all the $p$-dependent constants in $C_p$, we deduce that inequality \eqref{hell.tail}  holds for  $\gamma_n = \sqrt{C_p} \epsilon_n$ where $\epsilon_n$ is defined in \eqref{eq:epsilon_defn}. This completes the proof of \eqref{tailbound} (note that $\exp(-nt^2C_p\epsilon_n^2/4)$ can be bounded by $\exp(-nt^2\epsilon_n^2)$ by taking $C_p$ larger than 4). 
	
	To prove \eqref{expectationbound}, we multiply both sides of \eqref{tailbound} by $t$ and integrate from $t = 1$ to $t = \infty$ to obtain
	\begin{equation*}
		\E \left(\frac{\mathfrak{H}_{\mathrm{fixed}}^2(f^{\hat G},
			f^{G^*})}{C_p \epsilon_n^2} - 1 \right)_+ \leq
		\frac{1}{n\epsilon_n^2}, 
	\end{equation*}
	where $x_+ := \max(x, 0)$ which implies
	\begin{equation*}
		\E \mathfrak{H}_{\mathrm{fixed}}^2(f^{\hat G}, f^{G^*}) \leq C_p
		\epsilon_n^2 + \frac{C_p}{n}.  
	\end{equation*}
	This proves \eqref{expectationbound} (after changing $C_p$ to $2 C_p$) as $\epsilon_n^2 \geq n^{-1}$.  
\end{proof}

\subsection{Proof of Theorem \ref{them:randomdesign_predictionerror}} \label{appendix:random_design} 
The proof of Theorem \ref{them:randomdesign_predictionerror} uses the
following result from the theory of empirical processes which follows
from \citet[Proof of Lemma 5.16]{geer2000empirical}.
\begin{lemma}\label{thm:vdg}
	Suppose $x_1,\dots, x_n$ are independently distributed according to
	a probability distribution $\mu$ and suppose $\mathcal{G}$ is a
	class of functions on the support of $\mu$ that are uniformly
	bounded by 1. Then  
	\begin{equation}\label{vdg}
		\mathbb{P} \left\{\sup_{g \in \mathcal{G}} \left(\sqrt{\int g^2
			d\mu} - 2 \sqrt{ \frac{1}{n} \sum_{i=1}^n g^2(x_i)} 
		\right) 
		> 4 \epsilon 
		\right\} \leq 4 \exp \left(-\frac{n \epsilon^2}{768} \right)
	\end{equation}
	provided $\epsilon > 0$ satisfies
	\begin{equation}\label{vdg.con}
		n \epsilon^2 \geq 768 \log N_{[]}(\epsilon, \mathcal{G}, L_2(\mu)).  
	\end{equation}
	Here $N_{[]}(\epsilon, \mathcal{G}, L_2(\mu))$ denotes the
	$\epsilon$-bracketing number of $\mathcal{G}$ in the $L_2(\mu)$
	metric defined as the smallest number of pairs of functions $g_j^L,
	g_j^U$ satisfying $\|g_j^U - g_j^L\|_{L_2(\mu)} \leq \epsilon$ and
	the property that every $g \in \mathcal{G}$ is sandwiched between
	one such pair (i.e., $g_j^L \leq g \leq g_j^U$ for some $j$).  
\end{lemma}

\begin{proof}[Proof of
	Theorem~\ref{them:randomdesign_predictionerror}] 
	We shall use Lemma \ref{thm:vdg} with $\mathcal{G}$ equal to the class
	of all functions 
	\begin{equation*}
		x \mapsto \frac{1}{2}\mathfrak{H}^2(f_x^G, f_x^{G^*})
	\end{equation*}
	on the set $S_0 := \{x \in \mathbb{R}^p : \|x\| \leq B\}$ as $G$
	ranges over the class of all probability measures on $\{\beta \in
	\mathbb{R}^p : \|\beta\| \leq R\}$. Note that the function above is
	uniformly bounded by 1. The key to the application of Lemma
	\ref{thm:vdg} is to bound $N_{[]}(\epsilon, \mathcal{G}, L_2(\mu))$
	and for this, we use the inequality:   
	\begin{equation}\label{prev.brak}
		N_{[]}(\epsilon, \mathcal{G}, L_2(\mu)) \le
		N\left(\frac{\epsilon^2}{4T_{G^*}}, \mathcal{M}_R,
		\|\cdot\|_{\infty}\right),
	\end{equation} 
	where $\mathcal{M}_R$ is as in \eqref{eq:M_defn},  
	\begin{equation*}
		T_{G^*} := \int \left(\int \sqrt{f_x^{G^*}(y)} dy \right)^2
		d\mu(x)
	\end{equation*}
	and $\|\cdot\|_{\infty}$ is the $L_{\infty}$ metric on the set
	$S_0 \times \mathbb{R}$. To prove \eqref{prev.brak}, let $\eta :=
	\epsilon^2/(4T_{G^*})$  and let $\{(x, y) \mapsto h_j(x, y), j = 1,
	\dots, N\}$ be an $\eta$-covering set of 
	$\mathcal{M}_R$ under the $L_{\infty}$-metric on $S_0 \times
	\mathbb{R}$. This means that for every probability measure $G$ on 
	$\{\beta \in \mathbb{R}^p : \|\beta\| \leq R\}$, there exists $1 \leq
	j \leq N$ such that
	\begin{equation*}
		\sup_{x \in S_0, y \in \mathbb{R}}  \left|f_x^G(y) - h_j(x, y)
		\right| \leq \eta,
	\end{equation*}
	which implies that $h_j(x, y) - \eta \leq f_x^G(y) \leq h_j(x, y) +
	\eta$ for all $x \in S_0, y \in \mathbb{R}$. As a result 
	\begin{equation*}
		\frac{1}{2} \int \left(\sqrt{f_x^{G}(y)} - \sqrt{f_x^{G^*}(y)}
		\right)^2 dy = 1 - \int \sqrt{f_x^G(y)} \sqrt{f_x^{G^*}(y)} dy 
	\end{equation*}
	lies in the interval
	\begin{equation*}
		\left[1 - \int \sqrt{h_j(x, y) + \eta} \sqrt{f_{x}^{G^*}(y)} dy, 1 + 
		\int \sqrt{\left(h_j(x, y) - \eta \right)_+} \sqrt{f_x^{G^*}(y)
			dy} \right],
	\end{equation*}
	where $x_+ := \max(x, 0)$. The squared $L_2$ distance between the two
	end points of the above interval equals
	\begin{equation}\label{l2dist}
		\int \left[\int \left( \sqrt{h_j(x, y) + \eta} - \sqrt{(h_j(x, y) -
			\eta)_+} \right) \sqrt{f_x^{G^*}(y)} dy \right]^2 d\mu(x).
	\end{equation}
	Because $\sqrt{a + \eta} - \sqrt{(a - \eta)_+} \leq 2 \sqrt{\eta}$ for
	all $a > 0, \eta > 0$, we can bound \eqref{l2dist} by
	\begin{equation*}
		4 \eta \int \left( \int \sqrt{f_x^{G^*}(y)} dy \right)^2 d\mu(x) = 4
		\eta T_{G^*} = \epsilon^2
	\end{equation*}
	and this proves \eqref{prev.brak}. 
	
	The quantity $T_{G^*}$ is bounded from above by a finite constant
	depending only on $\sigma, B$ and $R$ because of the following
	argument.   
	\begin{align*}
		T_{G^*} &\leq \sup_{x : \|x\| \leq B} \left(\int \sqrt{f_x^{G^*}(y)}
		dy \right)^2 \\ &\leq \sup_{x : \|x\| \leq B} \left(\int I\{|y| < 2 B
		R\}  \sqrt{f_x^{G^*}(y)} dy + \int I\{|y| \geq 2 B R\}
		\sqrt{f_x^{G^*}(y) dy} \right)^2. 
	\end{align*}
	For $|y| < 2 B R$, we use the trivial inequality
	\begin{equation*}
		f_x^{G^*}(y) = \frac{1}{\sqrt{2 \pi} \sigma}\int  \exp
		\left(-\frac{(y - x^{\top} \beta)^2}{2 \sigma^2} 
		\right) dG^*(\beta) \leq \frac{1}{\sqrt{2 \pi} \sigma}
	\end{equation*}
	and for $|y| \geq 2 R B$, we use 
	\begin{equation*}
		f_x^{G^*}(y) = \frac{1}{\sqrt{2 \pi} \sigma}\int  \exp
		\left(-\frac{(y - x^{\top} \beta)^2}{2 \sigma^2} 
		\right) dG^*(\beta) \leq \frac{1}{\sqrt{2 \pi} \sigma} \exp
		\left(-\frac{y^2}{8 \sigma^2} \right). 
	\end{equation*}
	which is true because (note that $G^*\{\beta :
	\|\beta\| \leq R\} = 1$)
	\begin{equation*}
		|y - x^{\top} \beta| \geq |y| - |x^{\top} \beta| \geq |y| - \|x\|
		\|\beta\| \geq |y| - R B \geq |y|/2. 
	\end{equation*}
	We thus get
	\begin{equation*}
		T_{G^*} \leq \frac{1}{\sqrt{2 \pi} \sigma}\left(4 R B + 2\int_{2 R
			B}^{\infty} \exp \left(-\frac{y^2}{16 \sigma^2} dy \right) 
		\right)^2 \leq \frac{C}{\sigma} \left(RB + \sigma \right)^2
	\end{equation*}
	for a universal positive constant $C$.
	
	Using \eqref{conv.entbound}, the covering number $N(\eta,
	\mathcal{M}, \|\cdot\|_{\infty})$ is bounded by
	\begin{equation*}
		\log N(\eta, \mathcal{M}, \|\cdot\|_{\infty}) \leq C_p \max \left(
		\left\{\log(\sigma^{-1} \eta^{-1}) \right\}^{p+1}, 
		\left(\frac{RB}{\sigma} \right)^p \left\{\log (3 \sigma^{-1}
		\eta^{-1}) \right\}^{p/2 + 1} \right)
	\end{equation*}
	for a positive constant $C_p$ depending on $p$ alone. Inequality
	\eqref{prev.brak} then gives
	\begin{equation*}
		\log N_{[]}(\epsilon, \mathcal{G}, L_2(P)) \leq C_p  \max \left(  \left\{\log(\sigma^{-1}
		\epsilon^{-2} T) \right\}^{p+1}, 
		\left(\frac{R B}{\sigma} \right)^p \left\{\log (3 \sigma^{-1}
		\epsilon^{-2} T) \right\}^{p/2 + 1}   \right),
	\end{equation*}
	where $T = T_{G^*}$.
	
	The condition \eqref{vdg.con} will therefore be satisfied provided
	(below $C_p$ equals 768 multiplied by the constant $C_p$ appearing in
	the above equation)
	\begin{equation*}
		n \epsilon^2 \geq C_p  \left\{\log(\sigma^{-1}
		\epsilon^{-2} T) \right\}^{p+1} ~~ \text{ and } ~~ n \epsilon^2
		\geq C_p \left(\frac{R B}{\sigma} \right)^p \left\{\log (3 \sigma^{-1}
		\epsilon^{-2} T) \right\}^{p/2 + 1}.
	\end{equation*}
	It is clear that both of these conditions will be satisfied for
	$\epsilon^2 \geq C_p \beta_n^2$ where $\beta_n$ is given by
	\eqref{eq:beta_defn}.  
	
	Lemma \ref{thm:vdg}  then gives that, for each $t \geq 1$, 
	\begin{equation*}
		\PP \left\{  \left(\frac{1}{2} \int \mathfrak{H}^2 \left(f_{x}^{\hat{G}},
		f_x^{G^*} \right) d\mu(x) \right)^{1/2} \geq 2
		\left(\frac{1}{2n} \sum_{i=1}^n \mathfrak{H}^2
		\left(f_{x_i}^{\hat{G}}, f_{x_i}^{G^*} \right) \right)^{1/2} + 4 t
		\beta_n \sqrt{C_p} \right\} \leq \exp \left(-\frac{nt^2
			\beta_n^2}{C_p} \right). 
	\end{equation*}
	The inequalities \eqref{tailbound.random} and
	\eqref{expectationbound.random} both follow from combining the above
	inequality with \eqref{tailbound} and \eqref{expectationbound}
	respectively in Theorem \ref{mainthem}.  
\end{proof}

\begin{revise}
\subsection{Proof of Theorem \ref{thm:randomdesign_identifiability}} \label{appendix:randomdesign_identifiability}

The identifiability result (Theorem  \ref{thm:randomdesign_identifiability}) is proved using the tools of characteristic functions. A key step in the proof uses the properties of analytic functions, and we need the following basic fact in Lemma~\ref{lemma:analytic_func}.

\begin{lemma}

For any probability measures $G$ over $\left\{\beta \in \mathbb{R}^p : \|\beta\| \le R \right\}$ and $x$ is a $p$-dimensional variable, $\Expt_{\beta\sim G} e^{ix^\T\beta}$ is analytic in each component of $x$.

\label{lemma:analytic_func}
\end{lemma}

\begin{proof}[Proof of Lemma~\ref{lemma:analytic_func}]

	We prove that for each component $x_{(j)}$ of $x$, $\Expt_{\beta\sim G} e^{ix^\T\beta}$ is an analytic function in $x_{(j)}$, $j = 1, \dots, p$. For any $C^1$ closed curve $\Gamma$, because the boundedness of $G$, we can adopt the Fubini's Theorem to exchange the integral order,
	\[
	\int_{\Gamma} \Expt_{\beta \sim G} e^{ix^\T \beta} d x_{(j)} = \Expt_{\beta \sim G}\left[
		\int_{\Gamma} e^{ix^\T\beta} d x_{(j)}.	
	 \right]
	\]
	By Cauchy's integral theorem, $\int_{\Gamma} e^{ix^\T\beta} d x_{(j)} = 0$ since $e^{ix^\T\beta}$ is analytic in $x_{(j)}$. Plugging back to the integral above, 
	\[
	\int_{\Gamma} \Expt_{\beta \sim G} e^{ix^\T \beta} d x_{(j)}  = 0,
	\]
	and therefore by  Morera's theorem in complex analysis \citep[Theorem 5.1, Chapter 2]{stein2010complex},  $\Expt_{\beta\sim G} e^{ix^\T\beta}$ is analytic in $x_{(j)}$.
	
\end{proof}

\begin{proof}[Proof of Theorem~\ref{thm:randomdesign_identifiability}]

Let $p_\mu$ denote the density function of $\mu$, then 
\[
\int \frac{1}{\sigma} \phi\left(\frac{y - x^\T\beta}{\sigma} \right) d G_1(\beta) \cdot p_\mu(x) = 
\int  \frac{1}{\sigma} \phi\left(\frac{y - x^\T\beta}{\sigma} \right) d G_2(\beta) \cdot p_\mu(x). 
\]

That is, the joint distributions of  $(X,Y)$ from the following two data generating mechanisms are the same,
\begin{enumerate}
\item $X\sim \mu$, $Y = X^\T \beta + \sigma Z$, $\beta \sim G_1$, and $Z \in N(0,1)$;
\item $X\sim \mu$, $Y = X^\T \beta + \sigma Z$, $\beta \sim G_2$, and $Z \in N(0,1)$	.
\end{enumerate}

Therefore, the characteristic functions are also the same, i.e., 
\[
\int e^{i u^\T x} \Expt e^{it \sigma Z} \Expt_{\beta\sim G_1} e^{it x^\T \beta} d \mu(x)
= \int e^{i u^\T x} \Expt e^{it \sigma Z} \Expt_{\beta \sim G_2} e^{it x^\T \beta} d \mu(x)
\]
for all $u \in \RS^p$ and $t \in \RS$. By Fourier inversion theorem, we have
\[
\Expt e^{it\sigma Z} \Expt_{\beta \sim G_1} e^{it x^\T \beta} 
= \Expt e^{it\sigma Z} \Expt_{\beta \sim G_2} e^{it x^\T \beta}.
\]

Since $\Expt e^{it\sigma Z} \neq 0$, 
\[
\Expt_{\beta \sim G_1} e^{it x^\T \beta} 
=  \Expt_{\beta \sim G_2} e^{itx^\T \beta}
\]
holds for all $t\in \RS$ and all $x$ in the support of $\mu$.

By Lemma \ref{lemma:analytic_func},  both $\Expt_{\beta \sim G_1} e^{it x^\T \beta}$
and $\Expt_{\beta \sim G_2} e^{itx^\T \beta}$ are analytic functions in each component of $x$. Combining with the fact that the support of $\mu$ contains an open set, it follows from the Identity theorem of analytic functions that 
\[
\Expt_{\beta \sim G_1} e^{it x^\T \beta} 
=  \Expt_{\beta \sim G_2} e^{itx^\T \beta}
\]
holds for all $t\in \RS$ and all $x \in \RS^p$.

We can view $(tx)$ as one variable,, the above equality essentially shows that $G_1$, $G_2$ have the same characteristic functions, and thus $G_1 = G_2$.

\end{proof}
\end{revise}

\subsection{Proof of Theorem
	\ref{thm:weakconsistency}} \label{appendix:weakconsistency}

The proof of Theorem~\ref{thm:weakconsistency} relies on
Theorem~\ref{them:randomdesign_predictionerror}. It also uses the
following lemma whose proof is similar to \citet[Proposition
2.2]{beran1994minimum}. We recall that the mixture of linear
regression model under random design can be expressed as  
\begin{equation}
	Y_i = X_i^\T \beta^i + \sigma Z_i, \beta^i \sim G^*,X_i \sim \mu, Z_i\sim N(0,1).
	\label{eq:mlrmodel_random}
\end{equation}
Let $P(G^*, \mu)$ denote the joint distribution of $(X_i, Y_i)$ under
the above model. We use $\hat{G}_n$ to denote an NPMLE  given $n$ data
points. Let $\dlp$ denote  the L\'{e}vy–Prokhorov metric, which is known
to metrize the weak convergence of probability measures.

\begin{lemma}
	\label{strongid} 
	Assume the support of $\mu$ contains an open set, if
	\[
	\dlp(P(G_n, \mu ), P(G^*, \mu))\rightarrow 0,
	\]
	where $\{ G_{n}\}$ denotes a sequence of probability measures such that $G_n\{ \beta\in \RS^p: \Vert \beta \Vert \leq R\} = 1$,  then
	\[
	\dlp(G_{n},G^*) \rightarrow 0.
	\]
\end{lemma}
\begin{proof}[Proof of Lemma~\ref{strongid}]
	Because $\{G_n\}$ is supported on a compact ball, $\{G_n\}$ is tight,
	and $\{G_n\}$ has a subsequence $\{G_{n_m}\}$ converging weakly
	(Theorem 3.10.3 in \citet{durrett2019probability}). Let $\tilde{G}$
	denote the limiting probability measure of the weakly convergent
	subsequence, then  
	\begin{equation*}
		\lim_{m\rightarrow \infty}E_{\beta \sim G_{n_m}}e^{itx^\T \beta}
		=\E_{\beta \sim \tilde{G}}e^{itx^\T \beta} \text{ for all } x \in
		\RS^p \text{ and } t \in
		\RS. 
	\end{equation*}
	Meanwhile, the weak convergence of $P(G_n, \mu )$ to $P(G^*, \mu)$ implies 
	\begin{equation*}
		\lim_{m\rightarrow \infty}\int e^{iu^\T x} \E e^{it\sigma Z}
		\E_{\beta \sim G_{n_m}}e^{itx^\T \beta} \diff \mu(x) = \int e^{iu^\T x} \E
		e^{it\sigma Z} \E_{\beta \sim G^*}e^{itx^\T \beta} \diff \mu(x) 
	\end{equation*}
	for all $u \in \RS^p$ and $t \in \RS$. Combining the above two
	equations, we get  
	\[
	\int e^{iu^\T x} \E e^{it \sigma Z} \E_{\beta \sim \tilde{G}}e^{itx^\T \beta} \diff\mu(x)
	= \int e^{iu^\T x} \E e^{it \sigma Z} \E_{\beta \sim G^*}e^{itx^\T \beta} \diff\mu(x)
	\text{ for all } u \in \RS^p \text{ and } t \in \RS. 
	\]
	The Fourier inversion theorem now gives,  
	\begin{equation}
		\E_{\beta\sim \tilde{G}}e^{itx^\T \beta}  = \E_{\beta \sim
			G^*}e^{itx^\T \beta} \text { for all } t \in \RS \text{ and $x$ in
			the support of $\mu$}. 
		\label{analytic}
	\end{equation}
	Both sides of equation \eqref{analytic} are bounded and thus analytic
	in each component of $x$, as previously shown in Lemma~\ref{lemma:analytic_func}. Furthermore, since the support of $\mu$ is assumed
	to contain an open set, \eqref{analytic} holds for all $x \in 
	\RS^p$. Alternatively, by viewing $(tx)$ as the argument of
	characteristic functions, \eqref{analytic} shows that $\tilde{G}$ and
	$G^*$ have the same characteristic functions and thus $\tilde{G} =
	G^*$. 
	
	Therefore, we have shown that every weakly convergent subsequence of
	$\{G_n\}$weakly  converges to $G^*$. Suppose  that $\{G_n\}$ does not
	converge weakly to $G^*$, then there exists $\epsilon>0$, for every
	$n$ there exists $n_k \geq n$ such that $d(G_{n_k},G^*) >\epsilon$. It
	is clear that any subsequence of $\{G_{n_k}\}$ cannot converge weakly
	to $G^*$. However, following the same argument before, $\{G_{n_k}\}$
	is tight and contains a weakly convergent subsequence converging to
	$G^*$ leading to a contradiction. This completes the proof of Lemma
	\ref{strongid}. 
\end{proof}

We are now ready to prove Theorem \ref{thm:weakconsistency}.

\begin{proof}[Proof of Theorem~\ref{thm:weakconsistency}]
	
	Based on \eqref{tailbound.random} in
	Theorem~\ref{them:randomdesign_predictionerror},
	$\mathfrak{H}_{\mathrm{random}}^2(f^{\hat{G}_n}, f^{G^*})$ 
	converges to $0$ in probability. We first notice that
	$\mathfrak{H}_{\mathrm{random}}^2(f^{\hat{G}_n}, f^{G^*})$ is exactly
	the Hellinger distance between $P( \hat{G}_n, \mu)$ and $P( G^*,
	\mu)$. Since convergence under Hellinger distance is stronger then
	weak convergence, we have 
	\[
	\dlp(P( \hat{G}_n, \mu), P( G^*, \mu)) \rightarrow 0
	\]
	in probability. We now invoke a classic probability result
	(Theorem~2.3.2 in \citet{durrett2019probability}): given random
	variables $\{D_n\}$ and $D$, $D_n\rightarrow D$ in probability if and
	only if for every subsequence $\{D_{n_{m}} \}$, there is a further
	subsequence $\{D_{n_{m_k}}\}$ converges almost surely to $D$.
	Consider the random sequences $\{ \dlp(\hat{G}_{n}, G^*)  \}$ and
	$\{\dlp(P(\hat{G}_n, \mu), P(G^*, \mu)) \}$, for any subsequence
	$\{\dlp(\hat{G}_{n_{m}}, G^*)\}$, there  
	is a further subsequence \[\{\dlp(P( \hat{G}_{n_{m_k}}, \mu), P( G^*,
	\mu)) \}\] that converges to $0$ almost surely because $\dlp(P( \hat{G}_n,
	\mu), P( G^*, \mu)) \rightarrow 0$ in probability and consequently
	$\dlp(\hat{G}_{n_{m_k}}, G^*) \rightarrow 0$ almost surely because of
	Lemma~\ref{strongid}. Thus we have shown that $\dlp(\hat{G}_{n},
	G^*)\rightarrow 0$ in probability. 
\end{proof}

\subsection{Metric Entropy Result: Theorem
	\ref{theorem:metric_entropy} and its proof}
\label{appendix:metric}
In this section, we prove our metric entropy results, and these
results provide key ingredients for the proof of
Theorem~\ref{mainthem} and
Theorem~\ref{them:randomdesign_predictionerror}. The main theorem of
this section is Theorem~\ref{theorem:metric_entropy}. We work here
under a more general  setting than linear regression
functions. Specifically, we use the function $r(x,\beta)$ to represent
the mean of the response $y$ given $x$ and $\beta$ so that the
conditional density function of $y$ given $x$ is  
\[
f^G_x(y) : = \int \frac{1}{\sigma}\phi\left(\frac{y -
    r(x,\beta)}{\sigma}\right) \diff G(\beta). 
\]
Although our main example is $r(x, \beta) = x^{\T} \beta$, Theorem
\ref{theorem:metric_entropy} can be used for other functions $r(x,
\beta)$  as well. 

Let $K$ denote an arbitrary compact set in $\RS^p$ and 
\begin{equation}
	\mathcal{M}_K:= \{ f^G_x(y) : G \text{ is a probability measure supported on } K \}. 
	\label{eq:Mkset}
\end{equation}
The goal of this section is to prove an upper bound on the covering
number $N(\eta, \mathcal{M}_K, \| \cdot \|_{\infty, S_0 \times \RS})$
of $\mathcal{M}_K$ under the metric $\| \cdot \|_{\infty, S_0 \times
  \RS}$: 
\begin{equation}\label{metinfs0_defn}
	\sup_{x \in S_0, y \in \RS} \left|f_x^G(y) - f_x^{G'}(y)  \right|.
\end{equation}
for an arbitrary set $S_0$ of $x$-values. General definitions of
covering numbers are given at the beginning of Subsection
\ref{subsection_mainthem}. 

For each $x$, let $
\lipsz(x)$ be defined as 
\begin{equation}
	\lipsz(x) := \sup_{\beta_1, \beta_2 \in K : \beta_1 \neq \beta_2} \frac{\left|r(x, \beta_1) - r(x, \beta_2) \right|}{\|\beta_1 - \beta_2\|} 
	\label{eq:lipsz_defn_appendix}
\end{equation}
so that
\begin{equation*}
	|r(x,\beta_1) - r(x,\beta_2) | \leq \lipsz(x) \| \beta_1- \beta_2 \| \qt{for all $\beta_1, \beta_2 \in K$}.
\end{equation*}

\begin{theorem}\label{theorem:metric_entropy}
	Suppose that, for every $x$, the function $\beta \mapsto r(x, \beta)$ is a polynomial function of degree at most $\zeta$. Then there exists a constant $C_p$ depending only on $p$ such that for every $0 < \eta < e^{-1}\sigma^{-1}$,  we have
	\begin{equation}\label{eq:metric_entropy}
		\log N(\eta, \mathcal{M}_K, \| \cdot \|_{\infty,S_0\times \RS}) \leq C_\pb \zeta^\pb N \left(\frac{\sigma}{\lipsz_{S_0}} \sqrt{2\log \frac{3}{\sigma \eta}}, K \right) \left(\log \frac{1}{\sigma \eta} \right)^{\pb+1},
	\end{equation}
	where $\lipsz_{S_0} = \sup_{x\in S_0} \lipsz(x)$. In the right
        hand side above, $N(\delta, K)$ denotes the $\delta$-covering
        number of $K$ in the usual Euclidean metric. 
\end{theorem}

We prove Theorem \ref{theorem:metric_entropy} by modifying
appropriately the proof of the metric entropy results for Gaussian
location mixtures in \citet{zhang2009generalized} (see also
\citet{ghosal2007posterior} and
\citet{saha2020nonparametric}). Actually Theorem
\ref{theorem:metric_entropy} can be seen as a generalization of metric
entropy results for Gaussian location mixtures. Indeed, in the special
case when $p = 1$, $\sigma = 1$,   $S_0 = \{0\}$, $r(x, \beta) =
\beta$ and $K = [-M, M]$ (for some $M > 0$), the class $\mathcal{M}_K$
becomes 
\begin{equation*}
	\mathcal{H}_M := \left\{y \mapsto \int \phi(y - \beta) dG(\beta) : G[-M, M] = 1 \right\}
\end{equation*}
and inequality \eqref{eq:metric_entropy}  gives that the $\eta$-metric entropy of $\mathcal{H}_M$ under the $L_{\infty}$ metric on $\mathbb{R}$ is bounded by
\begin{equation*}
	C N\left(\sqrt{2 \log \frac{3}{\eta}}, [-M, M] \right) \left(\log \frac{1}{\eta} \right)^{2} \leq    C \left(1 + \frac{2M}{\sqrt{2 \log (3/\eta)}} \right) \left(\log \frac{1}{\eta} \right)^2
\end{equation*}
for all $0 < \eta < e^{-1}$. This is essentially \citet[inequality (5.8)]{zhang2009generalized}.

The proof of Theorem~\ref{theorem:metric_entropy} crucially relies on Lemma~\ref{lemma:mml} (moment matching accuracy) and Lemma~\ref{lemma:abdm} (approximation by discrete mixtures) which are given next. Lemma \ref{lemma:mml} follows almost directly from the corresponding result for Gaussian location mixtures (see \citet[Lemma 1]{jiang2009general} or \citet[Lemma D.2]{saha2020nonparametric}) but Lemma \ref{lemma:abdm} requires additional arguments. 

\begin{lemma}	\label{lemma:mml}
	Fix a pair $(x, y)$ and let $A$ be a subset of $\RS^p$ such that
	\[
	\mathring{O}((x,y), a) \subseteq A \subseteq O((x,y), ca)
	\]
	for some $a>1$ and $c \geq 1$ where
	\[
	O((x,y), a) = \{ \beta \in  K: | y - r(x,\beta) |/\sigma \leq  a \}.
	\]
	and
	\[
	\mathring{O}((x,y), a) = \{ \beta \in K: | y -r(x,\beta)| /\sigma<  a \}. 
	\]
	Let $G$ and $G'$ be two probability measures on $\RS^p$ such that for some $m \geq 1$ and all integers $0\leq k \leq 2m$, we have
	\begin{equation}
		\int_A \{r(x,\beta)\}^k \diff G(\beta) = \int_A \{r(x,\beta)\}^k \diff G'(\beta).
		\label{eq:mml_equal}
	\end{equation}
	Then 
	\begin{equation}
		|f^G_x(y) - f^{G'}_x(y)|\leq \frac{1}{2 \pi \sigma } \left( \frac{c^2a^2e}{2(m+1)}\right)^{m+1}   + \frac{2e^{-a^2/2}}{(2\pi)^{1/2}\sigma}.
		\label{eq:mml_result}
	\end{equation}
\end{lemma}
\begin{proof}[Proof of Lemma~\ref{lemma:mml}]
	This result follows from the moment matching lemma for the univariate Gaussian location mixtures in  \citet[Lemma 1]{jiang2009general} or \citet[Lemma D.2]{saha2020nonparametric}. These results are stated for the $\sigma = 1$ case but the extension to arbitrary $\sigma$ is straightforward. 
\end{proof}

\begin{lemma} 	\label{lemma:abdm}
	Let $G$ be a probability measure supported on $K$. For every $a \geq 1$, there exists a discrete probability measure $G'$ supported on at most
	\begin{equation}
		(2\lfloor 13.5 a^2\rfloor \zeta+1)^\pb N(a \sigma/\lipsz_{S_0}, K)+1,
		\label{eq:abdm_ellbound}
	\end{equation}
	points in $K$ such that 
	\begin{equation}
		\sup_{(x,y) \in S_0 \times \RS 
		} |f^G_x(y) - f^{G'}_x(y)| \leq \left(1+\frac{1}{\sqrt{2\pi}}\right)\frac{e^{-a^2/2}}{ (2\pi)^{1/2} \sigma}.
		\label{eq:abdm_result}
	\end{equation}
\end{lemma}
\begin{proof}[Proof of Lemma~\ref{lemma:abdm}]
	Let us introduce a pseudometric $\dbeta$ on $K$ as
	\begin{equation}
		\dbeta(\beta_1,\beta_2) = \sup_{x \in S_0} | r(x,\beta_1) - r(x,\beta_2)|/\sigma.
		\label{eq:amdm_metric}
	\end{equation}
	Fix $a \geq 1$ and let $L: = N(a,K,d)$ denote the $a$-covering number of $K$ under the pseudometric $\dbeta$.  
	Let $E_1,\dots,E_L$ denote balls of radius $a$ (with respect
        to $\dbeta$) within $K$ whose union is equal to $K$. We define
        $B_1 = E_1$ and $B_i = E_i \cap (\cup_{j=1}^{i-1} B_j)^c$ for
        $i = 2,\dots,L$. Let $m  = \lfloor 13.5 a^2\rfloor$ and let
        $T_{int}$ denote the collection of the following
        $(2m\zeta+1)^pL$-dimensional vectors:
        \begin{equation*}
\left( \int \beta_1^{k_1}\dots \beta_\pb^{k_\pb}\mathbb{I}\{ \beta \in B_i\} \diff G(\beta)\right)_{ 0 \leq k_1,\dots, k_\pb \leq2 m \zeta, 1\leq i \leq L}           
\end{equation*}
as $G$ ranges over the class of all probability  measures over $K$. By standard results, it follows that $T_{int}$ is the convex hull of 
	\[
	T: = \left\{\left( \beta_1^{k_1}\dots \beta_\pb^{k_\pb}\mathbb{I}\{ \beta \in B_i\} \right)_{ 0 \leq k_1,\dots,k_\pb \leq2 m \zeta, 1\leq i \leq L}: \beta \in  K\right\}.
	\]
	This follows, for example, from \citet[Theorem 6.3]{parthasarathy2005probability} and the fact that $T$ is closed. 
	Notice that both $T_{int}$ and $T$ lie in the Euclidean space
        of dimension $(2m\zeta +1)^\pb L$. By Carath\'eodory's
        theorem, any vector in $T_{int}$ can be written as a convex
        combination of at most $\{(2m\zeta +1)^\pb L+1\}$ elements in
        $T$. This implies that for every probability measure $G$ on
        $K$, there exists a discrete measure $G'$ which is supported
        on a discrete subset of $K$ of cardinality at most $\{(2m\zeta
        +1)^\pb L+1\}$ such that  
	\begin{equation}\label{mmat}
		\int_{B_i} \beta_1^{k_1}\dots \beta_\pb^{k_\pb} \diff
                G(\beta) = \int_{B_i} \beta_1^{k_1}\dots
                \beta_\pb^{k_\pb} \diff G'(\beta) 
              \end{equation}
 for all $0 \leq
                k_1,\dots,k_\pb \leq 2m\zeta$ and all $1 \leq i
                \leq L$.                
	Fix $x \in S_0$ and $y \in \RS$. We shall prove the bound \eqref{eq:abdm_result} for $|f_x^G(y) - f_{x}^{G'}(y)|$ by using Lemma \ref{lemma:mml}. First note that since $\mathring{O}((x,y),a)$ is contained in $K$, the sets $B_1,\dots,B_L$ cover $\mathring{O}((x,y),a)$. Let $F := \{  1\leq i \leq L:B_i \bigcap  \mathring{O}((x,y),a)\neq \emptyset\}$ so that 
	\[
	\mathring{O}((x,y),a) \subseteq \bigcup_{i\in F} B_i.
	\]
	We shall prove below that 
	\begin{equation}\label{temong}
		\bigcup_{i\in F} B_i \subseteq O((x,y), 3a), 
	\end{equation}
	which will enable us to apply Lemma \ref{lemma:mml} with $A = \bigcup_{i \in F} B_i$. To see \eqref{temong}, note that for each fixed $i \in F$, there exists $\beta_0 \in B_i$ such that $\beta_0 \in \mathring{O}((x,y),a)$, i.e., $|y - r(x, \beta_0)|/\sigma \leq a$. As the diameter of $B_i$ (under the metric $\dbeta$) is at most $2a$, it follows  that $\dbeta(\beta, \beta_0) \leq 2a$ for every $\beta \in B_i$. Consequently,
	\[
	| y - r(x,\beta)|/\sigma  \leq | y - r(x, \beta_0)|/\sigma + | r(x, \beta) - r(x, \beta_0) |/\sigma \\
	\leq a + \dbeta(\beta,\beta_0) \leq 3a.
	\]
	This proves \eqref{temong}. In order to apply Lemma \ref{lemma:mml}, we need to check that inequalty \eqref{eq:mml_equal} holds. This basically follows from \eqref{mmat} and the fact that $r(x, \beta)$ is assumed to be a polynomial function of the components of $\beta$ with degree $\zeta$ (this will ensure that the terms being integrated on both sides of \eqref{eq:mml_equal} are polynomials of components of $\beta$ with degree up to $2m\zeta$). Lemma \ref{lemma:mml} can thus be applied (with $A = \bigcup_{i \in F} B_i$ and $c = 3$), which gives 
	\[
	|f^G_x(y) - f^{G'}_x(y)|\leq \frac{1}{2 \pi \sigma} \left( \frac{9a^2e}{2(m+1)}\right)^{m+1}   + \frac{e^{-a^2/2}}{(2\pi)^{1/2}\sigma}.
	\]
	Because $m= \lfloor 13.5a^2\rfloor$, we have $m+1 \geq 13.5 a^2$ and
	\[ 
	\left( \frac{9 a^2 e}{2(m+1)}\right)^{m+1} \leq \left(\frac{e}{3}\right)^{m+1} \leq \exp( - \frac{m+1}{12}) \leq \exp\left(-\frac{27 a^2}{24}\right) \leq e^{-a^2/2},
	\]
	where we used the simple fact that $e/3\leq e^{-1/12}$. This proves \eqref{eq:abdm_result}. It remains to prove that the cardinality of the support of $G'$ is at most \eqref{eq:abdm_ellbound}. As we have already seen that the cardinality of the support of $G'$ is at most $\{(2m\zeta +1)^\pb L+1\}$, we only need to show that $L= N(a, K,d)$ is at most the Euclidean covering number $N(a \sigma/\lipsz_{S_0}, K)$. For this, note that by definition of $\lipsz_{S_0}$, we have
	\begin{equation*}
		\dbeta(\beta_1,\beta_2) = \sup_{x\in S_0} | r(x,\beta_1) - r(x,\beta_2) |/\sigma \leq \lipsz_{S_0} \sigma^{-1}\| \beta_1 - \beta_2\|,
	\end{equation*}
	for every $\beta_1, \beta_2$. This gives
	\begin{equation}\label{dbeuc} 
		N(a, K, \dbeta) \leq N(a \sigma/\lipsz_{S_0}, K), 
	\end{equation}
	which completes the proof of Lemma \ref{lemma:abdm}. 
\end{proof}

\begin{proof}[Proof of Theorem~\ref{theorem:metric_entropy}]
	Fix a probability measure $G$ that is supported on $K$. By Lemma~\ref{lemma:abdm}, for each fixed $a \geq 1$, there exists a probability measure $G'$ supported on $K$  such that
	\[
	\sup_{(x,y) \in S_0\times \RS} |f^G_x(y) - f^{G'}_x(y)| \leq \left(1+\frac{1}{\sqrt{2\pi}}\right)\frac{e^{-a^2/2}}{(2\pi)^{1/2}\sigma},
	\]
	and such that the cardinality of the support of $G'$ is at most $\ell$ where $\ell$ is given by \eqref{eq:abdm_ellbound}.
	
	Now let $\alpha = \nu = e^{-a^2/2}$. Let $s_1, \dots, s_{N_1}$ be an $\alpha$-covering of $K$ under the $\dbeta$ pseudometric (defined in \eqref{eq:amdm_metric}), where (via \eqref{dbeuc})
	\begin{equation}\label{tte1}
		N_1 := N(\alpha, K, \dbeta) \leq N(\alpha \sigma/\lipsz_{S_0}, K).
	\end{equation}
	Also let $t_1, \dots, t_{N_2}$ be a $\nu$-covering of the probability simplex $\Delta_{\ell} := \{(p_1, \dots, p_{\ell}) : p_j \geq 0, \sum_{j} p_j = 1\}$ under the $L^1$-metric $(p, q) \mapsto \sum_{j} |p_j - q_j|$ where $N_2 := N(\nu, \Delta_{\ell}, L_1)$.  We can write $G' = \sum_{i=1}^\ell w_i \delta_{a_i}$ for some $(w_1, \dots, w_{\ell}) \in \Delta_{\ell}$ and $a_1, \dots, a_{\ell} \in K$. Since  $s_1,\dots,s_{N_1}$ form an $\alpha$-covering of $K$, we  can find $\ell$ (not necessarily distinct) elements $s_{G'1},\dots, s_{G'\ell}$ from  $\{s_1,\dots,s_{N_1}\}$ such  that $\dbeta(a_i,s_{G'i})\leq \alpha, i = 1, \dots, \ell$. Letting $G''= \sum_{i=1}^\ell w_i \delta_{s_{G'i}}$, we have 
	\[
	\begin{split}
		|f_x^{G'}(y) - f_x^{G''}(y)|& = \frac{1}{\sigma} \left|\sum_{i=1}^\ell w_i \phi \left(\frac{y - r(x,a_i)}{\sigma} \right) - \sum_{i=1}^\ell w_i \phi \left(\frac{y -r(x, s_{G'i})}{\sigma} \right) \right|\\
		&\leq \frac{1}{\sigma} \sum_{i=1}^\ell w_i \cdot  \left| \phi \left(\frac{y - r(x,a_i)}{\sigma} \right) -  \phi\left(\frac{y -r(x, s_{G'i})}{\sigma} \right) \right| \\
		& \leq \frac{1}{\sigma} \sum_{i=1}^\ell w_i \cdot\sup_{z }|\phi'(z)|\cdot \dbeta (a_i, s_{G'i}) \leq \alpha \frac{e^{-1/2}}{(2\pi)^{1/2} \sigma}
	\end{split}
	\]
	for every $x \in S_0$ and $y \in \RS$. Also since $t_1,\dots, t_{N_2}$ is a $\nu$-covering of $\Delta_{\ell}$ under the $L^1$ metric, there exist $t_{G'1}, \dots, t_{G'\ell}$ from $\{t_1,\dots, t_{N_2} \}$ such that $\sum_{i=1}^\ell|t_{G'i} - w_i| \leq \nu$. Denote $G''' = \sum_{i=1}^\ell t_{G'i} \delta_{s_{G'i}} $, then for every $x\in S_0$ and any $y \in \RS$, we have
	\[
	\begin{split}
		|f_x^{G''}(y) - f_x^{G'''}(y)|& = \frac{1}{\sigma}\left|\sum_{i=1}^\ell w_i \phi\left(\frac{y - r(x,s_{G'i})}{\sigma}\right) - \sum_{i=1}^\ell t_{G'i} \phi \left( \frac{y -r(x, s_{G'i})}{\sigma} \right)\right|\\
		&\leq \frac{1}{\sigma} \sum_{i=1}^\ell |w_i - t_{G'i}|\cdot  \phi \left(\frac{y -r(x, s_{G'i})}{\sigma} \right) \leq  \frac{\nu}{\sigma} \cdot \sup_{z }|\phi(z)|   \leq \frac{\nu}{\sigma} \frac{1}{(2\pi)^{1/2}}.
	\end{split}
	\]
	Combining three inequalities together, we have
	\begin{equation}
		\begin{split}
			|f_x^{G}(y) - f_x^{G'''}(y)|&\leq 
			|f^G_x(y) - f^{G'}_x(y)| + |f_x^{G'}(y) - f_x^{G''}(y)| + |f_x^{G''}(y) - f_x^{G'''}(y)| \\
			&\leq \left(1+(2\pi)^{-1/2}\right)\frac{e^{-a^2/2}}{(2\pi)^{1/2}\sigma}+\alpha \frac{e^{-1/2}}{(2\pi)^{1/2}\sigma}+\nu \frac{1}{(2\pi)^{1/2}\sigma}
		\end{split}
		\label{eq:G_difference_3}
	\end{equation}
	for all $x\in S_0$ and $y \in \RS$. We now take $\alpha = \nu = \eta \sigma/3$ so that $a = \{2\log(\alpha^{-1})\}^{1/2} = \{2\log(3 \sigma^{-1}\eta^{-1})\}^{1/2}$. The right hand side of \eqref{eq:G_difference_3} is bounded by
	\[
	\alpha/\sigma\left( 2(2\pi)^{-1/2} + (2\pi)^{-1} + (2\pi)^{-1/2} e^{-1/2}\right)
	\leq \eta. 
	\]
	Therefore, as $G'''$ varies, the collection of functions $(x, y) \mapsto f_x^{G'''}(y)$ forms an $\eta$-covering of $\mathcal{M}_K$ under the metric $\| \cdot \|_{\infty, S_0 \times \RS}$. It remains to bound the cardinality of this collection which equals $\binom{N_1}{\ell} N_2$. Thus
	\begin{equation*}
		\log N(\eta, \mathcal{M}_K, \| \cdot\|_{\infty, S_0 \times \RS}) \leq \log \binom{N_1}{\ell} + \log N_2 .
	\end{equation*}
	By Stirling's formula, 
	\[\binom{N_1}{\ell} \leq \frac{N_1^\ell}{\ell!}\leq \left(\frac{N_1 e}{\ell}\right)^\ell.
	\] 
	By \eqref{tte1} and \eqref{eq:abdm_ellbound}, we have $N_1 \leq \ell$ so that $\log \binom{N_1}{\ell} \leq \ell$. Also $N_2$  is the $\nu$-covering number of $\Delta_{\ell}$ under the $L^1$-metric which implies, by a well known result, that $\log N_2 \leq \ell \log (1+2/v)$. We thus get  
	\begin{equation*}
		\log N(\eta, \mathcal{M}_K, \| \cdot\|_{\infty, S_0 \times \RS}) \leq \ell (\log(1+2/\nu) + 1).
	\end{equation*}
	By $1/\nu = 3 \sigma^{-1} \eta^{-1}$ and $\eta < e^{-1} \sigma^{-1}$, we get
	\begin{equation}\label{eq:entropy_ell}
		\log N(\eta, \mathcal{M}_K, \| \cdot\|_{\infty, S_0 \times \RS}) \leq \ell (\log(1+2/\nu) + 1) \leq C \ell \log (\sigma^{-1} \eta^{-1})
	\end{equation}
	for a universal constant $C$. It also follows from \eqref{eq:abdm_ellbound} that
	\begin{equation*}
		\ell = (2\lfloor 13.5 a^2\rfloor \zeta +1)^\pb N(a \sigma/\lipsz_{S_0},K) \leq C_\pb\{\log(\sigma^{-1}\eta^{-1})\}^\pb \zeta^\pb N \left(\frac{\sigma}{\lipsz_{S_0}} \sqrt{2 \log \frac{3}{\sigma \eta}},K \right). 
	\end{equation*}
	This, combined with \eqref{eq:entropy_ell}, completes the proof of Theorem \ref{theorem:metric_entropy}. 
      \end{proof}

\subsection{Proof of Theorem~\ref{thm:onedim_selfregu}}
\label{section:proof_selfregu}

We introduce the Jensen's formula, which is a classic result in complex analysis (see e.g. \citet[Chapter 5]{stein2010complex}). It is a useful tool for analyzing holomorphic functions and their zeros. The version we present here is adapted to our context. 

\begin{lemma}[Jensen's Formula]\label{lemma:jensen_formula}
	Let $\Omega$ be an open set that contains the closure of a disc
	$D_R$ and suppose that f is holomorphic in $\Omega$, $f(0)\neq 0$, and $f$ vanishes
	nowhere on the circle $C_R$. If $z_1,\dots,z_N$ denote the zeros of $f$ inside the
	disc (counted with multiplicities), then
	\[
	\log |f(0)| = \sum_{k=1}^N \log \left(  \frac{|z_k|}{R} \right) + \frac{1}{2 \pi} 
	\int_{0}^{2 \pi} \log |f (R e^{i\theta})| \diff \theta.
	\]
Further, if $\mathfrak{n}(r)$ denotes the number of zeros of $f$ in disk $D_r$, then
\begin{equation}\label{eq:jensen_corollary}
 \int_{0}^R \frac{\mathfrak{n}(r)}{r} \diff r 
 = \frac{1}{2 \pi} 
 \int_{0}^{2 \pi} \log |f (R e^{i\theta})| \diff \theta - \log |f(0)|.
\end{equation}
\end{lemma}

\begin{proof}[Proof of Theorem~\ref{thm:onedim_selfregu}]
	

We note that $x_i \neq 0$  for all $i$ because of the restriction $|x_i/x_j| \leq r_0$ for all $i$ and $j$. 
The minimum and maximum of $y_i/x_i, 1\leq i \leq n $ are
denoted by $\beta_{\min}$ and $\beta_{\max}$ respectively.

We claim a basic fact that any support point $\tilde{\beta}$ of NPMLE
$\hat{G}$ must lie in the interval $[\beta_{\min}, \beta_{\max}]$. The
validity of this fact can be shown by contradiction. If it is not
true, we can move the support point $\tilde{\beta}$ from outside of
the interval $[\beta_{\min}, \beta_{\max}]$ to $\beta_{\min}$ or
$\beta_{\max}$,  whichever that is closer to $\tilde{\beta}$. After
the move, the function   
$D(\hat{G}, \tilde{\beta})$ defined in \eqref{eq:D_G_beta} is strictly
increased, which is a contraction with the optimal condition
$D(\hat{G}, \beta) \leq 0$ for all $\beta$. Lastly, if $\beta_{\max} =
\beta_{\min}$, it means that the ratio of $y_i/x_i$ are equal to
$\beta_{\min} = \beta_{\max}$ for all $i = 1,\dots, n$. Thus the NPMLE
has only one support point at $\beta_{\min} = \beta_{\max}$ and the
conclusion of this theorem holds trivially. Therefore, we focus on the
non-degenerate case $\beta_{\max} > \beta_{\min}$ from now on.

Based on Proposition~\ref{prop:opt_condition}, all support points of
$\hat{G}$ (except those lying on the boundary of $K$) are critical
points of $\beta \mapsto D(\hat{G}, \beta)$. Therefore for $p= 1$, the
number of  support points in $\hat{G}$ is at most $2$ plus the number
of zeros of $D'(\hat{G}, \beta):= \frac{d }{d\beta} D(\hat{G},
\beta)$. For mathematical convenience, we define $g(\beta) :=
\frac{\diff D(\hat{G},\beta + \beta_{\min} -\Delta )}{\diff \beta}$
over $\beta \in [\Delta, \beta_{\max}- \beta_{\min}+\Delta]$, where
$\Delta$ is some positive number in $(\frac{1}{2}(\beta_{\max}-
\beta_{\min}), \beta_{\max}- \beta_{\min})$ such that $g(0) \neq 0$.
Such  a $\Delta$  is always feasible because the analytic function
$\frac{\diff D(\hat{G},\beta)}{\diff \beta}$ cannot be uniformly all
$0$ within the interval $[\beta_{\min} - (\beta_{\max}- \beta_{\min}),
\beta_{\min} - \frac{1}{2} (\beta_{\max}- \beta_{\min})]$. Note here
we shift the position of origin by $\beta_{\min}- \Delta$ when
defining $g(\beta)$, and the condition $g(0) \neq 0$ will be necessary
when we invoke \eqref{eq:jensen_corollary} later. Next, we will
expand  $g(\beta)$ to the complex domain and bound the number of zeros
of $g(\beta)$ over the disc of radius $\mathfrak{R} := \beta_{\max} -
\beta_{\min} +\Delta \in (2\Delta, 3 \Delta)$, which is naturally an
upper bound of the zeros over $[-\mathfrak{R}, \mathfrak{R}]$.  

By Jensen's formula \eqref{eq:jensen_corollary}, we have
\[
\int_{0}^{2\mathfrak{R}} \frac{\mathfrak{n}(r)}{r} \diff r =
\frac{1}{2\pi} \int_{0}^{2 \pi} \log |g (2\mathfrak{R} e^{i\theta})|
\diff \theta - \log |g(0)| \leq \sup_{\theta} \log |f (2\mathfrak{R}
e^{i\theta})| - \log |g(0)|. 
\]

On the other hand, by the monotonicity and non-negativity of
$\mathfrak{n}(r)$, we have 
\[
\int_{0}^{2\mathfrak{R}} \frac{\mathfrak{n}(r)}{r} \diff r  \geq
\int_{\mathfrak{R}}^{2\mathfrak{R}} \frac{\mathfrak{n}(r)}{r} \diff r
\geq \mathfrak{n}(\mathfrak{R}) \cdot
\int_{\mathfrak{R}}^{2\mathfrak{R}} \frac{1}{r} \diff r =
\mathfrak{n}(\mathfrak{R}) \cdot \log 2. 
\]

Therefore, we have
\[
\mathfrak{n}(\mathfrak{R}) \leq \frac{1}{\log 2} \left[\sup_{\theta}
  \log |g (2\mathfrak{R} e^{i\theta})| - \log |g(0)|\right] =
\frac{1}{\log 2}  \sup_{\theta} \log \left|
  \frac{g(2\mathfrak{R}e^{i\theta}) }{g(0) } \right|. 
\]

By definition, 
\[
g(\beta) = \frac{1}{n} \sum_{i=1}^n \frac{2x_i[x_i(\beta+\beta_{\min}-
  \Delta)-y_i] f^{\beta+\beta_{\min}-
    \Delta}_{x_i}(y_i)}{f^{\hat{G}}_{x_i}(y_i)}. 
\]

We note that the numerator $f^{\hat{G}}_{x_i}(y_i)$ in $g(\beta)$ does
not depend on $\beta$, therefore 
\[
\begin{split}
	\sup_{\theta} \log \left| \frac{g(2\mathfrak{R}e^{i\theta})
  }{g(0) } \right| &  \leq  
	 \sup_\theta \log \max_{1 \leq j \leq n}  \left| \frac{
                     2x_j[x_j(2\mathfrak{R} e^{i\theta}+\beta_{\min}-
                     \Delta)-y_j] f^{2\mathfrak{R}
                     e^{i\theta}+\beta_{\min}- \Delta}_{x_j}(y_j)}{
                     2x_j[x_j(\beta_{\min}- \Delta)-y_j]
                     f^{\beta_{\min}- \Delta}_{x_j}(y_j) }  \right|\\ 
	&  \leq \log \sup_{\theta} \max_{1 \leq j \leq n} \left|
          \frac{x_j(2 \mathfrak{R} e^{i\theta} + \beta_{\min} -
          \Delta)- y_j }{x_j(\beta_{\min} - \Delta) - y_j} \right| +
          \log \sup_{\theta} \max_{1\leq j \leq n}  \left|
          \frac{f^{2\mathfrak{R} e^{i\theta}+\beta_{\min}-
          \Delta}_{x_j}(y_j)}{f^{\beta_{\min}- \Delta}_{x_j}(y_j)}
          \right|. 
\end{split}	
\]

Because 
\begin{equation}
y_j/x_j\in [\beta_{\min}, \beta_{\max}], \beta_{\max}-\beta_{\min} \in
(\Delta, 2\Delta), \mathfrak{R} \in (2\Delta, 3\Delta), 
\label{eq:jensen_basicbounds}
\end{equation}
we have
\[
\begin{split}
	\log \sup_{\theta} \max_{1 \leq j \leq n} \left| \frac{x_j(2
  \mathfrak{R} e^{i\theta} + \beta_{\min} - \Delta)- y_j
  }{x_j(\beta_{\min} - \Delta) - y_j} \right| & \leq  \log
                                                \sup_{\theta} \max_{1
                                                \leq j \leq n} \left|
                                                \frac{ 2 \mathfrak{R}
                                                e^{i\theta} +
                                                \beta_{\min} - \Delta
                                                -
                                                y_j/x_j}{\beta_{\min}-
                                                \Delta -
                                                y_j/x_j}\right|\\ 
	& \leq \frac{2\cdot 3\Delta + 3\Delta}{\Delta} = \log 9. 
\end{split}
\]

Additionally, we can bound the second term as follows. 
\begin{align}
&\log \sup_{\theta} \max_{1\leq j \leq n}  \left|  \frac{f^{2\mathfrak{R} e^{i\theta}+\beta_{\min}- \Delta}_{x_j}(y_j)}{f^{\beta_{\min}- \Delta}_{x_j}(y_j)}  \right|\notag \\
 = & \log \sup_{\theta} \max_{1\leq j \leq n}  \left| \exp \left\{-\frac{1}{2\sigma^2} \left[   x_j(2 \mathfrak{R} e^{i\theta} + \beta_{\min} - \Delta)- y_j   \right]^2 + \frac{1}{2\sigma^2} \left[   x_j(\beta_{\min} - \Delta) - y_j    \right]^2  \right\} \right| \notag\\
  = & \log \sup_{\theta} \max_{1\leq j \leq n} \exp \left\{\frac{x_j^2}{2\sigma^2} \cdot 2 \mathfrak{R} e^{i \theta} \cdot [2\mathfrak{R} e^{i \theta} + 2(\beta_{\min} - \Delta - y_j/x_j) ]    \right\} \notag\\
  \leq & \log \max_{1\leq j \leq n}\exp \left\{\frac{x_j^2}{2\sigma^2} (2\cdot 3\Delta)\cdot(2\cdot3\Delta+3\Delta)  \right\} \label{eq:jensenbound_exp}\\
  = &  27 \frac{\Delta^2}{\sigma^2} \max_{1\leq j \leq n} x_j^2 \notag\\
  \leq & 27 \frac{1}{\sigma^2} (\beta_{\max} - \beta_{\min})^2 \max_{1\leq j \leq n} x_j^2 \label{eq:jensenbound_explast}
\end{align}
where \eqref{eq:jensenbound_exp} follows from  \eqref{eq:jensen_basicbounds}, and  \eqref{eq:jensenbound_explast} follows from the fact that $\Delta \leq \beta_{\max}- \beta_{\min}$.  Consider the definition of $\beta_{\max}$ and $\beta_{\min}$ as well as the restriction that $|x_i  /x_j| \leq r_0$, \eqref{eq:jensenbound_explast}  can be further bounded as
\[
\begin{split}
	27 \frac{1}{\sigma^2} (\beta_{\max} - \beta_{\min})^2 \max_{1\leq j \leq n} x_j^2  \leq &
	108 \frac{1}{\sigma^2} r_0^2 \max_{1\leq j \leq n} |y_j|^2.
\end{split}
\]

To summarize, we have shown that
\[
\mathfrak{n}(\mathfrak{R}) \leq C_0 + C_1 \frac{r_0^2}{\sigma^2} \max_{1\leq j \leq n} |y_j|^2,
\]
where $C_0 = \frac{\log 9}{\log 2}$ and $C_1 = \frac{108}{\log 2}$ are constants that do not depend on problem parameters.

Because $\Vert x_j\Vert \leq R$, $\Vert \beta_j\Vert \leq B$,
\[
\max_{1\leq j \leq n} |y_j| \leq BR + \max_{1\leq j \leq n} |z_j|,
\]
where the error term $z_j\sim N(0,\sigma^2)$. 


Furthermore, since $z_j\sim N(0,\sigma^2)$ i.i.d., it holds that $\max_{1\leq j \leq n} |z_j| \leq \sigma \sqrt{2\log n} + \sigma \sqrt{2\log (1/\delta) }$ with probability at least $1-\delta$, where the bound $\mathbb{E} [\max_{1\leq j \leq n} |z_j|]\leq \sigma \sqrt{2\log n}$ is a well established result on maxima of $n$ Gaussians, and the probabilistic statement follows from a Gaussian process tail bound. Let $\delta = n^{-\tau}$ for $\tau>1$, it follows that $\max_{1\leq j \leq n} |z_j| \leq  (\sqrt{2} + \sqrt{2\tau}) \sigma \sqrt{\log n}$. Plugging back to the bound of $\max_{1\leq j \leq n} |y_j|$ and $\mathfrak{n}(\mathfrak{R})$, we have
\begin{equation}
	\mathfrak{n}(\mathfrak{R}) \leq C_0  + C_1 r_0^2 (B^2R^2\sigma^{-2} + 8\tau \log n)
\end{equation}
with probability at least $1- n^{-\tau}$. That is, for $n$ such that $\log n > \max\{C_0, C_1 r_0^2 B^2R^2 \sigma^{-2}\}$, we have $\mathfrak{n}(\mathfrak{R}) \leq \tau r_0^2 O(\log n)$ with probability at least $1- n^{-\tau}$.
\end{proof}

\newpage
\section{Additional Numerical Results: Fitted Coefficients for
  Simulations in Subsections \ref{sec:sinusoid_simulation} and
  \ref{sec:changepoint_simulation} } \label{subsec:detailed_coef}

\begin{table}[!htb]
	\centering
	\begin{tabular}{c|ccccccc|c} 
		\toprule
		{Method} & \multicolumn{7}{c|}{$\beta$} & $\pi$ \\
		\hline
		\multirow{4}{*}{{True parameters}} 
		& -2.143 & 4.008 & -0.188 & 2.584 & 1.136 & 1.039 & 2.849 & 0.250 \\
		& 1.060 & 0.719 & -1.263 & -2.457 & -1.195 & 1.807 & -0.052 & 0.250 \\
		& -0.809 & 1.219 & 0.943 & 1.938 & 1.394 & 1.584 & -1.140 & 0.250 \\
		& -4.251 & 1.949 & 1.916 & -3.289 & 1.666 & 0.383 & 0.489 & 0.250 \\
		\hline
		\multirow{9}{*}{{NPMLE with $\hat{\sigma}$}} 
		& 0.875 & 0.809 & -0.934 & -2.099 & -0.963 & 1.707 & -0.624 & 0.154 \\
		& -0.817 & 1.449 & 1.004 & 2.107 & 1.831 & 1.528 & -1.028 & 0.148 \\
		& -3.896 & 2.466 & 2.224 & -3.479 & 0.970 & 0.323 & 0.543 & 0.146 \\
		& -1.983 & 4.046 & -1.087 & 2.385 & 1.021 & 0.991 & 3.130 & 0.120 \\
		& -3.800 & 0.118 & 3.241 & -1.425 & 0.870 & -0.386 & 1.271 & 0.099 \\
		& -1.883 & 3.265 & 0.218 & 2.546 & 1.835 & 1.538 & 2.736 & 0.096 \\
		& -0.436 & 0.704 & 2.570 & 1.140 & 0.444 & 2.159 & -1.106 & 0.094 \\
		& -0.066 & 1.253 & -4.019 & -2.822 & 0.876 & 1.069 & -0.355 & 0.085 \\
		& -1.843 & 5.972 & 0.372 & -0.421 & 1.818 & 1.815 & 0.416 & 0.057 \\
		\hline
		\multirow{11}{*}{{NPMLE with true $\sigma$}} 
		& -0.817 & 1.449 & 1.004 & 2.107 & 1.831 & 1.528 & -1.028 & 0.157 \\
		& 0.875 & 0.809 & -0.934 & -2.099 & -0.963 & 1.707 & -0.624 & 0.148 \\
		& -3.896 & 2.466 & 2.224 & -3.479 & 0.970 & 0.323 & 0.543 & 0.133 \\
		& -3.800 & 0.118 & 3.241 & -1.425 & 0.870 & -0.386 & 1.271 & 0.102 \\
		& -0.066 & 1.253 & -4.019 & -2.822 & 0.876 & 1.069 & -0.355 & 0.084 \\
		& -1.883 & 3.265 & 0.218 & 2.546 & 1.835 & 1.538 & 2.736 & 0.079 \\
		& -1.983 & 4.046 & -1.087 & 2.385 & 1.021 & 0.991 & 3.130 & 0.071 \\
		& 0.067 & 0.695 & 1.435 & 1.286 & 0.084 & 2.283 & -1.100 & 0.061 \\
		& -2.175 & 4.329 & 0.123 & 3.317 & 0.838 & 0.845 & 3.046 & 0.060 \\
		& -1.843 & 5.972 & 0.372 & -0.421 & 1.818 & 1.815 & 0.416 & 0.059 \\
		& -1.379 & 3.306 & 2.137 & -1.004 & -1.654 & 2.696 & -0.227 & 0.047 \\
		\bottomrule
	\end{tabular}
	\caption{True and fitted mixtures  coefficients  for the sinusoid example.}
	\label{table:sinusoid_results}
\end{table}

\begin{table}[!htb]
	\small
	\centering
	\begin{tabular}{c|ccccc|c} 
		\toprule
		{Method} & \multicolumn{5}{c|}{$\beta$} & $\pi$ \\
		\hline
		\multirow{4}{*}{{True parameters}} 
		& 1.376 & 0.118 & 0.002 & 0.638 & -1.553 & 0.250 \\
		& 0.073 & 1.466 & 0.414 & 0.240 & -2.588 & 0.250 \\
		& 3.048 & -0.379 & -0.345 & -1.099 & 1.837 & 0.250 \\
		& -0.737 & -3.120 & 5.936 & -0.898 & -0.392 & 0.250 \\
		\hline
		\multirow{7}{*}{{NPMLE with $\hat{\sigma}$}} 
		& 1.201 & 0.354 & 0.354 & 0.054 & 1.097 & 0.264 \\
		& 3.062 & -1.554 & 0.553 & -0.095 & -1.350 & 0.187 \\
		& -0.446 & -3.501 & 6.099 & -0.844 & -0.844 & 0.129 \\
		& 0.208 & -4.133 & 5.905 & -0.696 & -0.696 & 0.121 \\
		& -0.769 & 3.557 & -0.731 & -0.719 & -0.719 & 0.103 \\
		& -0.541 & 2.078 & 0.265 & 0.265 & -2.349 & 0.100 \\
		& 2.844 & -0.200 & -0.614 & -0.644 & -1.863 & 0.096 \\
		\hline
		\multirow{7}{*}{{NPMLE with true $\sigma$}} 
		& 1.184 & 0.314 & 0.314 & 0.314 & -2.543 & 0.263 \\
		& 3.062 & -1.554 & 0.553 & -0.095 & -1.350 & 0.177 \\
		& -0.613 & -3.437 & 6.342 & -1.068 & -0.560 & 0.161 \\
		& 3.047 & -0.240 & -1.322 & -0.343 & -0.343 & 0.123 \\
		& 0.025 & 2.457 & -0.512 & -0.874 & 1.997 & 0.121 \\
		& -0.361 & -3.425 & 5.721 & 0.286 & 0.762 & 0.096 \\
		& -0.992 & 2.472 & 0.901 & -1.282 & 1.952 & 0.059 \\
		\bottomrule
	\end{tabular}
	\caption{True and fitted mixtures coefficients for the change-point example.}
	\label{table:change-point_results}
\end{table}

\end{appendices}

\end{document}